\documentclass{article}

\usepackage{microtype}
\usepackage{graphicx}

\usepackage{booktabs} % for professional tables

\usepackage{hyperref}

\usepackage[accepted]{icml2024}

\usepackage{amsmath}
\usepackage{amssymb}
\usepackage{mathtools}
\usepackage{amsthm}
\usepackage{enumitem}
\usepackage{caption}

\usepackage[subpreambles=true]{standalone}
\usepackage{xstring}
\usepackage{subcaption}

\usepackage{wrapfig}

\usepackage[capitalize,noabbrev]{cleveref}

\theoremstyle{plain}
\newtheorem{theorem}{Theorem}
\newtheorem{proposition}{Proposition}
\newtheorem{lemma}{Lemma}
\newtheorem{corollary}{Corollary}[theorem]
\theoremstyle{definition}

\theoremstyle{remark}

\DeclareMathOperator*{\argmin}{arg\,min}

\usepackage[textsize=tiny]{todonotes}

\usepackage{standalone}  % for tikz and pgfplots
\usepackage{tikz}
\usetikzlibrary{arrows.meta}

\icmltitlerunning{State-Constrained Zero-Sum Differential Games with One-Sided Information}

\begin{document}

\twocolumn[
\icmltitle{State-Constrained Zero-Sum Differential Games with One-Sided Information}



\icmlsetsymbol{equal}{*}

\begin{icmlauthorlist}
\icmlauthor{Mukesh Ghimire}{mae}
\icmlauthor{Lei Zhang}{mae}
\icmlauthor{Zhe Xu}{mae}
\icmlauthor{Yi Ren}{mae}
% \icmlauthor{}{sch}
\end{icmlauthorlist}

\icmlaffiliation{mae}{Department of Mechanical and Aerospace Engineering, Arizona State University, Tempe, AZ, USA}

\icmlcorrespondingauthor{Yi Ren}{yiren@asu.edu}

\icmlkeywords{Machine Learning, ICML}

\vskip 0.3in
]



\printAffiliationsAndNotice{}  % leave blank if no need to mention equal contribution

\begin{abstract} % 4-6 sentences
We study zero-sum differential games with state constraints and one-sided information, where the informed player (Player 1) has a categorical payoff type unknown to the uninformed player (Player 2). The goal of Player 1 is to minimize his payoff without violating the constraints, while that of Player 2 is to violate the state constraints if possible, or to maximize the payoff otherwise. One example of the game is a man-to-man matchup in football.  
% The values of such games are governed by Hamilton-Jacobi (HJ) equations that suffer from the curse of dimensionality (CoD). Even worse, the values are convex with respect to the common belief of players due to information asymmetry, making existing reinforcement learning algorithms based on Bellman backup infeasible.
Without state constraints, Cardaliaguet~\yrcite{cardaliaguet2007differential} showed that the value of such a game exists and is convex to the common belief of players.
Our theoretical contribution is an extension of this result to games with state constraints and the derivation of the primal and dual subdynamic principles necessary for computing behavioral strategies.   
% We alleviate CoD through model-based deep reinforcement learning and ensure value convexity via input convex neural network. 
% Lastly, 
Different from existing works that are concerned about the scalability of no-regret learning in games with discrete dynamics, our study reveals the underlying structure of strategies for belief manipulation resulting from information asymmetry and state constraints. This structure will be necessary for scalable learning on games with continuous actions and long time windows. 
% Since backward induction requires computing the convex hull of the value with respect to belief, standard reinforcement learning cannot be applied. Instead, we resort to a dynamic programming algorithm with $\epsilon$-optimality and introduce an input convex neural network to alleviate the curse of dimensionality while ensuring value convexity. 
We use a simplified football game to demonstrate the utility of this work, where we reveal player positions and belief states in which the attacker should (or should not) play specific random deceptive moves to take advantage of information asymmetry, and compute how the defender should respond.    
% We use a two-vehicle uncontrolled intersection case to show that the stochastic behavioral strategy of an adversarial Player 1 tricks a naive Player 2 into believing that it is benign, and successfully launches an attack once it moves into Player 2's zero-sum backward reachability set. We then show that the behavioral strategy of Player 2 mitigates this risk.  
\end{abstract}

%===============================================================================
\vspace{-0.05in}
\section{Introduction}
\vspace{-0.05in}
We study fixed-time zero-sum differential games with state constraints and one-sided information, where Player 1 holds a private type (e.g., an intent or preference) that defines the payoffs of the game.
The goal of Player 1 (resp. Player 2) is to minimize (resp. maximize) the cost. Since violation of the state constraint results in $+\infty$ penalty to Player 1, Player 2 resorts to violating the constraints when possible; and Player 1 resigns when state violation is inevitable. 
At the beginning of the game, Nature draws a type from a distribution known to both players and assigns the type only to Player 1. Initialized as Nature's distribution, the common belief about Player 1's type is updated dynamically during the game based on observations, and shared between the players. A stochastic state trajectory is produced based on the initial state and belief, the deterministic system dynamics, and the behavioral strategies of the two players.   
The value of the game, when exists, follows a Hamilton-Jacobi (HJ) equation and is a function of time, state, and belief.  
Importantly, Player 1 may control the release of information at the equilibrium to manipulate the common belief and take advantage of information asymmetry.

% In response, Player 2 may also need a mixed strategy to mitigate risks. In theory, the value and the equilibrium strategies can be computed via dynamic programming, yet the computation encounters the curse of dimensionality~\cite{***}.  

% To give an example, consider an uncontrolled intersection case in Fig.~\ref{***}a where the goal of Player 2 (an autonomous vehicle) is to pass the intersection as soon as possible without collision. Player 1 (a human vehicle) may either be benign, in which case it shares the same form of loss as Player 2, or adversarial, in which case its goal is to collide with Player 2. Conventionally, Player 2 computes a worst-case backward reachability set by assuming Player 1 to be always adversarial and avoids Player 1 from entering this set, see Fig.~\ref{***}b. \cite{***} improves this overly conservative strategy by allowing Player 2 to update a binary belief about Player 1's type during the interaction: Once Player 2 believes that Player 1 is benign and is not looking for a collision, it shrinks its backward reachability set (by solving a general-sum rather than zero-sum game). However, this adaptive strategy can be exploited by a best response from an intelligent Player 1, as shown in Fig.~\ref{***}c. In contrary, the PBE of Player 2 ensures its safety without being overly conservative, see Fig.~\ref{***}d. 

% In this paper, we prove the existence of PBE for two-player general-sum differential games with one-sided information and address the curse of dimensionality (CoD) issue in computing such PBEs.

We use Hexner's game~\cite{hexner1979differential} as a minimal example to demonstrate information control by Player 1: Consider two players with linear dynamics
\vspace{-0.05in}
\begin{equation*}
    \dot{x}_i = A_i x_i + B_i u_i,
\vspace{-0.05in}
\end{equation*}
for $i=1,2$, where $x_i(t) \in \mathbb{R}^{d_x}$ are system states, $u_i(t) \in \mathcal{U}$ are control inputs belonging to the admissible set $\mathcal{U}$, 
% $A_i \in \mathbb{R}^{d_x \times d_x}$ and $B_i \in \mathbb{R}^{d_x \times d_u}$. 
$A_i, B_i \in \mathbb{R}^{d_x \times d_x}$. Let $\theta \in \{-1, 1\}$ be Player 1's type unknown to Player 2\footnote{Hexner's analysis is applicable to $\theta \in \mathbb{R}^{d_x}$, but is not generalizable to games with arbitrary dynamics and payoff functions. Here we adopt Cardaliaguet's setting where types are categorical~\cite{cardaliaguet2007differential}.}. 
Let $p_{\theta}$ be Nature's probability distribution of $\theta$. Consider that the game is to be played infinite many times, the payoff is an expectation over $\theta$:
\vspace{-0.05in}
\begin{equation} \label{eq:value_def}
\small
\begin{aligned}
    J(u_1, u_2) &= \mathbb{E}_{\theta} \Bigl[\int_0^T \left(\|u_1\|^2_{R_1} - \|u_2\|^2_{R_2}\Bigr) dt \right.+ \\ 
    & \quad \Bigl. \|x_1(T) - z\theta \|^2_{K_1(T)} - \|x_2(T) - z\theta \|^2_{K_2(T)}\Bigr],
\end{aligned}
\vspace{-0.05in}
\end{equation}
where $z \in \mathbb{R}^{d_x}$, $R_1$ and $R_2$ are positive-definite, continuous matrix functions, and $K_1(T)$ and $K_2(T)$ are positive semi-definite matrices. All parameters are common knowledge except $\theta$. Essentially Player 1's goal is to get closer to $z\theta$ than Player 2. Since Player 2 can infer the target based on Player 1's control, Player 1 may play a non-revealing strategy for some time, i.e., as if he also only knows $p_{\theta}$ rather than the actual $\theta$, before he eventually reveals.

The equilibrium of this game is as follows: 
First, it can be shown that players' control has a 1D representation, denoted by $\tilde{\theta}_i \in \mathbb{R}$, through:
\begin{equation*}\label{eq:action}
    u_i = - R_i^{-1}B_i^TK_ix_i + R_i^{-1}B_i^TK_i\Phi_i z \tilde{\theta}_i,
\end{equation*}
for $i = 1,2$, where $\dot{\Phi}_i = A_i\Phi_i$ with boundary condition $\Phi_i(T) = I$, and 
\begin{equation*}
    \dot{K}_i = -A_i^TK_i - K_iA_i + K_i^TB_iR_i^{-1}B_i^TK_i.
    \label{eq:K}
\end{equation*}
Then by introducing 
\begin{equation}
    d_i = z^T \Phi^T_i K_i B_i R_i^{-1} B_i^T K_i^T \Phi_i z, 
    \label{eq:d}
\end{equation}
and defining the critical time as
\begin{equation*}
    t_r = \argmin_{t} \int_{0}^{t} (d_1(s) - d_2(s))ds,
\end{equation*}
one can derive Player 1's strategy as $\tilde{\theta}_1(t) = 0$ for $t \in [0, t_r]$ and $\tilde{\theta}_1(t) = \theta$ for $t \in (t_r, T]$, i.e., Player 1 reveals its type at $t_r$. Player 2's strategy turns out to be to strictly follow Player 1: $\tilde{\theta}_2(t) = \tilde{\theta}_1(t)$. 
The original analysis by Hexner exploits the fact that both players solve linear-quadratic regulators parameterized by $\theta$. We will revisit this game after introducing the differential game theory for one-sided information games~\cite{cardaliaguet2007differential,cardaliaguet2009numerical}, which arrives at Hexner's solution but can also solve games with arbitrary dynamics and payoff functions, subject to continuity assumptions. 
% It will turn out that belief manipulation is a way for Player 1 to convexify the value for a non-revealing version of the game, and thus the belief-manipulating strategy is always no worse than the non-revealing one. 
% While existing theoretical developments can be generalized to games with arbitrary dynamics and losses (subject to continuity assumptions)~\cite{***,***},
This paper extends the unconstrained settings in Cardaliaguet~\yrcite{cardaliaguet2009numerical} and Souquiere~\yrcite{souquiere2010approximation}: We prove that value exists for differential games \textit{with state constraints} and one-sided information, and derive the primal-dual HJ equations necessary for computing player strategies. 

Different from existing works that focus on scalable no-regret learning on imperfect-information games with discrete dynamics~\cite{brown2020combining, perolat2022mastering}, this paper builds on top of repeated games and incomplete-information differential games~\cite{cardaliaguet2007differential,cardaliaguet2009numerical} to reveal the underlying mechanism of belief manipulation resulted from information asymmetry and state constraints. Specifically, we show that in any subgame, Player 1 plays a behavioral strategy (i.e., probability distributions over the action space for all his types) that convexifies his value with respect to the common belief. As a result, the common belief ``splits" to vertices of the value convex hull with probabilities that are optimal for Player 1. See Fig.~\ref{fig:value_curvature} for an illustration using Hexner's game. Importantly, the number of splits for Player 1 is no more than the number of possible player types. On the other hand, Player 2 counters Player 1 by playing a dual game where her behavioral strategy is determined by the convexification of the conjugate value. See Sec.~\ref{sec:problem} and \ref{sec:numerical} for details.

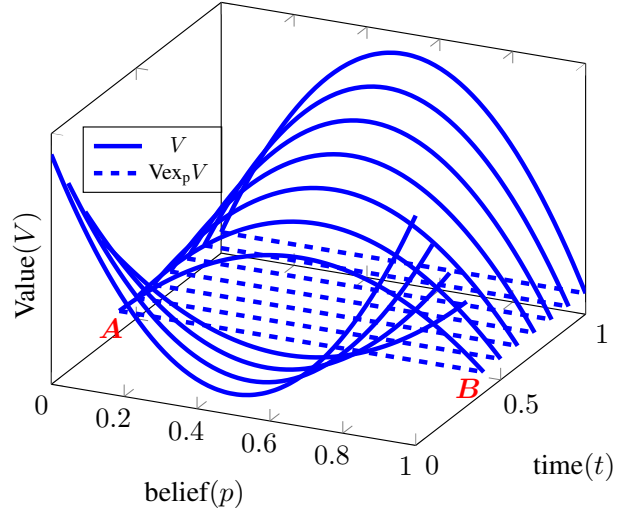
\begin{figure}[!t]
    \centering
    \begin{tikzpicture}[
    declare function={
        normal(\m,\s)= \s * \m^2;
    }
]
\begin{axis}[
    legend style={nodes={scale=0.78, transform shape}, fill=none, style={at={(0.32, 0.72)}}},
    legend entries={$V$, $\text{Vex}_{\text{p}} V$},
    samples=30,
    domain=0:1,
    xlabel={belief$(p)$},
    ylabel={time$(t)$},
    zlabel={Value$(V)$},
    samples y=0, ytick={0, 0.5, 1}, ztick=\empty,
    area plot/.style={
        style=ultra thick,
        draw=blue,
        fill=none,
        mark=none,
        smooth
    }
]
\addlegendimage{blue, ultra thick}
\addlegendimage{blue, ultra thick, dashed}
\addplot3 [area plot, domain=0:1, samples=30] (x,1,{1 * (-(x-0.5)^2) + 0.25*1});
\addplot3 [area plot, domain=0:1, samples=30] (x,0.9,{0.9 * (-(x-0.5)^2) + 0.25*0.9});
\addplot3 [area plot, domain=0:1, samples=30] (x,0.8,{0.8 * (-(x-0.5)^2) + 0.25*0.8});
\addplot3 [area plot, domain=0:1, samples=30] (x,0.7,{0.7 * (-(x-0.5)^2) + 0.25*0.7});
\addplot3 [area plot, domain=0:1, samples=30] (x,0.6,{0.6 * (-(x-0.5)^2) + 0.25*0.6});
\addplot3 [area plot, domain=0:1, samples=30] (x,0.5,{0.5 * (-(x-0.5)^2) + 0.25*0.5});
\addplot3 [area plot, domain=0:1, samples=30] (x,0.4,{0.4 * (-(x-0.5)^2) + 0.25*0.4});

\addplot3 [area plot, dashed, domain=0:1, samples=30] (x,1,{0*x});
\addplot3 [area plot,dashed, domain=0:1, samples=30] (x,0.9,{0*x});
\addplot3 [area plot,dashed, domain=0:1, samples=30] (x,0.8,{0*x});
\addplot3 [area plot, dashed,domain=0:1, samples=30] (x,0.7,{0*x});
\addplot3 [area plot, dashed,domain=0:1, samples=30] (x,0.6,{0*x});
\addplot3 [area plot, dashed,domain=0:1, samples=30] (x,0.5,{0*x});
\addplot3 [area plot, dashed,domain=0:1, samples=30] (x,0.4,{0*x});

\addplot3 [area plot, domain=0:1, samples=30] (x,0.3,{(1 - 2*0.3)*(x-0.5)^2)});
\addplot3 [area plot, domain=0:1, samples=30] (x,0.2,{(1 - 2*0.2) * (x-0.5)^2)});
\addplot3 [area plot, domain=0:1, samples=30] (x,0.1,{(1 - 2*0.1) *(x-0.5)^2)});
\addplot3 [area plot, domain=0:1, samples=30] (x,0,{*(x-0.5)^2)});

\addplot3 [mark=none, color=red](0, 0.35, 0.01) node[below] {$\boldsymbol{A}$};
\addplot3 [mark=none, color=red](0.98, 0.35, 0.01) node[below] {$\boldsymbol{B}$};
\end{axis}
\end{tikzpicture}
    \caption{Value along belief $(p)$ and time $(t)$ in Hexner's game. Belief splits to $A$ ($p=0$) and $B$ ($p=1$) depending on the true type of Player 1, when the value becomes concave should Player 1 play a non-revealing strategy. In other words, Player 1 delays the release of his type until a critical time. In more general cases, belief splitting may not fully reveal Player 1's type, leading to belief manipulation.}
    \label{fig:value_curvature}
\end{figure}

Within this context, it becomes clear that understanding whether and how belief should be manipulated relies on knowing the value landscape over the belief space at any time and state. In addition to the curse of dimensionality (CoD) commonly present for games with non-trivial state/action/belief spaces and time horizons, we also experience computational challenges due to value discontinuity and the need for convexification and splitting. We discuss in Sec.~\ref{sec:value} a set of solutions, including using physics-informed neural network to characterize the backward reachable set to smooth value approximation, and using an input convex architecture~\cite{amos2017input} for predicting convex values. 

To summarize, we claim the following contributions:
\begin{itemize}
\vspace{-0.05in}
    \item We extend the theory of zero-sum differential games with one-sided information to games with state constraints by proving value existence of such games and deriving the primal and dual subdynamic principles;
    \vspace{-0.05in}
    \item We elucidate, with detailed examples, how the subdynamic principles lead to the derivation of behavioral strategies;
    \vspace{-0.05in}
    \item We develop numerical tools to alleviate CoD in value approximation and to infer behavioral strategies from values. In Sec.~\ref{sec:cases}, we solve an 8D man-to-man matchup game and reveal player positions in which the attacker can take advantage of information asymmetry by playing specific deceptive moves, and to derive the defender's best response in the lack of information. See Fig.~\ref{fig:fig_1}.
\end{itemize}

% \begin{figure}[!ht]
%     \centering
%     \includegraphics[width=\linewidth]{icml2024/figures/fig_2.png}
%     \caption{Trajectories of informed Player 1 (red) and uninformed Player 2 (blue) in an 8D Hexner's game w/ and w/o a state constraint or information asymmetry. Color shades indicate probabilities. When constrained, Player 1 stays away from Player 2 while trying to be closer to the target (the circle) than Player 2. Diamonds indicate initial states and stars indicate final states. See Sec.~\ref{sec:cases} for details.}
%     \label{fig:fig_1}
% \end{figure}

% figure 1 %%%%%%
\begin{figure}[!ht]
    \centering
    \begin{minipage}{0.48\linewidth}
    \centering
    \begingroup%
    \StrCount{figures/traj_comp_uncons}{/}[\matches]%
    \StrBefore[\matches]{figures/traj_comp_uncons}{/}[\datapath]%
    \pgfplotstableread[col sep=comma,]{\datapath/complete_info_uncons.csv}\datatablea
\begin{tikzpicture}[]
\begin{axis}[title={Const.(\ding{55}), Info.(\ding{51})},
    xmin=-0.52,
    xmax=0.1,
    ymin=-1.1,  % -1.1
    ymax=0.5, 
    point meta=explicit,
    ylabel shift=-15pt,
    xlabel shift=-5pt,
    % legend entries={Attacker, Defender},
    % legend image post style={thick},
    xlabel={$d_x$},
    ylabel={$d_y$},
            colormap name = p1,
    % colorbar,
    point meta min=0, 
    point meta max=1,
    % colorbar horizontal,
    %     colorbar style={
    %   at={(0.75,-0.1)},
    %   anchor=center,
    %   width=4.5cm
    % },
  %   colorbar/draw/.append code={
  % \begin{axis}[
  %   % colormap/p2/.style={
  %   % colormap name=p2,
  %   % }, 
  %   colormap name=p2,
  %   colorbar horizontal,
  %   point meta min=0,
  %   point meta max=1,
  %   every colorbar,
  %     at={(0.25,-0.1)},
  %     anchor=center,
  %     width=4.5cm,
  %   colorbar shift,
  %   colorbar=false,
  %   ]
  %       \pgfkeysvalueof{/pgfplots/colorbar addplot}
  %     \end{axis}
  %   }]
]
\addplot[color=red, mark=diamond, mark size=3pt] (-0.5, 0);
\addplot[color=blue, mark=diamond, mark size=3pt] (-0.4, -0.08);
\addplot [mark=o, mark size=0 pt, red, ultra thick] table [x={x1}, y={y1}]
\datatablea;
\addplot [mark=o, mark size=0 pt, blue, ultra thick] table [x={x2}, y={y2}]
\datatablea;
\addplot[color=magenta, mark=*, mark options={fill=magenta}, only marks, mark size=3pt](0, -1);
\addplot[color=magenta, mark=*, mark options={fill=white}, only marks, mark size=3pt](0, 1);

\pgfplotstablegetelem{10}{x1}\of{\datatablea}
\pgfmathsetmacro{\a}{\pgfplotsretval}
\pgfplotstablegetelem{10}{y1}\of{\datatablea}
\pgfmathsetmacro{\b}{\pgfplotsretval}
\addplot [red, mark=star, only marks, mark size=2pt, thick] coordinates {(\a, \b)};

\pgfplotstablegetelem{10}{x2}\of{\datatablea}
\pgfmathsetmacro{\a}{\pgfplotsretval}
\pgfplotstablegetelem{10}{y2}\of{\datatablea}
\pgfmathsetmacro{\b}{\pgfplotsretval}
\addplot [blue, mark=star, only marks, mark size=2pt, thick] coordinates {(\a, \b)};
\end{axis}
\end{tikzpicture}%
    \endgroup%

    \end{minipage}%
    \begin{minipage}{.48\linewidth}
    \centering
    \begingroup%
    \StrCount{figures/traj_comp_cons_new}{/}[\matches]%
    \StrBefore[\matches]{figures/traj_comp_cons_new}{/}[\datapath]%
    \pgfplotstableread[col sep=comma,]{\datapath/new_results_2/cons_comp_-1.csv}\datatableb
% \pgfplotstableread[col sep=comma,]{\datapath/cons_complete_2/cons_complete_2.csv}\datatablec
% \pgfplotstableread[col sep=comma,]{\datapath/cons_complete_2/cons_complete_3.csv}\datatabled
% \pgfplotstableread[col sep=comma,]{\datapath/cons_complete_2/cons_complete_4.csv}\datatablee
% \pgfplotstableread[col sep=comma,]{\datapath/cons_complete_2/cons_complete_5.csv}\datatablef
% \pgfplotstableread[col sep=comma,]{\datapath/cons_complete_2/cons_complete_6.csv}\datatableg
% \pgfplotstableread[col sep=comma,]{\datapath/cons_complete_2/cons_complete_7.csv}\datatableh
% \pgfplotstableread[col sep=comma,]{\datapath/cons_complete_2/cons_complete_8.csv}\datatablei
% \pgfplotstableread[col sep=comma,]{\datapath/cons_complete_2/cons_complete_9.csv}\datatablej

\begin{tikzpicture}
\begin{axis}[title={Const.(\ding{51}), Info.(\ding{51})},
    xmin=-0.52,
    xmax=0.1,
    ymin=-1.1,  % -1.1
    ymax=0.5, 
    point meta=explicit,
    ylabel shift=-15pt,
    xlabel shift=-5pt,
    xlabel={$d_x$},
    ylabel={$d_y$},
        % colormap name = p1,
  %   colorbar,
    point meta min=0, 
    point meta max=1,
  %   colorbar horizontal,
  %       colorbar style={
  %     at={(0.75,-0.1)},
  %     anchor=center,
  %     width=4.5cm
  %   },
  %   colorbar/draw/.append code={
  % \begin{axis}[
  %   % colormap/p2/.style={
  %   % colormap name=p2,
  %   % }, 
  %   colormap name=p2,
  %   colorbar horizontal,
  %   point meta min=0,
  %   point meta max=1,
  %   every colorbar,
  %     at={(0.25,-0.1)},
  %     anchor=center,
  %     width=4.5cm,
  %   colorbar shift,
  %   colorbar=false,
  %   ]
  %       \pgfkeysvalueof{/pgfplots/colorbar addplot}
  %     \end{axis}
  %   }]
]

% r-1 
\addplot [mark=o, mark size=0pt, only marks,
            mark options={draw=red},
            ultra thick, 
                  mesh, % Use line segments instead of one unbroken line
        colormap={}{ % Define the colormap
            color(0cm) = (yellow);
            % color(0.3cm)=(gray);
            color(1cm)=(red);
        },]  table [x={x1}, y={y1}, meta={c}]
\datatableb;
\addplot [mark=o, mark size=0pt, only marks,
            mark options={draw=blue},
            ultra thick, 
                  mesh, % Use line segments instead of one unbroken line
        colormap={}{
         rgb255(0cm)=(238,140,238);
         % rgb255(0.5cm)=(150, 10, 255);
                color(1cm)=(blue);
        },
        ,]  table [x={x2}, y={y2}, meta={c}]
\datatableb;
\pgfplotstablegetelem{10}{x1}\of{\datatableb}
\pgfmathsetmacro{\a}{\pgfplotsretval}
\pgfplotstablegetelem{10}{y1}\of{\datatableb}
\pgfmathsetmacro{\b}{\pgfplotsretval}
\addplot [red, mark=star, only marks, mark size=2pt, thick] coordinates {(\a, \b)};

\pgfplotstablegetelem{10}{x2}\of{\datatableb}
\pgfmathsetmacro{\a}{\pgfplotsretval}
\pgfplotstablegetelem{10}{y2}\of{\datatableb}
\pgfmathsetmacro{\b}{\pgfplotsretval}
\addplot [blue, mark=star, only marks, mark size=2pt, thick] coordinates {(\a, \b)};

\addplot[color=magenta, mark=*, mark options={fill=magenta}, only marks, mark size=3pt](0, -1);
\addplot[color=magenta, mark=*, mark options={fill=white}, only marks, mark size=3pt](0, 1);

% \addplot [red, thick, dashed] table [x={x1}, y={y1}]
% \datatableunc;

% \addplot [blue, thick, dashed] table [x={x2}, y={y2}]
% \datatableunc;
%% add more plots here

\end{axis}
\begin{axis}[xtick=\empty, ytick=\empty, minor tick num=0, scaled ticks=false, xticklabel=\empty,yticklabel=\empty,
    xmin=-0.52,
    xmax=0.1,
    ymin=-1.1,  % -1.1
    ymax=0.5, ]
    % legend entries = {Attacker, Defender}, 
    % legend image post style={thick}]
\addplot[color=red, mark=diamond, mark size=3pt] (-0.5, 0);
\addplot[color=blue, mark=diamond, mark size=3pt] (-0.4, -0.08);
% \addplot [mark=o, mark size=0.5pt, blue, ultra thick] table [x={x2}, y={y2}]
% {\datatablea};
\end{axis}
% \begin{axis}[title={Incomplete Information, Unconstrained},
%     xmin=-0.55,
%     xmax=0.1,
%     ymin=-1.1,
%     ymax=0.5,
%     point meta=explicit,
%     ylabel shift=-15pt,
%     xlabel shift=-5pt,
%     legend entries={Attacker, Defender},
%     xlabel={$x_1$},
%     ylabel={$y_1$},
% ]
% \addplot [mark=o, mark size=0.5pt, red, ultra thick] table [x={x1}, y={y1}]
% \datatablea;
% \addplot [mark=o, mark size=0.5pt, blue, ultra thick] table [x={x2}, y={y2}]
% \datatablea;

% \addplot[color=red, mark=diamond*, mark size=3pt] (-0.5, 0);
% \addplot[color=blue, mark=diamond*, mark size=3pt] (-0.4, -0.08);
% \addplot[color=magenta, mark=*, mark options={fill=magenta}, only marks, mark size=3pt](0, -1);
% \end{axis}
\end{tikzpicture}%
    \endgroup%

    \end{minipage}
    %% belief plots
    \begin{minipage}{.48\linewidth}
    \centering
    \begingroup%
    \StrCount{figures/traj_incomp_uncons}{/}[\matches]%
    \StrBefore[\matches]{figures/traj_incomp_uncons}{/}[\datapath]%
    \pgfplotstableread[col sep=comma,]{\datapath/uncons_incomp_2/incomplete_info_uncons_0.csv}\datatablea
\pgfplotstableread[col sep=comma,]{\datapath/uncons_incomp_2/incomplete_info_uncons_1.csv}\datatableb
\pgfplotstableread[col sep=comma,]{\datapath/uncons_incomp_2/incomplete_info_uncons_2.csv}\datatablec
\pgfplotstableread[col sep=comma,]{\datapath/uncons_incomp_2/incomplete_info_uncons_3.csv}\datatabled
\pgfplotstableread[col sep=comma,]{\datapath/uncons_incomp_2/incomplete_info_uncons_4.csv}\datatablee
\pgfplotstableread[col sep=comma,]{\datapath/uncons_incomp_2/incomplete_info_uncons_5.csv}\datatablef
\pgfplotstableread[col sep=comma,]{\datapath/uncons_incomp_2/incomplete_info_uncons_6.csv}\datatableg
\pgfplotstableread[col sep=comma,]{\datapath/uncons_incomp_2/incomplete_info_uncons_7.csv}\datatableh
\pgfplotstableread[col sep=comma,]{\datapath/uncons_incomp_2/incomplete_info_uncons_8.csv}\datatablei
\pgfplotstableread[col sep=comma,]{\datapath/uncons_incomp_2/incomplete_info_uncons_9.csv}\datatablej

\begin{tikzpicture}
\begin{axis}[title={Const.(\ding{55}), Info.(\ding{55})},
    xmin=-0.52,
    xmax=0.1,
    ymin=-1.1,  % -1.1
    ymax=0.5, 
    point meta=explicit,
    ylabel shift=-15pt,
    xlabel shift=-5pt,
    xlabel={$d_x$},
    ylabel={$d_y$},
  %       colormap name = p1,
  %   colorbar,
    % point meta min=0, 
    % point meta max=1,
  %   colorbar horizontal,
  %       colorbar style={
  %     at={(0.75,-0.1)},
  %     anchor=center,
  %     width=4.5cm
  %   },
  %   colorbar/draw/.append code={
  % \begin{axis}[
  %   % colormap/p2/.style={
  %   % colormap name=p2,
  %   % }, 
  %   colormap name=p2,
  %   colorbar horizontal,
  %   point meta min=0,
  %   point meta max=1,
  %   every colorbar,
  %     at={(0.25,-0.1)},
  %     anchor=center,
  %     width=4.5cm,
  %   colorbar shift,
  %   colorbar=false,
  %   ]
  %       \pgfkeysvalueof{/pgfplots/colorbar addplot}
  %     \end{axis}
  %   }]
]

% r-1 
% \addplot [mark=o, mark size=0pt, only marks,
%             mark options={draw=red},
%             ultra thick, 
%                   mesh, % Use line segments instead of one unbroken line
%         colormap={}{ % Define the colormap
%             color(0cm) = (yellow);
%             % color(0.3cm)=(gray);
%             color(1cm)=(red);
%         },]  table [x={x1}, y={y1}, meta={c}]
% \datatableb;
% \addplot [mark=o, mark size=0pt, only marks,
%             mark options={draw=blue},
%             ultra thick, 
%                   mesh, % Use line segments instead of one unbroken line
%         colormap={}{
%          rgb255(0cm)=(238,140,238);
%          % rgb255(0.5cm)=(150, 10, 255);
%                 color(1cm)=(blue);
%         },
%         ,]  table [x={x2}, y={y2}, meta={c}]
% \datatableb;
% \pgfplotstablegetelem{10}{x1}\of{\datatableb}
% \pgfmathsetmacro{\a}{\pgfplotsretval}
% \pgfplotstablegetelem{10}{y1}\of{\datatableb}
% \pgfmathsetmacro{\b}{\pgfplotsretval}
% \addplot [red, mark=star, only marks, mark size=2pt, thick] coordinates {(\a, \b)};

% \pgfplotstablegetelem{10}{x2}\of{\datatableb}
% \pgfmathsetmacro{\a}{\pgfplotsretval}
% \pgfplotstablegetelem{10}{y2}\of{\datatableb}
% \pgfmathsetmacro{\b}{\pgfplotsretval}
% \addplot [blue, mark=star, only marks, mark size=2pt, thick] coordinates {(\a, \b)};

% r-2
\addplot [mark=o, mark size=0pt, only marks,
            mark options={draw=red},
            ultra thick, 
                  mesh, % Use line segments instead of one unbroken line
        colormap={}{ % Define the colormap
            color(0cm) = (red);
            % color(0.3cm)=(gray);
            color(1cm)=(red);
        },]  table [x={x1}, y={y1}, meta={c}]
\datatablec;
\addplot [mark=o, mark size=0pt, only marks,
            mark options={draw=blue},
            ultra thick, 
                  mesh, % Use line segments instead of one unbroken line
        colormap={}{
         color(0cm)=(blue);
         % rgb255(0.5cm)=(150, 10, 255);
                color(1cm)=(blue);
        },
        ,]  table [x={x2}, y={y2}, meta={c}]
\datatablec;
\pgfplotstablegetelem{10}{x1}\of{\datatablec}
\pgfmathsetmacro{\a}{\pgfplotsretval}
\pgfplotstablegetelem{10}{y1}\of{\datatablec}
\pgfmathsetmacro{\b}{\pgfplotsretval}
\addplot [red, mark=star, only marks, mark size=2pt, thick] coordinates {(\a, \b)};

\pgfplotstablegetelem{10}{x2}\of{\datatablec}
\pgfmathsetmacro{\a}{\pgfplotsretval}
\pgfplotstablegetelem{10}{y2}\of{\datatablec}
\pgfmathsetmacro{\b}{\pgfplotsretval}
\addplot [blue, mark=star, only marks, mark size=2pt, thick] coordinates {(\a, \b)};

\addplot[color=magenta, mark=*, mark options={fill=magenta}, only marks, mark size=3pt](0, -1);

% \addplot [red, thick, dashed] table [x={x1}, y={y1}]
% \datatableunc;

% \addplot [blue, thick, dashed] table [x={x2}, y={y2}]
% \datatableunc;

\end{axis}
\begin{axis}[xtick=\empty, ytick=\empty, minor tick num=0, scaled ticks=false, xticklabel=\empty,yticklabel=\empty,
    xmin=-0.51,
    xmax=0.1,
    ymin=-1.1,  % -1.1
    ymax=0.5, ]
    % legend entries = {Attacker, Defender}, 
    % legend image post style={thick}]
\addplot[color=red, mark=diamond, mark size=3pt] (-0.5, 0);
\addplot[color=blue, mark=diamond, mark size=3pt] (-0.4, -0.08);
% \addplot [mark=o, mark size=0.5pt, blue, ultra thick] table [x={x2}, y={y2}]
% {\datatablea};
\end{axis}
% \begin{axis}[title={Incomplete Information, Unconstrained},
%     xmin=-0.55,
%     xmax=0.1,
%     ymin=-1.1,
%     ymax=0.5,
%     point meta=explicit,
%     ylabel shift=-15pt,
%     xlabel shift=-5pt,
%     legend entries={Attacker, Defender},
%     xlabel={$x_1$},
%     ylabel={$y_1$},
% ]
% \addplot [mark=o, mark size=0.5pt, red, ultra thick] table [x={x1}, y={y1}]
% \datatablea;
% \addplot [mark=o, mark size=0.5pt, blue, ultra thick] table [x={x2}, y={y2}]
% \datatablea;

% \addplot[color=red, mark=diamond*, mark size=3pt] (-0.5, 0);
% \addplot[color=blue, mark=diamond*, mark size=3pt] (-0.4, -0.08);
% \addplot[color=magenta, mark=*, mark options={fill=magenta}, only marks, mark size=3pt](0, -1);
% \end{axis}
\end{tikzpicture}%
    \endgroup%

    \end{minipage}%
    \begin{minipage}{.48\linewidth}
    \centering
    \begingroup%
    \StrCount{figures/traj_incomp_cons_p1_-1_new_2}{/}[\matches]%
    \StrBefore[\matches]{figures/traj_incomp_cons_p1_-1_new_2}{/}[\datapath]%
    \includestandalone[width=\linewidth]{figures/traj_incomp_cons_p1_-1_new_2}%
    \endgroup%

    \end{minipage}
        \begin{minipage}{.48\linewidth}
        \hspace{0.4in}
    \centering
    \begingroup%
    \StrCount{figures/only_colorbar_1}{/}[\matches]%
    \StrBefore[\matches]{figures/only_colorbar_1}{/}[\datapath]%
    \includestandalone[]{figures/only_colorbar_1}%
    \endgroup%

    \end{minipage}%
    \begin{minipage}{.48\linewidth}
    \hspace{0.1in}
    \centering
    \begingroup%
    \StrCount{figures/only_colorbar_2}{/}[\matches]%
    \StrBefore[\matches]{figures/only_colorbar_2}{/}[\datapath]%
    \includestandalone[]{figures/only_colorbar_2}%
    \endgroup%

    % \caption{}
    \end{minipage}
    \caption{Trajectories of informed Player 1 (red) and uninformed Player 2 (blue) in an 8D Hexner's game w/ and w/o a state constraint or information asymmetry. Color shades indicate probabilities. When constrained, Player 1 stays away from Player 2 while trying to be closer to the target (the circle) than Player 2. Diamonds indicate initial states and stars indicate final states. See Sec.~\ref{sec:cases} for details.}
    \label{fig:fig_1}
    \vspace{-0.1in}
\end{figure}
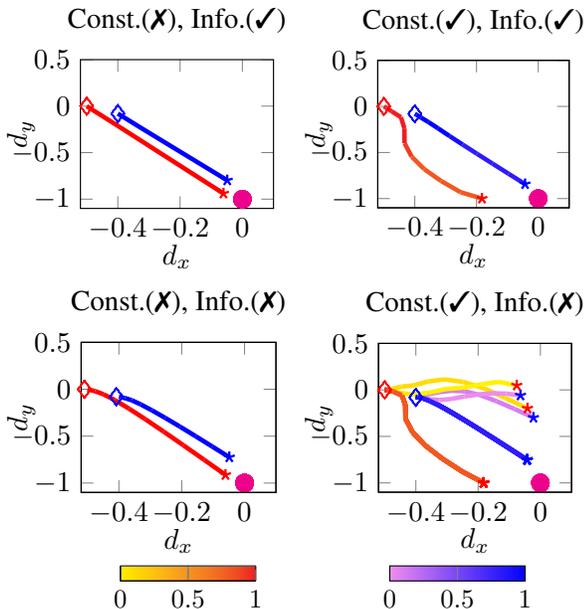

% Our study is built on top of \cite{***}, which proved the existence of PBEs for \textit{zero-sum} differential games with one-sided information and derived the primal and dual Hamilton-Jacobi (HJ) equations that respectively govern the values of Player 1 and Player 2 (see Sec.~\ref{***} for details). We extend this result to show that PBEs exist for general-sum games when Hamiltonian is linear with respect to players' control inputs and derive the corresponding HJ equations for both players' equilibrium values. Similar to zero-sum games, we show that the values are convex with respect to the belief of Player 2 due to information asymmetry. We then address the curse of dimensionality (CoD) issue common in solving HJ equations using Monte Carlo methods including deep reinforcement learning~\cite{***} and physics-informed neural net~\cite{***}, by adopting a neural architecture that ensures intrinsic convexity~\cite{***}.   

%%%%%%%%%%%%%%%%%%%%%%%%%%%%%%%%%%%%%%%%%%%%
\vspace{-0.1in}
\section{Related Work}
\vspace{-0.05in}
\paragraph{Incomplete-information repeated and differential games} Harsanyi~\yrcite{harsanyi1967games} first formalized information asymmetry in a stage game by introducing a private player types. 
Aumann et al.~\yrcite{aumann1995repeated} provided a framework to study repeated games with incomplete information on one side. De Meyer~\yrcite{de1996repeated} introduced dual games from where strategies of the uninformed player can be derived from a recursive structure of the conjugate value.
Extending these results to differential games with Markov rewards, Cardaliaguet~\yrcite{cardaliaguet2007differential} and Souquiere~\yrcite{souquiere2010approximation} confirmed the structures of incomplete-information games with one sided information on player type: (1) the game enjoys a primal-dual decomposition so that the informed player does not need to know the equilibrium of the uninformed player to compute his own; and (2) the value is convexified by belief splitting at the equilibrium. 
% the informed player may delay the release of information (or manipulation of the common belief) during a differential game. 
% Differential games cover a wide array of problems in engineering \cite{pontryagin1966theory}, and economics \cite{dockner2000differential}, and therefore are of deep interest. 
% While restricted to a class of linear-quadratic problems, Hexner solved a two-player one-sided zero-sum differential game by decoupling the game into two optimal control problems \cite{hexner1979differential}. However, the work lacks characterization of the value function for the incomplete-information game, and cannot be generalized to games with arbitrary payoff structure. Nonetheless, it serves as a case study to apply Cardaliaguet's framework. Minimal case studies exist that utilize the theory of the incomplete-information differential game. Some of the existing examples include -- differential games with random initial position \cite{cardaliaguet2007differential}, differential games with state-independent dynamics, and running payoffs \cite{souquiere2010approximation}. Other prominent, but similar works include \cite{jimenez2016differential, wu2017existence, wu2021differential, cardaliaguet2014pure}. 
Recently, Hu et al.~\yrcite{deceptiongame} proposed independently a belief-space HJ formulation for zero-sum differential games with one-sided information. While their framework can incorporate state constraints, it does not reveal the above structure of the equilibrium strategies of such games.

\vspace{-0.05in}
\paragraph{Imperfect-information dynamic games} Since player types can be considered as static private states, our work belongs to the category of imperfect-information dynamic games, where more general dynamics and information structures (e.g., disturbances, partial observability, and delayed information sharing) are considered. Nayyar et al.~\yrcite{nayyar2013common} showed that the game can be reformulated as perfect-information by introducing a common belief state, provided that the belief is strategy independent.
% provided that the common belief is independent from players' behavioral strategies. However, this assumption is inconsistent with the martingale nature of common beliefs in the narrower game setting we focused on (i.e., zero-sum, deterministic dynamics, full observability, and atomic probability measure on player types), where the splitting of belief is explicitly governed by the strategy of the informed player(s). 
% In fact, the backward induction algorithm proposed in \cite{nayyar2013common} leads to non-revealing strategies that does not take advantage of information asymmetry. 
This strategy-independence assumption is relaxed in Kartik and Nayyar~\yrcite{kartik2021upper} for zero-sum dynamic games by introducing past strategies as part of the players' information state. The general setting of Kartik and Nayyar~\yrcite{kartik2021upper}, however, does not facilitate a value existence proof. 
% Under a similar setting (private player initial states) but for deterministic dynamics, Cardaliaguet~\yrcite{cardaliaguet2007differential,cardaliaguet2014pure} proved the existence of value and showed that controlled information revealing of the informed player(s) leads to value convexity with respect to the common beliefs, similar to when player types are private. 
A significant amount of recent work build on top of common belief to approximate values of imperfect-information dynamic games (e.g., ReBeL~\cite{brown2020combining}, DeepNash~\cite{perolat2022mastering}, and SoG~\cite{schmid2023student}). Following Nayyar et al.~\yrcite{nayyar2013common}, these algorithms model behavioral strategies as \textit{prescriptions}, i.e., belief-conditioned action distributions. In addition, by taking advantage of the equivalence between local regret matching and Nash equilibrium in two-player zero-sum games~\cite{zinkevich2007regret}, no-regret algorithms~\cite{brown2018superhuman, brown2019superhuman, brown2020combining} have been developed for more scalably solving games with large action spaces and long time horizons than linear programming based methods~\cite{koller1995generating}. It should be noted that these algorithms scale linearly to the square-root of the action space, and thus induce high costs as the action space grows. 
While often disconnected, the studies on imperfect-information dynamic (or extensive-form) games and those on incomplete-information differential games are consistent in theory. Specifically, regret matching, i.e., solving subgame minimax problems with respect to behavioral strategies in the former, leads to strategies that satisfy the subdynamic programming principles stated in Cardaliaguet~\yrcite{cardaliaguet2007differential}, due to the fact that the behavioral strategies intrinsically convexify values. The key difference, however, is that regret matching algorithms do not enforce belief splitting. In practice, this means that the resultant strategy, often as a result of manually chosen action discretization, does not explicitly \textit{explain} whether a certain random action is to be taken in a given belief state in order to delay information release or to manipulate the belief in a specific way.
% \vspace{-0.05in}
% \paragraph{Positioning of this paper.} Our setting and results extend from \cite{cardaliaguet2007differential} to differential games with \textit{state constraints}, which serves as an initial step towards games with temporal logic specifications (e.g., competitive sports, autonomous driving, and human-robot interactions).   

%===============================================================================

\vspace{-0.05in}
\section{State Constrained Zero-Sum Differential Games with One-Sided Information}
\label{sec:problem}
\vspace{-0.05in}
% \paragraph{Zero-sum differential game with one-sided information}

\paragraph{Preliminaries}
We consider a time-invariant deterministic dynamical system that defines the dynamics of the combined state $x$ of Players 1 and 2, whose control inputs are $u$ and $v$, respectively:
\vspace{-0.05in}
\begin{equation}
\left \{ \begin{array}{ll}
     \dot{x}(t) = f(x(t), u(t), v(t)), &  u(t) \in \mathcal{U}, ~v(t) \in \mathcal{V} \\
     x(t_0) = x_0 & \\
\end{array}
\right.   
\label{eq:dynamics}
\vspace{-0.05in}
\end{equation}
The game starts at $t_0 \in [0, T]$ with an initial state $x_0 \in \mathbb{R}^{d_x}$.
Denote $g_i: \mathbb{R}^{d_x} \rightarrow \mathbb{R}$ the terminal payoff functions for $i \in [I]$, each corresponding to a Player 1 type drawn from Nature's distribution $p = \{p_1, ..., p_I\} \in \Delta(I)$, where $\Delta(I)$ is an $I$-dimensional simplex;
% \footnote{While Hexner's game allows continuous types, here we follow Cardaliageut's formulation with discrete types.}
denote $\mathcal{C} = \{x \in \mathbb{R}^{d_x} | c(x) \leq 0\}$ the set of feasible states. 
The goal of Player 1 is to minimize the expected payoff while keeping the state in $\mathcal{C}$. Player 1 receives $+\infty$ if state violation occurs; the goal of Player 2 is to maximize the expected payoff and hence may resort to violating the state constraint. We omit instantaneous payoffs (e.g., effort losses due to control) for conciseness, and discuss in Sec.~\ref{sec:numerical} modifications to the Bellman backup when common-knowledge instantaneous payoffs exist.  

The following assumptions will be used:
\begin{enumerate}
\vspace{-0.05in}
    \item $\mathcal{U}$ and $\mathcal{V}$ are compact and finite-dimensional sets;
    \vspace{-0.05in}
    \item $f: \mathbb{R}^{d_x} \times \mathcal{U} \times \mathcal{V} \rightarrow \mathbb{R}^{d_x}$ is bounded, continuous, and uniformly Lipschitz continuous with respect to $x$;
    \vspace{-0.05in}
    \item $g_i: \mathbb{R}^{d_x} \rightarrow \mathbb{R}$ for $i=1,...,I$ and $c: \mathbb{R}^{d_x} \rightarrow \mathbb{R}$ are Lipschitz continuous and bounded.
    \vspace{-0.05in}
    \item Isaacs' condition holds for the Hamiltonian $H: \mathbb{R}^{d_x} \times \mathbb{R}^{d_x} \rightarrow \mathbb{R}:$
        \begin{equation}
        \begin{aligned}
        H(x, \xi) &:= \min_{u \in \mathcal{U}} \max_{v \in \mathcal{V}} f(x, u, v)^T \xi \\
        &= \max_{v \in \mathcal{V}}\min_{u \in \mathcal{U}}f(x, u, v)^T \xi.
        \label{eq:hamiltonian}
        \end{aligned}
    \end{equation}
    \item Control inputs and states of both players are fully observable by all. The dynamics, the payoff set, and the equilibrium strategies are common knowledge to all. 
\end{enumerate}

\paragraph{Behavioral strategy} Let $\mathcal{A}(t)$ (resp. $\mathcal{D}(t)$) be the set of open-loop controls for Player 1 (resp. Player 2):
\vspace{-0.05in}
\begin{equation*}
\begin{aligned}
     \mathcal{A}(t) := \{\alpha: [t, T] \rightarrow \mathcal{U} ~|~ \text{Lebesgue measurable}\},\\
     \mathcal{D}(t) := \{\delta: [t, T] \rightarrow \mathcal{V} ~|~ \text{Lebesgue measurable}\}.
\end{aligned}
\vspace{-0.05in}
\end{equation*}
Following \cite{cardaliaguet2007differential}, 
% we consider two Stackelberg games where one plays after the other player at each time. A unique Nash equilibrium exists when both games have the same value. We will later narrow down our study to games where this is the case by introducing the Isaacs' condition (see Sec.~\ref{***}). To model the Stackelberg games, 
we introduce $\mathcal{H}(t)$ (resp. $\mathcal{Z}(t)$) as the set of non-anticipative pure strategies with delay for Player 1 (resp. Player 2)~\cite{elliott1972existence}:
\vspace{-0.05in}
\begin{equation*}
\begin{aligned}
        \mathcal{H}(t) := \{&\eta: \mathcal{D}(t) \rightarrow \mathcal{A}(t) ~|~ \exists \tau > 0 \text{ such that }\\
        &\forall s \in (t, T-\tau) \text{ and } \delta, \bar{\delta} \in \mathcal{D}(t), \text{ if } \delta = \bar{\delta} \text{ a.e.} \\
        &\text{ in }[t,s], \text{then } \eta(\delta) = \eta(\bar{\delta}) \text{ a.e. in } [t,s + \tau]\}.\\
        \mathcal{Z}(t) := \{&\zeta: \mathcal{A}(t) \rightarrow \mathcal{D}(t) ~|~ \exists \tau > 0 \text{ such that } \\
        & \forall s \in (t, T-\tau) \text{ and } \alpha, \bar{\alpha} \in \mathcal{A}(t), \text{ if } \alpha = \bar{\alpha}\text{ a.e.} \\
        & \text{ in } [t,s], \text{then } \zeta(\alpha) = \zeta(\bar{\alpha})\text{ a.e. in } [t,s + \tau]\}
\end{aligned}
\vspace{-0.05in}
\end{equation*}
A behavioral (mixed) strategy for Player 1 is defined by a pair $((\Omega_{\eta}, \mathcal{F}_{\eta}, \textbf{P}_{\eta}), \eta)$, where $(\Omega_{\eta}, \mathcal{F}_{\eta}, \textbf{P}_{\eta})$ is a probability space 
% (where $\Omega_{\eta}\in [0,1]^k$,  $\mathcal{F}_{\eta}$ is a subset of the Borel $\sigma$-algebra on $[0, 1]^k$, and $\textbf{P}_{\eta}$ is the Lebesgue measure)
and the strategy $\eta : \Omega_{\eta}\times \mathcal{D}(t) \rightarrow \mathcal{A}(t)$ is measurable and non-anticipative with delay, i.e., there is some $\tau > 0$ such that, for any $s \in (t, T - \tau)$ and $\delta$, $\bar{\delta} \in \mathcal{D}(t)$, if $\delta = \bar{\delta}$ a.e. in $[t, s]$ then $\eta(\omega, \delta) = \eta(\omega, \bar{\delta})$ a.e. in $[t, s + \tau]$ for any $\omega \in \Omega_{\eta}$. 
We denote the sets of behavioral strategies of Player 1 by $(\mathcal{H}_r(t))^I$ and the behavioral strategy of Player 2 by $\mathcal{Z}_r(t)$. With mild notational abuse, we will denote by $(\eta_i) \in (\mathcal{H}_r(t))^I$ a particular set of behavioral strategies of Player 1, and by $\zeta \in \mathcal{Z}_r(t)$ a particular behavioral strategy of Player 2. Lastly, we assume that $\eta_i$ for $i=1,...,I$ are defined on the same probability space.

\textit{Remarks.} Nonanticipative strategies \textit{with delay} are used, as opposed to ones without delay that are often used in complete-information games~\cite{elliott1972existence}, in order to enable Lemma~\ref{lemma:strategy} that associates random strategies with open-loop controls. This association will become useful in proving the existence of value of incomplete-information differential games and in value characterization (see discussions in \cite{cardaliaguet2007differential}):
\begin{lemma}
\label{lemma:strategy}
    (Lemma 2.2 of \cite{cardaliaguet2007differential}) For any pair $(\eta, \zeta) \in \mathcal{H}_r(t) \times \mathcal{Z}_r(t)$ and any $\omega := (\omega_1, \omega_2) \in \Omega_{\eta} \times \Omega_{\zeta}$, there is a unique pair $(\alpha_{\omega}, \delta_{\omega}) \in \mathcal{A}(t) \times \mathcal{D}(t)$ such that
    \vspace{-0.05in}
    \begin{equation}
        \eta(\omega_1, \delta_{\omega}) = \alpha_{\omega} \text{  and  } \zeta(\omega_2, \alpha_{\omega}) = \delta_{\omega}.
        \label{eq:random2det}
        \vspace{-0.05in}
    \end{equation}
    Furthermore the map $\omega \rightarrow (\alpha_{\omega}, \delta_{\omega})$ is measurable from $\Omega_{\eta}\times \Omega_{\zeta}$ endowed with $\mathcal{F}_{\eta} \otimes \mathcal{F}_{\zeta}$ into $\mathcal{A}(t) \times \mathcal{D}(t)$ endowed with the Borel $\sigma$-field associated with the $L^1$ distance.
\end{lemma}

% It is proved in \cite{***} that the game has a value when the following Isaacs' condition holds (cf. Theorem 2.1 in \cite{***}):
% \begin{equation}
% \begin{aligned}
%     \inf_{(\alpha_i) \in (\mathcal{A}_r(t_0))^I} \sup_{\beta \in \mathcal{B}_r(t_0)} \sum_{i=1}^I p_i \textbf{E}_{\alpha_i, \beta}\left(g_i \left(X_T^{t_0, x_0, \alpha_i, \beta}\right)\right) = \\
%     \sup_{\beta \in \mathcal{B}_r(t_0)} \inf_{(\alpha_i) \in (\mathcal{A}_r(t_0))^I}
%     \sum_{i=1}^I p_i \textbf{E}_{\alpha_i, \beta}\left(g_i \left(X_T^{t_0, x_0, \alpha_i, \beta}\right)\right)
%  \end{aligned}   
% \end{equation}
% where $\alpha_i \in \mathcal{A}_r(t_0)$ (for $i=1,...,I$) is a random strategy for Player 1 of type $i$, $\beta \in \mathcal{B}_r(t_0)$ is a random strategy for Player 2, $X_t^{t_0, x_0, \alpha_i, \beta}$ is a random system state at $t$ along the trajectory starting from $t_0$ and $x_0$ and driven by $\alpha_i$ and $\beta$. $\textbf{E}_{\alpha_i, \beta}(\cdot)$ is the expectation over random strategies $\alpha_i$ and $\beta$. $\sum_{i=1}^I p_i (\cdot)$ is an expectation over the random type.

% We will also show that in general the equilibrial strategies of an incomplete-information game are mixed. For Player 1, this randomness allows manipulation of the common belief about his type, in order to achieve the optimal expected value; for Player 2, the randomness allows hedging against risks due to the lack of knowledge about Player 1's type.
\vspace{-0.05in}
\paragraph{Backward reachable set} 
% Let $\mathcal{T}_i = \{x \in \mathbb{R}^{d_x} ~|~ g_i(x) < 0\}$ for $i = 1,...,I$ be the target sets and $\mathcal{C} = \{x \in \mathbb{R}^{d_x} ~|~ c(x) > 0\}$ the constraint set. 
% We assume that $c: \mathbb{R}^{d_x} \rightarrow \mathbb{R}$ and $g_i: \mathbb{R}^{d_x} \rightarrow \mathbb{R}$ for $i=1,...,I$ are Lipschitz continuous and bounded. 
% The goal of Player 1 of the $i$th type is to 
% minimize the minimum payoff $g_i(\cdot)$ within $[0, T]$ while staying in $\mathcal{C}$ before the minimum payoff is reached. 
% minimize the expected terminal payoff 
% \begin{equation}
%     G_i
% \end{equation}

% The goal of Player 2 is to either maximize the minimum $g_i(\cdot)$ or forcing Player 1 out of $\mathcal{C}$, in which case the payoff is considered $+\infty$. 
Let $X_{\tau}^{t_0,x_0,\alpha,\delta}$ be the solution of Eq.~\eqref{eq:dynamics} at $t = \tau$ when starting at $(t_0, x_0)$ and following $(\alpha,\delta)$.
With behavioral strategies $(\eta, \zeta)$ and initials $(t_0, x_0)$,
we denote by 
$\mathcal{X}^{t_0,x_0,\alpha, \delta}_{\tau}$ the trajectory of states reachable by $(\alpha, \delta)$ within $[t_0, \tau]$, and $\mathcal{X}^{t_0,x_0,\eta, \zeta}_{\tau}$ as states reachable by $(\eta, \zeta)$ within $[t_0, \tau]$:
\vspace{-0.05in}
\begin{equation*}
% \begin{aligned}
%         \mathcal{X}^{t_0,x_0,\eta, \zeta}_{\tau} := \{&y \in \mathbb{R}^{d_x} ~|~ \exists s \in [t_0, \tau],~\omega \in \Omega_{\eta} \times \Omega_{\zeta} \text{ and } (\alpha_{\omega}, \delta_{\omega}) \in \mathcal{A}(t_0) \times \mathcal{D}(t_0) \\
%         &\text{ such that } y = X^{t_0,x_0,\alpha_{\omega}, \delta_{\omega}}_{s}\},
% \end{aligned}
\mathcal{X}^{t_0,x_0,\eta, \zeta}_{\tau} := \bigcup_{\omega \in \Omega_{\eta} \times \Omega_{\zeta}} \mathcal{X}^{t_0,x_0,\alpha_{\omega}, \delta_{\omega}}_{\tau}
\vspace{-0.05in}
\end{equation*}
where $(\alpha_{\omega}, \delta_{\omega})$ is defined by Eq.~\eqref{eq:random2det}. Introduce $\rho(\mathcal{S}) = 1$ if $\mathcal{S} \subseteq \mathcal{C}$, and otherwise $\rho(\mathcal{S}) = +\infty$;
and the backward reachable (infeasible) set as
\vspace{-0.05in}
\begin{equation*}
\begin{aligned}
        \bar{\mathcal{Q}}(t) := \{x \in \mathbb{R}^{d_x} ~|~ & \forall \eta \in \mathcal{H}_r(t), \exists \zeta \in \mathcal{Z}_r(t), \tau \in (t, T], \\
        & s.t., \rho\left(\mathcal{X}^{t,x,\eta, \zeta}_{\tau}\right) = +\infty\}. 
\end{aligned}
\vspace{-0.05in}
\end{equation*}
$\mathcal{Q}(t):= \mathbb{R}^{d_x}\setminus \bar{\mathcal{Q}}(t)$ is the set of \textit{feasible states}. $\mathcal{Q}(T) = \mathcal{C}$. Lastly, we use $\bar{\rho}(t,x) = 1$ if $x \in \mathcal{Q}(t)$ and otherwise $\bar{\rho}(t,x) = +\infty$.

% We introduce
% \begin{equation}
%     C(\mathcal{X}^{t_0,x_0,\eta_i, \zeta}_{\tau}) = 
% \end{equation}
\vspace{-0.05in}
\paragraph{Payoff and value} We define the expected payoff of player type $i$ for taking behavioral strategies $(\eta, \zeta)$ as
\vspace{-0.05in}
\begin{equation*}
\small
\begin{aligned}
    G_i(t_0, x_0, \eta, \zeta) &:= \mathbb{E}_{\eta, \zeta} \left[g_i(X_{T}^{t_0,x_0,\eta,\zeta})\rho(\mathcal{X}^{t_0,x_0,\eta, \zeta}_{\tau})\right] \\ 
    % \left \{ \begin{array}{ll}
    %         \int_{\Omega_{\eta} \times \Omega_{\zeta}}  \min_{\tau \in [t, T]} g_i\left(X_{\tau}^{t,x,\alpha_{\omega},\delta_{\omega}}\right) d \textbf{P}_{\eta} \otimes \textbf{P}_{\zeta}(\omega) & \text{if } c(x) \leq 0 \quad \forall x \in \mathcal{X}^{t,x,\eta_i,\zeta}  \\
    %         +\infty & \text{otherwise} 
    %     \end{array}        \right..
    &= \int_{\Omega_{\eta} \times \Omega_{\zeta}}  g_i\left(X_{T}^{t_0,x_0,\alpha_{\omega},\delta_{\omega}}\right)\\ 
    &\quad \quad \quad \quad \rho(\mathcal{X}^{t_0,x_0,\alpha_{\omega}, \delta_{\omega}}_{\tau}) d \textbf{P}_{\eta} \otimes \textbf{P}_{\zeta}(\omega).
\end{aligned}
\vspace{-0.05in}
\end{equation*}
The payoff of Player 1 is $\sum_{i=1}^I p_i G_i(t_0, x_0, \eta_i, \zeta)$.
% \begin{equation}
%         J_1(t_0, x_0, (\eta_i), \zeta, p) :=  
%         % \sup_{\zeta \in \mathcal{Z}_r(t)} \inf_{(\alpha_i) \in (\mathcal{A}(t))^I}
%         % \left \{ \begin{array}{ll}
%         %     \sum_{i=1}^I p_i \textbf{E}_{\eta_i, \zeta}\left(\min_{\tau \in [t, T]} g_i(X_{\tau}^{t,x,\eta_i,\zeta})\right) & \text{if } c(x) \leq 0 \quad \forall x \in \bigcup_{i=1}^I\mathcal{X}^{t,x,\eta_i,\zeta}  \\
%         %     +\infty & \text{otherwise} 
%         % \end{array}        \right.,
%          \sum_{i=1}^I p_i G_i(t_0, x_0, \eta_i, \zeta).
% \end{equation}
% where the expectation $\textbf{E}_{\eta, \zeta}$ is an integral over $\Omega_{\eta} \times \Omega_{\zeta}$ against the probability measure $\textbf{P}_{\eta} \otimes \textbf{P}_{\zeta}$.
With strategys $(\eta_i)\in (\mathcal{H}_r(t_0))^I$ and $\zeta \in \mathcal{Z}_r(t_0)$, the upper value function is defined by
\begin{equation*}
\small
        V^+(t_0, x_0, p) := \inf_{(\eta_i)} \sup_{\zeta} \sum_{i=1}^I p_i G_i(t_0, x_0, \eta_i, \zeta),
% \begin{aligned}
%         &V^+(t_0, x_0, p) :=  \\
%         & \inf_{(\eta_i) \in (\mathcal{H}_r(t_0))^I} \sup_{\zeta \in \mathcal{Z}_r(t_0)} \sum_{i=1}^I p_i G_i(t_0, x_0, \eta_i, \zeta)
% \end{aligned}
\label{eq:valueplus}
\end{equation*}
and the lower value function is given by
\begin{equation*}
\small
        V^-(t_0, x_0, p) :=\sup_{\zeta} \inf_{(\eta_i)} \sum_{i=1}^I p_i G_i(t_0, x_0, \eta_i, \zeta).
% \begin{aligned}
%         &V^-(t_0, x_0, p) :=  \\
%         & \sup_{\zeta \in \mathcal{Z}_r(t_0)} \inf_{(\eta_i) \in (\mathcal{H}_r(t_0))^I} \sum_{i=1}^I p_i G_i(t_0, x_0, \eta_i, \zeta)
% \end{aligned}
\label{eq:valueminus}
\end{equation*}

%%%%%%%%%%%%%%%%%%%%%%%%%%%%%%%%%%%%%%%%%%
%%%%%%%%%%%%%%%%%%%%%%%%%%%%%%%%%%%%%%%%%%
\vspace{-0.1in}
\paragraph{The existence of the value} While the existence of value is proven for both zero-sum complete-information state-constrained differential games~\cite{lee2022safety} and zero-sum differential games with one-sided information~\cite{cardaliaguet2009numerical}, the proof for games with both one-sided information and state constraints is missing. Our main theoretical result fills in this gap (see Appendix \ref{sec:proofexist} for the proof):

%overarching theorem here
\begin{theorem}
\label{theorem:existence}
    If assumptions 1-5 hold, we have
    $V^+(t,x,p)=V^-(t,x,p)$ for all $(t,x,p) \in [0,T] \times \mathbb{R}^{d_x} \times \Delta(I)$.
    \vspace{-0.05in}
\end{theorem}

% \begin{proof}
% Sketch: Use Lemma 3.1 of \cite{cardaliaguet2007differential} to show that the second term in $\max$ is Lipschitz continuous. Show that the set in the first term changes continuously wrt $t$ and $x$, and then show the first term is also Lipschitz continuous.  
% \end{proof}
\vspace{-0.05in}
\paragraph{Characterization of the value}
We need to first characterize the value of the unconstrained game since this value will later appear in that of the state-constrained game. 
% First, let $w: [0, T] \times \mathbb{R}^{d_x} \times \Delta(I) \rightarrow \mathbb{R}$ be some function and $w^*$ be its convex conjugate with respect to $p$:

Let $U: [0,T] \times \mathbb{R}^{d_x} \times \Delta(I) \rightarrow \mathbb{R}$ be the value of the unconstrained version of the game, and $U^*: [0,T] \times \mathbb{R}^{d_x} \times \mathbb{R}^I \rightarrow \mathbb{R}$ its convex conjugate: 
\begin{equation*}
\begin{aligned}
    U^*(t, x, \hat{p}) := &\sup_{p \in \Delta(I)} \hat{p}^Tp - U(t, x, p) \quad \\
    & \forall (t, x, \hat{p}) \in [0, T] \times \mathbb{R}^{d_x} \times \mathbb{R}^{I}. 
    \end{aligned}
\end{equation*}
We have the following properties for $U$ and $U^*$:
\begin{enumerate}[noitemsep]
\vspace{-0.1in}
    \item $U$ is Lipschitz continuous in $(t,x,p)$ and convex to $p$. $U(T,x,p) = \sum_{i=1}^I p_i g_i(x)$, $\forall (x,p) \in \mathbb{R}^{d_x} \times \Delta(I)$.
    $U^*$ is Lipschitz continuous in $(t,x,\hat{p})$ and convex to $\hat{p}$. $U^*(T,x,\hat{p}) = \max_{i \in [I]} \hat{p}_i - g_i(x)$, $\forall (x,\hat{p}) \in \mathbb{R}^{d_x} \times \mathbb{R}^I$. 
    \item For any $p \in \Delta(I)$, $(t,x) \rightarrow U(t,x,p)$ is a viscosity subsolution to the primal HJ equation
    \vspace{-0.05in}
    \begin{equation*}
        w_t + H(x, Dw) = 0,
        \vspace{-0.05in}
    \end{equation*}
    where $H$ is defined by Eq.~\eqref{eq:hamiltonian}.
    \item For any $\hat{p} \in \mathbb{R}^I$, $(t,x) \rightarrow U^*(t,x,p)$ is a viscosity subsolution to the dual HJ equation
    \vspace{-0.05in}
    \begin{equation*}
        w_t + H^*(x, Dw) = 0,
        \vspace{-0.05in}
    \label{eq:dualHJ_uncons}
    \end{equation*}
    where $H^*(x, \xi) = -H(x, -\xi)$ $\forall~(x, \xi) \in \mathbb{R}^{d_x} \times \mathbb{R}^{d_x}$.
    \vspace{-0.2in}
\end{enumerate}

The conjugate $U^*$ defines the value of a \textit{dual game} where Player 2 minimizes her payment, $U^*(T,x,\hat{p})$, to Player 1 where $\hat{p}$ is common knowledge. We note that by definition (see \cite{de1996repeated}), the dual variables $\hat{p}$ are the info-state values defined in Brown et al.~\yrcite{brown2020combining}, i.e., $\hat{p}[i]$ captures the value of Player 1 if he is of type $i$ and plays the best response to Player 2's equilibrium strategy. De Meyer~\yrcite{de1996repeated} showed that when $\hat{p}\in \partial_p U(0,x,p)$, Player 2's strategy in the dual game is her equilibrium in the primal game. We note that Player 2's strategy cannot be derived from the primal subdynamic principle because her best response is dependent on Player 1's type which is unknown to her. The introduction of the dual game allows us to derive a subdynamic principle of the conjugate value from where her equilibrium strategy can be computed.

For the state-constrained game, the following corollary is a result of the subdynamic principles derived from Theorem~\ref{theorem:existence}, and will guide the numerical approximation of values for the state-constrained game (Sec.~\ref{sec:numerical}): 

\begin{corollary}
    Under assumptions 1-5, the value function $V := V^+ = V^-$ (resp. $V^*$) is a unique function defined on $[0,T] \times \mathbb{R}^{d_x} \times \Delta(I)$ (resp. $[0,T] \times \mathbb{R}^{d_x} \times \mathbb{R}^{I}$) such that:
    \begin{enumerate}[noitemsep]
        \item $V$ is convex respect to $p$ and
        \begin{equation}
            V(T, x, p) = \rho(x) U(T,x,p) \; \forall (x, p) \in \mathbb{R}^{d_x} \times \Delta(I);
            \label{eq:primalBC}
        \end{equation}
        $V^*$ is convex respect to $\hat{p}$ and
        \begin{equation}
            V^*(T, x, \hat{p}) = \max_{i \in [I]} \hat{p}_i - \rho(x)g_i(x)  \; \forall (x, \hat{p}) \in \mathbb{R}^{d_x} \times \mathbb{R}^I
            \label{eq:dualBC}
        \end{equation}
        \item For all $p \in \Delta(I)$, $(t,x) \rightarrow V(t,x,p)$ is a viscosity subsolution to the primal HJ equation
        \begin{equation}
            \min\{\rho(x)U(t,x,p)-w, w_t + H(x, Dw)\} = 0.
        \label{eq:primalHJ}
        \end{equation}
        
        % where $H$ is defined by Eq.~\eqref{eq:hamiltonian}.
        \item For all $\hat{p} \in \mathbb{R}^I$, $(t,x,z) \rightarrow V^*(t,x,z,p)$ is a viscosity subsolution to the dual HJ equation
        \begin{equation}
            \min\{\rho(x)U^*(t,x,\hat{p}/\rho(x)), w_t + H^*(x, Dw)\} = 0.
        \label{eq:dualHJ}
        \end{equation}
        % where $H^*(x, \xi) = -H(x, -\xi)$ for any $(x, \xi) \in \mathbb{R}^{d_x} \times \mathbb{R}^{d_x}$.
    \end{enumerate}
    % $V$ is called the unique dual solution of Eq.~\eqref{eq:primalHJ} with boundary condition in Eq.~\eqref{eq:primalBC}.
    \vspace{-0.1in}
\end{corollary}

\vspace{-0.1in}
\section{Bellman Backup and Behavioral Strategies}\label{sec:numerical}
\vspace{-0.05in}
% The value of the game, $V$, satisfies conditions in Eq.~\eqref{eq:primalBC} and \eqref{eq:primalHJ}.
% The value of the game, $V$, satisfies the following conditions (cf. Proposition 2.1 in \cite{cardaliaguet2009numerical}):
% \begin{enumerate}[label=(\roman*)]
%     \item $V$ is Lipschitz continuous in all its variables ($t$, $x$, $p$), convex with respect to $p$, and such that
%         \begin{equation}
%         \begin{aligned}
%             V(T, x, p) = \delta(x) \sum_{i=1}^I p_i g_i(x), \\
%             \quad \forall (x, p) \in \mathbb{R}^{d_x} \times \Delta(I)
%         \end{aligned}
%         \end{equation}
%     \item for any $p \in \Delta(I)$, $(t, x) \rightarrow V(t, x, p)$ is a viscosity subsolution of the primal HJ equation
%     \begin{equation}
%     \begin{aligned}
%         \max\{\delta (x)U(t, x, p), w_t + H(x, Dw)\} = 0 ~ \\
%         \text{in } [0, T] \times \mathbb{R}^{d_x},
%     \end{aligned}
%     \end{equation}
%     where under the Isaacs' condition
%     \begin{equation}
%     \begin{aligned}
%         H(x, \xi) &:= \min_{u \in \mathcal{U}} \max_{v \in \mathcal{V}} f(x, u, v)^T \xi \\ &= \max_{v \in \mathcal{V}}\min_{u \in \mathcal{U}}f(x, u, v)^T \xi
%     \end{aligned}
%     \end{equation}
%     for any $(x, \xi) \in \mathbb{R}^{d_x} \times \mathbb{R}^{d_x}$.
% \end{enumerate}
Discrete-time Bellman backup computes an approximated value $V_{\tau}(t_k, \cdot, \cdot): \mathbb{R}^{d_x} \times \Delta(I) \rightarrow \mathbb{R}$, with time step $\tau = T/L$ for some large $L$ and $t_k = k\tau$ for $k = 0,...,L$:
\begin{enumerate}[label=(\roman*)]
\vspace{-0.05in}
    \item At the terminal time, set $V_{\tau}(T, x, p) = \rho(x)\sum_i p_i g_i(x)$.
    % \vspace{-0.05in}
    % if $\rho(x) = 1$, otherwise $V_{\tau}(T,x,p) = +\infty$.
    % \vspace{-0.05in}
    \item At $k \in \{0, ..., L-1\}$
    % \vspace{-0.05in}
    \begin{equation}\label{eq:backup_cons}
    % \small
    % \begin{aligned}
    % &V_\tau(t_k, x, p) =\\
    % &\begin{cases}
    %     \text{Vex}_p \left(\min_u \max_v V_\tau (t_{k+1}, x + \tau f(x, u, v), p)\right),\\ \hfill x \in \mathcal{C}\\
    %     K, x \notin \mathcal{C}
    %     \end{cases}
    % \end{aligned}
    V_\tau(t_k, x, p) = \rho(x)\text{Vex}_p \left(\min_u \max_v V_\tau (t_{k+1}, x', p)\right),
    % \vspace{-0.05in}
    \end{equation}
    % if $\bar{\rho}(t_k,x)=1$, 
    where $x' = x + \tau f(x, u, v)$ and $\text{Vex}_p(\cdot)$ is the convex hull with respect to $p$. 
    % otherwise $V_{\tau}(t_k, x, p) = +\infty$.
\end{enumerate}

Let $l: \mathcal{U}\times\mathcal{V} \rightarrow \mathbb{R}$ be a Lipschitz continuous and bounded function that represents the instantaneous payoff of the game. To incorporate $l$, Eq.~\eqref{eq:backup_cons} becomes
\begin{equation}
\label{eq:backup_cons_instantaneous}
\small
    V_\tau(t_k, x, p) = \rho(x)\text{Vex}_p \left(\min_u \max_v V_\tau (t_{k+1}, x', p) + \tau l(u,v)\right)
\end{equation}
Theorem~\ref{thm:value_convergence} states that $V_{\tau}$ uniformly converges to $V$ as $\tau \rightarrow 0^+$ (see proof in Appendix~\ref{sec:value_convergence}):

% Note that for the unconstrained game, $\delta(x) = 1$ (always), and the backward induction reduces to only the first case in \eqref{eq:backup_cons} for all $x \in \mathbb{R}^{d_x}$, as discussed in \cite{cardaliaguet2009numerical}.
% \paragraph{Backward Induction for State Constrained Game} Assuming $V_\tau (t_{k+1}, \cdot, \cdot)$ is built, 
% % \begin{equation}\label{eq:backup_cons}
% %     \begin{aligned}
% %         &V_\tau(t_k, x, p) = \\
% %         & \delta(x) \text{Vex}_p \left(\min_u \max_d V_\tau (t_{k+1}, x_k + \tau f(x_k, u, d), p)\right)
% %     \end{aligned}
% % \end{equation}
% \begin{equation}\label{eq:backup_cons}
% \small
%     \begin{aligned}
%         &V_\tau(t_k, x, p) =\\
%         &\begin{cases}
%             \text{Vex}_p \left(\min_u \max_d V_\tau (t_{k+1}, x_k + \tau f(x_k, u, d), p)\right), x \in \mathcal{C}\\
%             +\infty, x \notin \mathcal{C}
%         \end{cases}
%     \end{aligned}
% \end{equation}
% \begin{remark}\label{rem:backup}
    % From \eqref{eq:backup_cons}, it is sufficient to build the value only for those states satisfying state constraints. At any given time, the values of the states not satisfying state constraint are set to $+\infty$. Hence, the convergence proof of the backward induction only deals with the states satisfying state constraints (i.e. $x \in \mathcal{C}$).  
% \end{remark}
\begin{theorem}\label{thm:value_convergence}
    % Given the assumptions in section (\ref{sec:problem}) hold, 
    If assumptions 1-5 hold,
    $V_\tau$ converges uniformly to $V$ on compact subsets of $[0, T] \times \mathbb{R}^{d_x} \times \Delta(I)$:
    \begin{align*}
        &\lim_{\scriptstyle \begin{array}{l}
             \tau \rightarrow 0^+,\; t_k \rightarrow t,\\
             x' \rightarrow x, p' \rightarrow p 
        \end{array}} V_\tau(t_k, x', p') = V(t, x, p)\\ &\forall (t, x, p) \in [0, T] \times \mathbb{R}^{d_x} \times \Delta(I).
    \end{align*}
    \vspace{-0.2in}
\end{theorem}
\vspace{-0.1in}
\paragraph{Behavioral strategy for Player 1 and the belief dynamics} Player 1's behavioral strategy is a probability distribution over $\mathcal{U}$ conditioned on $(t, x, p)$. 
% We note that the belief, $p(t)$, changes during the game and is common knowledge under the assumption that both players know $p(t_0)$ and the belief dynamics. 
% To reduce notational burden, we will shorten $a(t_k)$ as $a_k$ for any time-dependent variable $a$ on the discrete-time grid. 
At time $t_k$, Player 1 resigns if $x_k \in \bar{\mathcal{Q}}(t_k)$; otherwise, he determines his strategy using the following steps: First he finds $\lambda = \{\lambda_1,...,\lambda_I\} \in \Delta(I)$ and $p^j \in \Delta(I)$ for $j=1,...,I$, such that
\vspace{-0.1in}
\begin{equation}\label{eq:strategy}
% \small
\begin{aligned}
    V_{\tau}(t_k, x_k, p_k) &= \sum_{j=1}^I \lambda_j \Big(\min_{u \in \mathcal{U}} \max_{v \in \mathcal{V}} V_{\tau}(t_{k+1}, x'_k, p^j)\Big), \\
    \sum_{j=1}^I \lambda_j p^j &= p_k.
\end{aligned}
\vspace{-0.05in}
\end{equation}
Then he computes $u^j$ as the minimax solution corresponding to $p^j$,
% $\min_{u \in \mathcal{U}} \max_{v \in \mathcal{V}} V_{\tau}(t_{k+1}, x_k + \tau f(x_k, u, v), p^j)$
and chooses $u_k = u^j$ with probability $\lambda_j p^j[i]/p_k[i]$, where $i$ is its true type. It is proved that this behavioral strategy of Player 1 is $\epsilon$-optimal for small enough $\tau$~\cite{cardaliaguet2009numerical}. 
Importantly, $\{p^j\}_{j=1}^I$ are vertices of the value convex hull. Thus by announcing his strategy, and by assuming that players use the same Bayes belief update, Player 1 controls the belief dynamics to follow a martingale that optimizes his gain, i.e., $p_{k+1} = p^j$ if $u^j$ is chosen. Note that the introduction of state constraints changes the minimax solutions, the value convex hulls, and thus the behavioral strategies. Lastly, Eq.~\eqref{eq:strategy} will be modified according to Eq.~\eqref{eq:backup_cons_instantaneous} when instantaneous loss is present. 

\vspace{-0.1in}
\paragraph{The dual game and behavioral strategy for Player 2} 
% \cite{de1996repeated} first introduced the concept of a dual game in the context of repeated games with one-sided information, and showed that a recursive value structure appears for Player 2 in the dual game which is missing in the primal game. 
% Similar to the recursive value structure for Player 1 in the primal game, this allows us to derive a Bellman equation and use backward induction to compute Player 2's value and behavioral strategy. 
% In addition, the value of the dual game is shown to be the conjugate of that of the primal. \cite{cardaliaguet2009numerical} extended this idea to differential games, which we adopt in this paper.
% To summarize, 
Player 2's strategy is determined by a dual game for which the conjugate value is approximated by $V^*_{\tau}(t_k, \cdot, \cdot): \mathbb{R}^{d_x} \times \mathbb{R}^{I} \rightarrow \mathbb{R}$. Specifically,
% \vspace{-3pc}
\vspace{-0.05in}
\begin{enumerate}[label=(\roman*), noitemsep]
\vspace{-0.1in}
    \item At the terminal time, set $V^*_{\tau}(T, x, \hat{p}) = \max_i\{\hat{p}_i - \rho(x)g_i(x)\}$.
    \item At $k \in \{0,...,L-1\}$
    \vspace{-0.05in}
    \begin{equation}
    %     \small
    % \begin{aligned}
    %     &C_{\tau}(t_k, x, \hat{p}) = \\
    %     &\begin{cases}
    %     \text{Vex}_{\hat{p}} \left(\min_v \max_u C_\tau (t_{k+1}, x + \tau f(x, u, v), \hat{p})\right),\\ \hfill x \in \mathcal{C}\\
    %     -K, x \notin \mathcal{C}
    %     \end{cases}
    %     % &\text{Vex}_{\hat{p}}\left ( \min_{v \in \mathcal{V}} \max_{w \in \text{Co}f(x,\mathcal{U}, v)} C_{\tau}(t_{k+1}, x + \tau w, \hat{p})\right),
    % \end{aligned}
    V^*_{\tau}(t_k, x, \hat{p}) =
        \text{Vex}_{\hat{p}} \left(\min_v \max_u V^*_\tau (t_{k+1}, x', \hat{p})\right),
    \vspace{-0.05in}
    \label{eq:backup_dual}
    \end{equation}
     if $\bar{\rho}(t_k,x)=1$; otherwise $V^*_{\tau}(t_k, x, \hat{p}) = -\infty$.
     \vspace{-0.1in}
\end{enumerate}
% where $\text{Vex}_{\hat{p}}(\cdot)$ is the convex hull with respect to $\hat{p}$.
Similar to Theorem~\ref{thm:value_convergence}, Theorem~\ref{thm:dual_convergence} proves that $V^*_{\tau}$ uniformly converges to $V^*$ as $\tau \rightarrow 0^+$ (proof omitted).

\begin{theorem}\label{thm:dual_convergence}
    % Given the assumptions in section (\ref{sec:problem}) hold, 
    If assumptions 1-5 holds,
    $V^*_\tau$ converges uniformly to $V^*$ on compact subsets of $[0, T] \times \mathbb{R}^{d_x} \times \mathbb{R}^{I}$:
    \begin{align*}
        &\lim_{\scriptstyle \begin{array}{l}
             \tau \rightarrow 0^+,\; t_k \rightarrow t,\\
             x' \rightarrow x, \hat{p}' \rightarrow \hat{p} 
        \end{array}} V^*_\tau(t_k, x', \hat{p}') = V^*(t, x, \hat{p})\\ &\forall (t, x, \hat{p}) \in [0, T] \times \mathbb{R}^{d_x} \times \mathbb{R}^{I}.
    \end{align*}
    \vspace{-0.1in}
\end{theorem}
\vspace{-0.1in}
With instantaneous loss $l$, the Bellman backup in Eq.~\eqref{eq:backup_dual} becomes
\vspace{-0.05in}
\begin{equation}
\small
    V^*_{\tau}(t_k, x, \hat{p}) =
        \text{Vex}_{\hat{p}} \left(\min_v \max_u V^*_\tau (t_{k+1}, x', \hat{p} - \tau l(u,v))\right)
\end{equation}
We explain this modification in detail in Appendix~\ref{sec:modify}. An intuitive explanation is as follows: Recall that each element of $\hat{p}$ represents Player 1's value for the corresponding type in the primal game. Hence $\hat{p}$ at the next time step should discount the common instantaneous loss incurred at the current time step. 

The behavioral strategy of Player 2 defines a probability distribution over $\mathcal{V}$ conditioned on $(t, x(t), \hat{p}(t))$, with the dual variable $\hat{p}(t_0) \in \partial_{p} V(t_0, x_0, p(t_0))$. At time $t_k$, if $x_k \in \bar{\mathcal{Q}}(t_k)$, Player 2 plays according to a pursuit-evasion game since she can always catch Player 1 according to the definition of $\bar{\mathcal{Q}}(t_k)$; otherwise, Player 2 determines her strategy using the following steps: First she finds $\lambda = \{\lambda_1,...,\lambda_{I+1}\} \in \Delta(I+1)$ and $\hat{p}^j \in \mathbb{R}^{I}$ for $j = 1,..., I+1$, such that
\vspace{-0.05in}
\begin{equation}
\begin{aligned}
\small
    V^*_{\tau}(t_k, x_k, \hat{p}) &= \sum_{j=1}^{I+1} \lambda_j \Biggl( \min_{v \in \mathcal{V}} \max_{u \in \mathcal{U}} V^*_{\tau}(t_{k+1}, x'_k, \hat{p}^j)\Biggr),\\
    \sum_{j=1}^{I+1} \lambda_j \hat{p}^j &= \hat{p}_k. 
    \label{eq:backup_cons_dual}
\end{aligned}
\vspace{-0.05in}
\end{equation}
Then she computes the minimax solution $v^j$, and chooses $v_k = v^j$ with probability $\lambda_j$. It is proved that this behavioral strategy of Player 2 is $\epsilon$-optimal for small enough $\tau$~\cite{cardaliaguet2009numerical}. $\hat{p}$ follows a martingale $\hat{p}_{k+1} = \hat{p}^j$ if $v^j$ is chosen by Player 2, or $\hat{p}_{k+1} = \hat{p}^j - \tau l(u^j,v^j)$ if $l$ is present, where $u^j$ is the best response to $v^j$ in the dual game.  %% e-optimal for \tau > 0. for p2 
Notice that due to her lack of information, Player 2 solves harder value approximation and control synthesis problems of belief dimension $I+1$ rather than $I$.   

To help readers better understand the mechanisms described in this section, we provide detailed derivation of behavioral strategies for two sample problems in Appendix~\ref{sec:example} (\ref{sec:beerquiche} for a zero-sum version of the beer-quiche game and \ref{sec:hexner} for Hexner's game).

% In this work, we are focused on solving the game from the perspective of the informed player by allowing the uninformed player to follow the worst-case strategy (from the informed player's view). Solving the game from the uninformed player's perspective is quite similar to that from the informed player's perspective, with the difference being that the uninformed player keeps track of a different process, namely, $\hat{p}$, which lies in the sub-differential of the primal value function. A comprehensive study of the strategies for both the informed and the uninformed player is a topic of immediate future work. 

%===============================================================================

% \section{A Reach-Avoid Game with One-Sided Information}
% \label{sec:reachavoid}

%===============================================================================
\vspace{-0.05in}
\section{Numerical Methods}
\label{sec:value}
% Approximating value function along time and space is crucial in deriving the control input of the players (for eg. see \eqref{eq:strategy}).
\vspace{-0.05in}
\subsection{Primal and dual value approximation}
\vspace{-0.05in}
% Standard RL cannot be used to solve Eq.~\eqref{eq:backup_cons} and \eqref{eq:backup_cons_dual} because the value convex hull at $t+1$ across $\Delta(I)$ (or $\mathbb{R}^{I}$ for dual variables) needs to be obtained for correct backup at $t$. Therefore we resort to 
We use backward induction to solve Eq.~\eqref{eq:backup_cons} and \eqref{eq:backup_cons_dual}, and discuss treatments that alleviate error propagation. We focus the discussion on the primal problem for brevity.   

\vspace{-0.1in}
\paragraph{Value discontinuity} At each time step $t$, $V_{\tau}(t,\cdot,\cdot)$ (resp. $V_{\tau}^*(t,\cdot,\cdot)$) can be approximated separately in $\bar{\mathcal{Q}}(t)$ and $\mathcal{Q}(t)$: the primal (resp. dual) value in the former is set to $+\infty$ (resp. $-\infty$) and value in the latter will be approximated using a neural network $\hat{V}_{\tau}(t,\cdot,\cdot)$ (resp. $\hat{V}^*_{\tau}(t,\cdot,\cdot)$). This avoids fitting the value networks to functions with large Lipschitz constants during numerical implementation. $\bar{\mathcal{Q}}(t)$ for $t \in [0,T]$ can be approximated by a physics-informed neural network (PINN) solver~\cite{bansal2021deepreach} (see details in Appendix~\ref{sec:BRT}), by recognizing that $\bar{\mathcal{Q}}(t)$ can be defined by pure strategies instead of behavioral ones using Lemma~\ref{lemma:strategy}. PINN alleviates CoD in solving HJ equations with Lipschitz continuous solutions~\cite{pinn_convergence}, and here it results in a separate value network $\tilde{V}(\cdot,\cdot): [0,T]\times \mathbb{R}^{d_x}$ that approximates $\bar{\mathcal{Q}}(t)$ as $\{x \in \mathbb{R}^{d_x} | \tilde{V}(t,x) \leq 0\}$.  

\vspace{-0.1in}
\paragraph{Partially convex values} At each $t_k$ and for uniformly sampled $\mathcal{S}(t) \subset \mathcal{Q}(t)$, we scan a lattice $\mathcal{P} \subset \Delta(I)$ to obtain the minimax solution of the RHS of Eq.~\eqref{eq:backup_cons} (denoted by $\vartheta^0(t,x,p)$ for $(x,p) \in \mathcal{S}(t)\times \mathcal{P}$), resulting in a dataset $\{(p, \vartheta^0(t,x,p))\}_{p \in \mathcal{P}}$. Value convexification is then obtained by applying the Monotone Chain Convex Hull algorithm to this dataset for each $x \in \mathcal{S}(t)$ and taking the lower hull of the resulting convex hull. Let $\vartheta(t,x,p)_{\mathcal{S}(t)}$ be the resultant value on the convex hull. A value network $\hat{V}_{\tau}(t,\cdot,\cdot)$ is then trained using data $\{(x, \vartheta_{\mathcal{S}(t)}(t,x,p) | (x,p) \in \mathcal{S}(t) \times \mathcal{P}\}$ so that during the Bellman backup at $t-1$, we can predict convexified values at previously unvisited states at $t$. 
We use a Partially Input Convex Neural Network (PICNN)~\cite{amos2017input} to ensure that $\hat{V}_{\tau}(t,\cdot,\cdot)$ is convex in $p$. Alg.~\ref{alg:value} summarizes the value approximation algorithm. Optionally, we also train a separate value network to predict the minimax values using $\{(p, \vartheta^0(t,x,p))\}_{p \in \mathcal{P}}$. This network helps remove the nested minimax problem during control synthesis.  

\vspace{-0.1in}
\paragraph{Convexification error.} Backward induction suffers from error propagation, where errors at each time step are originated from (i) value approximation through neural networks, (ii) backward reachable set approximation, (iii) convex hull approximation, and (iv) finite time discretization (and Euler method for ODE). Here we discuss control of the error resulted from convex hull approximation, which is unique to incomplete-information games. We leave a full analysis of data complexity for error control to future studies.
At each $t\in [0,T]$ and $x \in \mathcal{Q}(t)$, let $\vartheta(t,x,\cdot)$ be the RHS of Eq.~\eqref{eq:backup_cons} after convexification, and the convexification error be $\varepsilon_{vex}(t,x) := \max_{p \in \Delta(I)} \|\vartheta(t,x,p) - \vartheta_{\mathcal{S}(t)}(t,x,p)\|_{\infty}$. Proposition~\ref{prop:convex_error} shows that $\varepsilon_{vex}(t,x)$ can be controlled by refining $\mathcal{P}$ (see proof in Appendix~\ref{sec:convex_error}):

\begin{proposition}\label{prop:convex_error}
    For given $(t,x)$, let the Lipschitz constant of $\vartheta(t,x,\cdot)$ be $L$, and $d_{\mathcal{P}}$ be the minimum distance between two neighboring nodes of the lattice $\mathcal{P}$. $\varepsilon_{vex}(t,x) \leq 2d_{\mathcal{P}}L$.
\end{proposition}

\vspace{-0.15in}
\paragraph{Approximation of the conjugate value.} Recall that the dual game is initialized by the dual variables $\hat{p} \in \partial_p V(0,x,p)$ when the primal game starts at $(x,p)$. Since $\hat{V}(0,x,\cdot)$ is a differentiable neural network defined on a simplex, subgradients can be found using $\hat{p}^Tp = V(0,x,p)$ and $\hat{p}^Tq \leq V(0,x,q)$ for all $q \in \Delta(I)$ and $q \neq p$. Specific to the case study where $I=2$ and $\hat{V}:=\hat{V}_{\tau}(0,x,p[1])$ is modeled to be a function of the first element of $p$ to reduce dimensionality, we have $\hat{p} = (\hat{V} + \nabla_{p[1]}\hat{V}(1-p[1]), \hat{V} - \nabla_{p[1]}\hat{V}p[1])^T$.
% We also note that $\hat{V}^*$ is one dimension larger than $\hat{V}$ since $\hat{p}$ is not constrained on a simplex. This reflects the additional cost of Player 2 due to her lack of information. 

\vspace{-0.05in}
\subsection{Synthesis of strategies}
\vspace{-0.05in}
Given $(t,x,p) \in [0,T]\times \mathbb{R}^{d_x}\times \Delta(I)$, Player 1 computes his behavioral strategy by finding $\lambda\in \Delta(I)$ and splitting beliefs $\{p^j \in \Delta(I)\}_{j=1}^I$ that best satisfy Eq.~\eqref{eq:strategy} in $L^2$, if $x \in \mathcal{Q}(t)$ (otherwise he surrenders). Given $(t,x,\hat{p}) \in [0,T]\times \mathbb{R}^{d_x}\times \mathbb{R}^{I}$, Player 2 finds $\lambda\in \Delta(I+1)$ and $\{\hat{p}^j \in \mathbb{R}^{I}\}_{j=1}^{I+1}$ that best satisfy Eq.~\eqref{eq:backup_cons_dual} in $L^2$. When $I=2$ as in the case study, the splitting beliefs and resultant strategies for Player 1 can be approximated through sweeping $p[1] \in [0,1]$. For Player 2, we use gradient descent to solve a 6D optimization problem with the initial guess $\lambda = [1/3,1/3,1/3]^T$ and $\hat{p}^j = \hat{p}$ for $j=[3]$.

\begin{algorithm}[!ht]
\caption{Value Approximation}
\label{alg:value}
\textbf{Inputs:} current time step $t$, time discretization $\tau$, sample size $N$, admissible action spaces $(\mathcal{U}(t), \mathcal{V}(t))$, value approximation at $t + \tau$: $\hat{V}_{\text{next}}(\cdot,\cdot) := \hat{V}_{\tau}(t+\tau,\cdot,\cdot)$, feasible state set $\mathcal{Q}(t)$, instantaneous loss $l(\cdot,\cdot)$, terminal loss in Eq.~\eqref{eq:primalBC}\\
\textbf{Initialize:} $\hat{V}_\tau(t,\cdot,\cdot)$, $\vartheta^0 = \emptyset$
\begin{algorithmic}[1]
\STATE $\mathcal{S}(t) \leftarrow $ sample $N$ states from $\mathcal{Q}(t)$
% $\mathcal{X} = (d_{x_1}, d_{y_1}, \dot{d}_{x_1}, \dot{d}_{y_1}, d_{x_2}, d_{y_2}, \dot{d}_{x_2}, \dot{d}_{y_2})$
% \STATE $\mathcal{S}'(t) \leftarrow$ compute next states from $\mathcal{S}(t)$ via Eq.~\eqref{eq:dynamics} $\forall (u,v) \in \mathcal{U}(t) \times \mathcal{V}(t)$
\FOR {$x$ in $\mathcal{S}$}
    \FOR {$p$ in $\mathcal{P}$}
        \STATE $v(x,p) \leftarrow \min_{u \in \mathcal{U}(t)} \max_{v \in \mathcal{V}} \hat{V}_{\text{next}}(x', p) + \tau l(\mathcal{U}, \mathcal{V})$ \COMMENT{if $t + \tau = T$, $\hat{V}_{\text{next}}$ is given by Eq.~\eqref{eq:primalBC}}
        \STATE append $v(x, p)$ to $\vartheta^0$
    \ENDFOR
    \STATE $\vartheta_{\mathcal{S}(t)}(t,x,\cdot) \leftarrow$ compute $\text{Vex}_p(\vartheta^0(x,\cdot))$ via Eq.~\eqref{eq:backup_cons}
\ENDFOR
\STATE Update $\hat{V}_{\tau}$ to match $\vartheta_{\mathcal{S}(t)}$
\end{algorithmic}
\end{algorithm}

%===============================================================================

\vspace{-0.1in}
\section{Case Study}
\label{sec:cases}
\vspace{-0.05in}
\paragraph{Setup} We study a state-constrained version of Hexner's game that represents a simplified football game: Player 1 (P1)'s goal is to move closer to one of the two targets than P2 without being caught during the interaction (see Fig.~\ref{fig:2d_soccer}); P2's goal is to catch P1 if possible, or otherwise move close to P1's target. Each player has 4 state variables: x- and y- position and velocity; and their actions encode x- and y-acceleration.
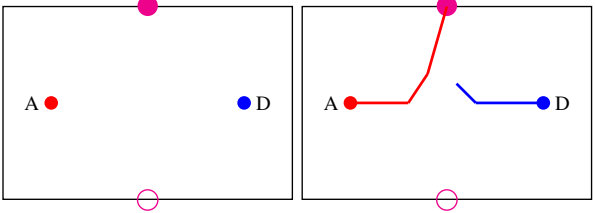
\begin{figure}[!ht]
    \centering
    \vspace{-0.1in}
    \begin{tikzpicture}
    \draw (0,0)  node[minimum height=2cm,minimum width=3cm,draw] (A) {};
    \fill[blue]  (1, 0) circle(2pt);
    \node[right] at (1, 0) {\tiny D};
    \fill[red]  (-1, 0) circle(2pt);
    \node[left] at (-1, 0) {\tiny A};
    \fill[magenta] (A.north) circle(3pt);
    \draw[magenta] (A.south) circle(3pt);
    % \node[below left] at (-1.4, -0.9) (l) {\tiny -1;
    % \node[below right] at (1.4, -0.9) (r) {\tiny 1};
    % \node[above left] at (-1.4, 0.9) (ul) {\tiny -1};
    % \node[above right] at (1.4, 0.9) (ur) {\tiny 1};
\end{tikzpicture}

\begin{tikzpicture}
    \draw (0,0)  node[minimum height=2cm,minimum width=3cm,draw] (A) {};
    \fill[blue]  (1, 0) circle(2pt);
    \node[right] at (1, 0) {\tiny D};
    \fill[red]  (-1, 0) circle(2pt);
    \node[left] at (-1, 0) {\tiny A};
    \fill[magenta] (A.north) circle(3pt);
    \draw[magenta] (A.south) circle(3pt);
    % \node[below left] at (-1.4, -0.9) (l) {\tiny -1};
    % \node[below right] at (1.4, -0.9) (l) {\tiny 1};
    % \node[above left] at (-1.4, 0.9) (l) {\tiny -1};
    % \node[above right] at (1.4, 0.9) (l) {\tiny 1};
    \draw[red, thick] (-1, 0) -- (-0.4, 0);
    \draw[red, thick] (-0.4, 0) -- (-0.2, 0.3);
    \draw[red, thick] (-0.2, 0.3) -- (0, 1);
    \draw[blue, thick] (1, 0) -- (0.3, 0);
    \draw[blue, thick] (0.3, 0) -- (0.1, 0.2);
\end{tikzpicture}
    \vspace{-0.1in}
    \caption{Schematics of a simplified football game with Player 1 (red) and Player 2 (blue). 
    Left: the initial configuration. Right: equilibrium trajectory. Magenta circles: two goals. The filled is the current type private to Player 1. Players move in a 2D space bounded by $[-1, 1] \times [-1, 1]$.}
    \label{fig:2d_soccer}
    \vspace{-0.1in}
\end{figure}
The parameters $R_A = diag([0.05, 0.025])$ and $R_D = diag([0.05, 0.1])$ are chosen so that P1 can afford to accelerate faster in the y-direction than P2. The state constraint is $c(x_1, x_2) = r - \|(d_{x_1}, d_{y_1}) - (d_{x_2}, d_{y_2})\|_2$, where $r = 0.05$. We note that due to the introduction of an (instantaneous) effort loss, the backward induction is modified as: $V_\tau(t_k, x, p) = \text{Vex}_p(\min_u \max_v V_\tau(t_{k+1}, x+\tau(f, x, u, v), p) + \tau l(u, v))$, where $l(u, v)$ is the integral term in $[t_k,t_{k+1}]$ in Eq.~\eqref{eq:value_def}.

\textbf{Value network architecture and training} 
The value network uses PICNN with 5 hidden layers and 256 neurons each and has 9-dimensional inputs (state and belief). 
% The networks are designed to be convex to $p$. 
We train 10 separate networks for each time step starting from $t=0.9$ with $\tau=0.1$, each being trained for 10 epochs.  
% In the case of the state-constrained game, we modify the neural network so that the value outside of the safe-region $\mathcal{C}$ is a high positive number.  
For each epoch, $\mathcal{S}(t)$ includes 5000 states sampled from $\mathcal{Q}(t)$. Since $I=2$, value networks can be considered as functions of $p[1]$ and thus we set $\mathcal{P} = \{p[1] = 0, 0.01, \dots, 0.99, 1\}$. $\hat{V}_{\tau}(t,\cdot,\cdot)$ is trained on data collected from $\mathcal{S}(t) \times \mathcal{P}$. For the conjugate $\hat{V}^*_{\tau}$, we set $\hat{\mathcal{P}} = \{\hat{p}[1] = \{-14, \dots, 14\}, \hat{p}[2] = \{-14, \dots, 14]\}$. 
% For the penultimate time step, the data generation process is comparatively efficient as the value at the terminal time is known. 
More details can be found in Appendix~\ref{sec:case_details}. 
%===============================================================================

% delay and manipulation
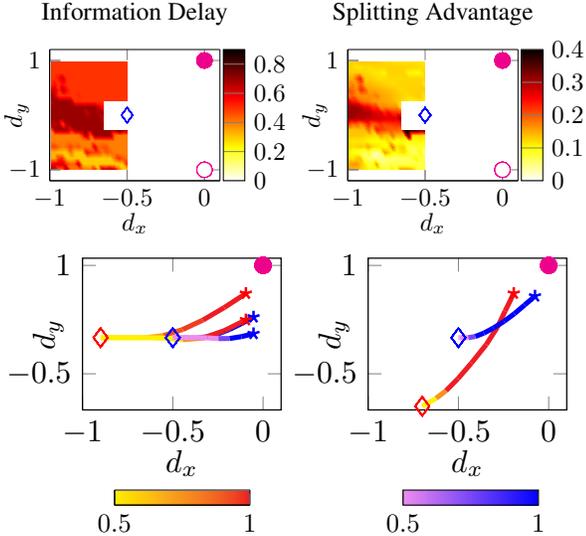
\begin{figure}
    \begin{minipage}{.48\linewidth}
    % \hspace{0.1in}
    \centering
    \begingroup%
    \StrCount{figures/delay_heat_new}{/}[\matches]%
    \StrBefore[\matches]{figures/delay_heat_new}{/}[\datapath]%
    \pgfplotstableread[col sep=comma,]{\datapath/new_results_2/for_heat_delay_new.csv}\datatablea

\begin{tikzpicture}
    \begin{axis}[title={Information Delay},
    view={0}{90},
    colorbar,
    colorbar,
        colorbar style={
        at={(1.02,0)},    % Adjust the position of the colorbar
        anchor=south west,  % Choose the anchor point
        yticklabel pos=right % Position of the tick labels
    },
    colormap={reverse hot2}{
        indices of colormap={
            \pgfplotscolormaplastindexof{hot2},...,0 of hot2}
    },
       colorbar style={
    width=0.3cm,},
    % colormap name=p1,
    point meta min=0,
    point meta max=0.9,
    shader=interp,
    ymin=-1.2, 
    ymax=1.2,
    xmin=-1, 
    xmax=0.1,    
    xlabel={$d_x$},
    ylabel={$d_y$},
    ytick={-1, 1},
    ylabel shift=-15pt,
    xlabel shift=-5pt,
    ]
    \addplot3 [surf] table [x={X}, y={Y}, meta={b}, point meta=explicit, opacity=Ceil{\pgfplotspointmetatransformed}]\datatablea;
    \addplot[color=blue, mark=diamond, mark size=3pt] (-0.5, 0);
    \addplot[color=magenta, mark=*, mark options={fill=white}, only marks, mark size=3pt](0, -1);
    
    \addplot[color=magenta, mark=*, mark options={fill=magenta}, only marks, mark size=3pt](0, 1);
    \end{axis}
\end{tikzpicture}%
    \endgroup%

    % \caption{}
    \end{minipage}%
    \begin{minipage}{.48\linewidth}
    % \hspace{0.1in}
    \centering
    \begingroup%
    \StrCount{figures/adv_heat_new}{/}[\matches]%
    \StrBefore[\matches]{figures/adv_heat_new}{/}[\datapath]%
    \pgfplotstableread[col sep=comma,]{\datapath/new_results_2/for_heat_adv_new.csv}\datatablea

\begin{tikzpicture}
    \begin{axis}[title={Splitting Advantage},
    view={0}{90},
    colorbar,
    colormap={reverse hot2}{
        indices of colormap={
            \pgfplotscolormaplastindexof{hot2},...,0 of hot2}
    },
    colorbar,
        colorbar style={
        at={(1.02,0)},    % Adjust the position of the colorbar
        anchor=south west,  % Choose the anchor point
        yticklabel pos=right,
        width=0.3cm,% Position of the tick labels
    },
    % colormap/hot,
    % colormap name=p1,
    point meta min=0,
    point meta max=0.4,
    shader=interp,
    ymin=-1.2, 
    ymax=1.2,
    xmin=-1, 
    xmax=0.1,    
    xlabel={$d_x$},
    ylabel={$d_y$},
    ytick={-1, 1},
    ylabel shift=-15pt,
    xlabel shift=-5pt,
    ]
    \addplot3 [surf] table [x={X}, y={Y}, meta={a}, point meta=explicit, opacity=Ceil{\pgfplotspointmetatransformed}]\datatablea;
    \addplot[color=blue, mark=diamond, mark size=3pt] (-0.5, 0);
    \addplot[color=magenta, mark=*, mark options={fill=white}, only marks, mark size=3pt](0, -1);
    
    \addplot[color=magenta, mark=*, mark options={fill=magenta}, only marks, mark size=3pt](0, 1);
    \end{axis}
\end{tikzpicture}%
    \endgroup%

    \end{minipage}
    \begin{minipage}{.46\linewidth}
    \centering
    \begingroup%
    \StrCount{figures/traj_delay_new}{/}[\matches]%
    \StrBefore[\matches]{figures/traj_delay_new}{/}[\datapath]%
    \pgfplotstableread[col sep=comma,]{\datapath/new_results_2/cons_inc_1_d_adv_0.csv}\datatablea

\pgfplotstableread[col sep=comma,]{\datapath/new_results_2/cons_inc_1_d_adv_0_0.csv}\datatableb

% \pgfplotstableread[col sep=comma,]{\datapath/new_results_2/cons_inc_1_d_adv_ns.csv}\datatableb

\begin{tikzpicture}
\begin{axis}[,
    xmin=-1,
    xmax=0.1,
    ymin=-1,  % -1.1
    ymax=1.1, 
    point meta=explicit,
    ylabel shift=-25pt,
    xlabel shift=-5pt,
    ytick={-0.5, 1},
    xlabel={$d_x$},
    ylabel={$d_y$},
  %       colormap name = p1,
  %   colorbar,
    point meta = explicit,
    point meta min=0.5, 
    point meta max=1,
]
\addplot[color=red, mark=diamond, mark size=3pt] (-0.9, 0);
\addplot[color=blue, mark=diamond, mark size=3pt] (-0.5, 0);

\addplot [mark=o, mark size=0pt, only marks,
            mark options={draw=red},
            ultra thick, 
                  mesh, % Use line segments instead of one unbroken line
        colormap={}{ % Define the colormap
            color(0cm) = (yellow);
            % color(0.3cm)=(gray);
            color(1cm)=(red);
        },]  table [x={x1}, y={y1}, meta={p}]
\datatablea;
\addplot [mark=o, mark size=0pt, only marks,
            mark options={draw=blue},
            ultra thick, 
                  mesh, % Use line segments instead of one unbroken line
        colormap={}{
         rgb255(0cm)=(238,140,238);
         % rgb255(0.5cm)=(150, 10, 255);
                color(1cm)=(blue);
        },
        ,]  table [x={x2}, y={y2}, meta={p}]
\datatablea;

% \addplot [mark=o, mark size=0 pt, red, ultra thick] table [x={x1}, y={y1}]
% \datatablea;
% \addplot [mark=o, mark size=0 pt, blue, ultra thick] table [x={x2}, y={y2}]
% \datatablea;
\addplot[color=magenta, mark=*, mark options={fill=magenta}, only marks, mark size=2.5pt](0, 1);
% \addplot[color=magenta, mark=*, mark options={fill=white}, only marks, mark size=3pt](0, -1);

\pgfplotstablegetelem{10}{x1}\of{\datatablea}
\pgfmathsetmacro{\a}{\pgfplotsretval}
\pgfplotstablegetelem{10}{y1}\of{\datatablea}
\pgfmathsetmacro{\b}{\pgfplotsretval}
\addplot [red, mark=star, only marks, mark size=2pt, thick] coordinates {(\a, \b)};

\pgfplotstablegetelem{10}{x2}\of{\datatablea}
\pgfmathsetmacro{\a}{\pgfplotsretval}
\pgfplotstablegetelem{10}{y2}\of{\datatablea}
\pgfmathsetmacro{\b}{\pgfplotsretval}
\addplot [blue, mark=star, only marks, mark size=2pt, thick] coordinates {(\a, \b)};

% r2
\addplot [mark=o, mark size=0pt, only marks,
            mark options={draw=red},
            ultra thick, 
                  mesh, % Use line segments instead of one unbroken line
        colormap={}{ % Define the colormap
            color(0cm) = (yellow);
            % color(0.3cm)=(gray);
            color(1cm)=(red);
        },]  table [x={x1}, y={y1}, meta={p}]
\datatableb;
\addplot [mark=o, mark size=0pt, only marks,
            mark options={draw=blue},
            ultra thick, 
                  mesh, % Use line segments instead of one unbroken line
        colormap={}{
         rgb255(0cm)=(238,140,238);
         % rgb255(0.5cm)=(150, 10, 255);
                color(1cm)=(blue);
        },
        ,]  table [x={x2}, y={y2}, meta={p}]
\datatableb;

% \addplot [mark=o, mark size=0 pt, red, ultra thick] table [x={x1}, y={y1}]
% \datatablea;
% \addplot [mark=o, mark size=0 pt, blue, ultra thick] table [x={x2}, y={y2}]
% \datatablea;
\addplot[color=magenta, mark=*, mark options={fill=magenta}, only marks, mark size=2.5pt](0, 1);
% \addplot[color=magenta, mark=*, mark options={fill=white}, only marks, mark size=3pt](0, -1);

\pgfplotstablegetelem{10}{x1}\of{\datatableb}
\pgfmathsetmacro{\a}{\pgfplotsretval}
\pgfplotstablegetelem{10}{y1}\of{\datatableb}
\pgfmathsetmacro{\b}{\pgfplotsretval}
\addplot [red, mark=star, only marks, mark size=2pt, thick] coordinates {(\a, \b)};

\pgfplotstablegetelem{10}{x2}\of{\datatableb}
\pgfmathsetmacro{\a}{\pgfplotsretval}
\pgfplotstablegetelem{10}{y2}\of{\datatableb}
\pgfmathsetmacro{\b}{\pgfplotsretval}
\addplot [blue, mark=star, only marks, mark size=2pt, thick] coordinates {(\a, \b)};

% r3 
% \addplot[color=black, dotted, thick] table [x={x1}, y={y1}]\datatableb;

% \addplot[color=black, dotted, thick] table [x={x2}, y={y2}]\datatableb;

% \pgfplotstablegetelem{10}{x1}\of{\datatableb}
% \pgfmathsetmacro{\a}{\pgfplotsretval}
% \pgfplotstablegetelem{10}{y1}\of{\datatableb}
% \pgfmathsetmacro{\b}{\pgfplotsretval}
% \addplot [black, mark=star, only marks, mark size=2pt] coordinates {(\a, \b)};

% \pgfplotstablegetelem{10}{x2}\of{\datatableb}
% \pgfmathsetmacro{\a}{\pgfplotsretval}
% \pgfplotstablegetelem{10}{y2}\of{\datatableb}
% \pgfmathsetmacro{\b}{\pgfplotsretval}
% \addplot [black, mark=star, only marks, mark size=2pt] coordinates {(\a, \b)};
\end{axis}
\end{tikzpicture}%
    \endgroup%

    % \caption{}
    \end{minipage}%
    \begin{minipage}{.46\linewidth}
    \centering
    \begingroup%
    \StrCount{figures/traj_delay_top_new}{/}[\matches]%
    \StrBefore[\matches]{figures/traj_delay_top_new}{/}[\datapath]%
    \pgfplotstableread[col sep=comma,]{\datapath/new_results_2/cons_inc_1_d_adv_ns_2.csv}\datatablea

% \pgfplotstableread[col sep=comma,]{\datapath/new_results_2/cons_inc_1_d_adv_ns_1.csv}\datatableb

\begin{tikzpicture}
\begin{axis}[,
    xmin=-1,
    xmax=0.1,
    ymin=-1,  % -1.1
    ymax=1.1, 
    point meta=explicit,
    ylabel shift=-25pt,
    xlabel shift=-5pt,
    ytick={-0.5, 1},
    xlabel={$d_x$},
    ylabel={$d_y$},
  %       colormap name = p1,
  %   colorbar,
    point meta = explicit,
    point meta min=0.5, 
    point meta max=1,
]
\addplot[color=red, mark=diamond, mark size=3pt] (-0.7, -0.95);
\addplot[color=blue, mark=diamond, mark size=3pt] (-0.5, 0);

\addplot [mark=o, mark size=0pt, only marks,
            mark options={draw=red},
            ultra thick, 
                  mesh, % Use line segments instead of one unbroken line
        colormap={}{ % Define the colormap
            color(0cm) = (yellow);
            % color(0.3cm)=(gray);
            color(1cm)=(red);
        },]  table [x={x1}, y={y1}, meta={p}]
\datatablea;
\addplot [mark=o, mark size=0pt, only marks,
            mark options={draw=blue},
            ultra thick, 
                  mesh, % Use line segments instead of one unbroken line
        colormap={}{
         rgb255(0cm)=(238,140,238);
         % rgb255(0.5cm)=(150, 10, 255);
                color(1cm)=(blue);
        },
        ,]  table [x={x2}, y={y2}, meta={p}]
\datatablea;

% \addplot [mark=o, mark size=0 pt, red, ultra thick] table [x={x1}, y={y1}]
% \datatablea;
% \addplot [mark=o, mark size=0 pt, blue, ultra thick] table [x={x2}, y={y2}]
% \datatablea;
\addplot[color=magenta, mark=*, mark options={fill=magenta}, only marks, mark size=2.5pt](0, 1);
% \addplot[color=magenta, mark=*, mark options={fill=white}, only marks, mark size=3pt](0, -1);

\pgfplotstablegetelem{10}{x1}\of{\datatablea}
\pgfmathsetmacro{\a}{\pgfplotsretval}
\pgfplotstablegetelem{10}{y1}\of{\datatablea}
\pgfmathsetmacro{\b}{\pgfplotsretval}
\addplot [red, mark=star, only marks, mark size=2pt, thick] coordinates {(\a, \b)};

\pgfplotstablegetelem{10}{x2}\of{\datatablea}
\pgfmathsetmacro{\a}{\pgfplotsretval}
\pgfplotstablegetelem{10}{y2}\of{\datatablea}
\pgfmathsetmacro{\b}{\pgfplotsretval}
\addplot [blue, mark=star, only marks, mark size=2pt, thick] coordinates {(\a, \b)};

\end{axis}
\end{tikzpicture}%
    \endgroup%

    \end{minipage}
        \begin{minipage}{.48\linewidth}
        \hspace{0.3in}
    \centering
    \begingroup%
    \StrCount{figures/only_colorbar_1_adv}{/}[\matches]%
    \StrBefore[\matches]{figures/only_colorbar_1_adv}{/}[\datapath]%
    \includestandalone[]{figures/only_colorbar_1_adv}%
    \endgroup%

    % \caption{}
    \end{minipage}%
    \begin{minipage}{.48\linewidth}
    \hspace{0.2in}
    \centering
    \begingroup%
    \StrCount{figures/only_colorbar_2_adv}{/}[\matches]%
    \StrBefore[\matches]{figures/only_colorbar_2_adv}{/}[\datapath]%
    \includestandalone[]{figures/only_colorbar_2_adv}%
    \endgroup%

    % \caption{}
    \end{minipage}
    \vspace{-0.1in}
    \caption{\textbf{Top:} Average delay ($\mathcal{T}$) in information reveal (left) and average maximum advantage of playing the revealing strategy (right), keeping P2's location fixed at (-0.5, 0) and changing P1's location. \textbf{Bottom:} Trajectory with high delay and advantage (left) and with low delay and advantage (right). Color shades indicate current belief.}
    \vspace{-0.2in}
    \label{fig:delay_adv}
\end{figure}
\textbf{Constrained vs. unconstrained strategies}
% Before solving the constrained game, first, we solve the unconstrained version of the game without any safety specifications. We set $R_A = 5e-2$ and $R_D = 1e-1$.
Fig.~\ref{fig:fig_1} compares strategies with and without the state constraint, visualizing the equilibrium strategies of P1 and the best responses of P2 given P1's strategies. Note that the best responses of P2 give P2 an advantage since she does not know actions to be taken by P1 in reality. 
The analytical solution to the unconstrained game is given by Hexner's analysis, where P1's strategy is to start moving to the goal after the critical time $t_r=0.4$. This strategy no longer holds in the constrained case as it violates the state constraint. Instead, P1 actively tries to stay clear of P2 while pursuing the goal (see 2nd col. of Fig.~\ref{fig:fig_1}). Note that in this case, P1 resorts to a random strategy with the presence of incomplete information and state constraints, as the two contribute to value nonconvexity with respect to the belief.   

\textbf{Information delay and advantage of random strategies}
To understand how P1 uses information asymmetry, we examine the delay in information reveal, measured by the time at which the belief converges to the true type, i.e. $\mathcal{T} = \inf \{k \in [L] : p_k = \{0, 1\}\}$. We then take average of $\mathcal{T}$ for each initial state over 10 simulations. In Fig.~\ref{fig:delay_adv}, we visualize $\mathcal{T}$ over the space of P1's starting positions, while fixing P2's starting position and setting both players' initial speed to 0. The trajectories represent cases where P1 delays (bottom left) and does not delay (bottom right) the release of information. We also compute the advantage of following a belief manipulating strategy (that convexifies the value) as opposed to taking the non-revealing strategy (i.e. never split), expressed as $[\min_u \max_v V(t+\tau, x', p) - V(t, x, p)]$. Overall, P1 tends to conceal and deceive when it has equal distances to the possible targets. 

\vspace{-0.1in}
\paragraph{Equilibrium strategy of Player 2} Fig.~\ref{fig:primal_dual} shows sample trajectories when both players play their equilibrium strategies. Note that compared with P2's best responses to P1 in previous examples, P2's equilibrium strategy is more conservative, due to her lack of knowledge about P1's type. We also note that P2's dual game is one dimension higher than P1's primal game, and thus encounters higher numerical errors in value and strategy approximation (see Appendix~\ref{sec:case_details}). 

\begin{figure}[!ht]
\vspace{-0.1in}
\begin{minipage}{.48\linewidth}
    % \hspace{0.1in}
    \centering
    \begingroup%
    \StrCount{figures/traj_primal_dual_1_new}{/}[\matches]%
    \StrBefore[\matches]{figures/traj_primal_dual_1_new}{/}[\datapath]%
    \pgfplotstableread[col sep=comma,]{\datapath/new_results_2/primal_dual_near_0.csv}\datapa
\pgfplotstableread[col sep=comma,]{\datapath/new_results_2/primal_dual_near_1.csv}\datapb
\pgfplotstableread[col sep=comma,]{\datapath/new_results_2/primal_dual_near_2.csv}\datapc
\pgfplotstableread[col sep=comma,]{\datapath/new_results_2/primal_dual_near_3.csv}\datapd
\pgfplotstableread[col sep=comma,]{\datapath/new_results_2/primal_dual_near_4.csv}\datape

\begin{tikzpicture}
\begin{axis}[title={},
    xmin=-0.65,
    xmax=0.4,
    ymin=-0.7,  % -1.1
    ymax=1.1,  % 0.5
    point meta=explicit,
    ylabel shift=-15pt,
    xlabel shift=-5pt,
    xlabel={$d_x$},
    ylabel={$d_y$},
  %       colormap name = p1,
  %   colorbar,
    point meta min=0, 
    point meta max=1,
  %   colorbar horizontal,
  %       colorbar style={
  %     at={(0.75,-0.1)},
  %     anchor=center,
  %     width=4.5cm
  %   },
  %   colorbar/draw/.append code={
  % \begin{axis}[
  %   % colormap/p2/.style={
  %   % colormap name=p2,
  %   % }, 
  %   colormap name=p2,
  %   colorbar horizontal,
  %   point meta min=0,
  %   point meta max=1,
  %   every colorbar,
  %     at={(0.25,-0.1)},
  %     anchor=center,
  %     width=4.5cm,
  %   colorbar shift,
  %   colorbar=false,
  %   ]
  %       \pgfkeysvalueof{/pgfplots/colorbar addplot}
  %     \end{axis}
  %   }]
]

\addplot [mark=o, mark size=0pt, only marks,
            mark options={draw=red},
            ultra thick, 
                  mesh, % Use line segments instead of one unbroken line
        colormap={}{ % Define the colormap
            color(0cm) = (yellow);
            % color(0.3cm)=(gray);
            color(1cm)=(red);
        },]  table [x={x1}, y={y1}, meta={c}]
\datapd;
\addplot [mark=o, mark size=0pt, only marks,
            mark options={draw=blue},
            ultra thick, 
                  mesh, % Use line segments instead of one unbroken line
        colormap={}{
         rgb255(0cm)=(238,140,238);
         % rgb255(0.5cm)=(150, 10, 255);
                color(1cm)=(blue);
        },
        ,]  table [x={x2}, y={y2}, meta={c}]
\datapd;
\pgfplotstablegetelem{10}{x1}\of{\datapd}
\pgfmathsetmacro{\a}{\pgfplotsretval}
\pgfplotstablegetelem{10}{y1}\of{\datapd}
\pgfmathsetmacro{\b}{\pgfplotsretval}
\addplot [red, mark=star, only marks, mark size=2pt, thick] coordinates {(\a, \b)};

\pgfplotstablegetelem{10}{x2}\of{\datapd}
\pgfmathsetmacro{\a}{\pgfplotsretval}
\pgfplotstablegetelem{10}{y2}\of{\datapd}
\pgfmathsetmacro{\b}{\pgfplotsretval}
\addplot [blue, mark=star, only marks, mark size=2pt, thick] coordinates {(\a, \b)};

%%% 5th trajectory%%%
\addplot [mark=o, mark size=0pt, only marks,
            mark options={draw=red},
            ultra thick, 
                  mesh, % Use line segments instead of one unbroken line
        colormap={}{ % Define the colormap
            color(0cm) = (yellow);
            % color(0.3cm)=(gray);
            color(1cm)=(red);
        },]  table [x={x1}, y={y1}, meta={c}]
\datape;
\addplot [mark=o, mark size=0pt, only marks,
            mark options={draw=blue},
            ultra thick, 
                  mesh, % Use line segments instead of one unbroken line
        colormap={}{
         rgb255(0cm)=(238,140,238);
         % rgb255(0.5cm)=(150, 10, 255);
                color(1cm)=(blue);
        },
        ,]  table [x={x2}, y={y2}, meta={c}]
\datape;
\pgfplotstablegetelem{10}{x1}\of{\datape}
\pgfmathsetmacro{\a}{\pgfplotsretval}
\pgfplotstablegetelem{10}{y1}\of{\datape}
\pgfmathsetmacro{\b}{\pgfplotsretval}
\addplot [red, mark=star, only marks, mark size=2pt, thick] coordinates {(\a, \b)};

\pgfplotstablegetelem{10}{x2}\of{\datape}
\pgfmathsetmacro{\a}{\pgfplotsretval}
\pgfplotstablegetelem{10}{y2}\of{\datape}
\pgfmathsetmacro{\b}{\pgfplotsretval}
\addplot [blue, mark=star, only marks, mark size=2pt, thick] coordinates {(\a, \b)};

\end{axis}
\begin{axis}[xtick=\empty, ytick=\empty, minor tick num=0, scaled ticks=false, xticklabel=\empty,yticklabel=\empty,
    xmin=-0.65,
    xmax=0.4,
    ymin=-0.7,  % -1.1
    ymax=1.1, ]
    % 0.5
    % legend entries = {Attacker, Defender}, 
    %     legend style={at={(0.02,0.85)},anchor=south west},
    % legend image post style={thick}]
\addplot[color=red, mark=diamond, mark size=3pt] (-0.6, -0.4);
\addplot[color=blue, mark=diamond, mark size=3pt] (-0.5, 0);
\addplot[color=magenta, mark=*, mark options={fill=magenta}, only marks, mark size=3pt](0, 1);
% \addplot [mark=o, mark size=0.5pt, blue, ultra thick] table [x={x2}, y={y2}]
% {\datatablea};
\end{axis}
% \begin{axis}[title={Incomplete Information, Unconstrained},
%     xmin=-0.55,
%     xmax=0.1,
%     ymin=-1.1,
%     ymax=0.5,
%     point meta=explicit,
%     ylabel shift=-15pt,
%     xlabel shift=-5pt,
%     legend entries={Attacker, Defender},
%     xlabel={$x_1$},
%     ylabel={$y_1$},
% ]
% \addplot [mark=o, mark size=0.5pt, red, ultra thick] table [x={x1}, y={y1}]
% \datatablea;
% \addplot [mark=o, mark size=0.5pt, blue, ultra thick] table [x={x2}, y={y2}]
% \datatablea;

% \addplot[color=red, mark=diamond*, mark size=3pt] (-0.5, 0);
% \addplot[color=blue, mark=diamond*, mark size=3pt] (-0.4, -0.08);
% \addplot[color=magenta, mark=*, mark options={fill=magenta}, only marks, mark size=3pt](0, -1);
% \end{axis}
\end{tikzpicture}%
    \endgroup%

    % \caption{}
    \end{minipage}%
    \begin{minipage}{.48\linewidth}
    % \hspace{0.1in}
    \centering
    \begingroup%
    \StrCount{figures/traj_primal_dual_2_new}{/}[\matches]%
    \StrBefore[\matches]{figures/traj_primal_dual_2_new}{/}[\datapath]%
    \pgfplotstableread[col sep=comma,]{\datapath/new_results_2/primal_dual_1_0.csv}\datatablea
\begin{tikzpicture}[]
\begin{axis}[title={},
    xmin=-0.6,
    xmax=0.65,
    ymin=-0.7,  % -1.1
    ymax=1.1, 
    point meta=explicit,
    ylabel shift=-15pt,
    xlabel shift=-5pt,
    % legend entries={Attacker, Defender},
    % legend image post style={thick},
    xlabel={$d_x$},
    ylabel={$d_y$},
            colormap name = p1,
    % colorbar,
    point meta min=0, 
    point meta max=1,
    % colorbar horizontal,
    %     colorbar style={
    %   at={(0.75,-0.1)},
    %   anchor=center,
    %   width=4.5cm
    % },
  %   colorbar/draw/.append code={
  % \begin{axis}[
  %   % colormap/p2/.style={
  %   % colormap name=p2,
  %   % }, 
  %   colormap name=p2,
  %   colorbar horizontal,
  %   point meta min=0,
  %   point meta max=1,
  %   every colorbar,
  %     at={(0.25,-0.1)},
  %     anchor=center,
  %     width=4.5cm,
  %   colorbar shift,
  %   colorbar=false,
  %   ]
  %       \pgfkeysvalueof{/pgfplots/colorbar addplot}
  %     \end{axis}
  %   }]
]
\addplot[color=red, mark=diamond, mark size=3pt] (-0.5, 0);
\addplot[color=blue, mark=diamond, mark size=3pt] (0.5, 0);
\addplot [mark=o, mark size=0 pt, red, ultra thick] table [x={x1}, y={y1}]
\datatablea;
\addplot [mark=o, mark size=0 pt, blue, ultra thick] table [x={x2}, y={y2}]
\datatablea;
\addplot[color=magenta, mark=*, mark options={fill=magenta}, only marks, mark size=3pt](0, 1);
\addplot[color=magenta, mark=*, mark options={fill=white}, only marks, mark size=3pt](0, -1);

\pgfplotstablegetelem{10}{x1}\of{\datatablea}
\pgfmathsetmacro{\a}{\pgfplotsretval}
\pgfplotstablegetelem{10}{y1}\of{\datatablea}
\pgfmathsetmacro{\b}{\pgfplotsretval}
\addplot [red, mark=star, only marks, mark size=2pt, thick] coordinates {(\a, \b)};

\pgfplotstablegetelem{10}{x2}\of{\datatablea}
\pgfmathsetmacro{\a}{\pgfplotsretval}
\pgfplotstablegetelem{10}{y2}\of{\datatablea}
\pgfmathsetmacro{\b}{\pgfplotsretval}
\addplot [blue, mark=star, only marks, mark size=2pt, thick] coordinates {(\a, \b)};
\end{axis}
\end{tikzpicture}%
    \endgroup%

    \end{minipage}
    \vspace{-0.1in}
    \caption{Trajectories where both players use their respective behavioral strategies. P1 keeps track of $p$, whereas P2 keeps track of $\hat{p}$.}
    \label{fig:primal_dual}
    \vspace{-0.1in}
\end{figure}
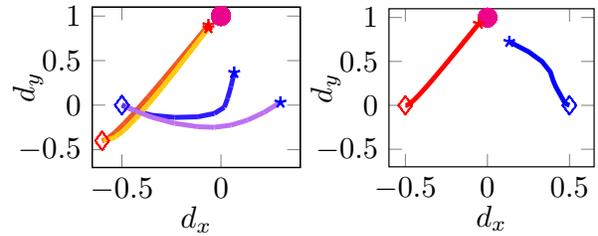
\vspace{-0.05in}
\section{Conclusion and Future Work}
\label{sec:conclusion}
\vspace{-0.05in}
We proved the existence of value for zero-sum differential games with state constraints and one-sided information and developed a backward induction scheme to approximate the value. Our method enables mechanistic synthesis of behavioral strategies and allows explanation of the resultant splitting of strategies and beliefs.
% We provide a framework that allows learning the value function across both time and space via the use of a PICNN which ensures that the value function is convex in belief, a consequence of information asymmetry. 
% Finally, we showcase the algorithm to generate strategies in a simplified football game, where the attacker has private information about its type, which is unknown to the defender. 
Future work will investigate more efficient learning+search methods that take advantage of value convexity and alleviate error propagation.
% data generation process due to the need for discretization of the action space. Future works will explore Reinforcement Learning methods with continuous action spaces to facilitate efficient learning. 

% A comparison to a recent work handling information asymmetry via reachability analysis in both physical and belief space \cite{deceptiongame} is forthcoming. 

% Following the current work, the immediate focus will be on acquiring the conjugate value for the uninformed player, with the overarching goal of conducting a thorough analysis of the game from the perspective of both players.

%===============================================================================

%===============================================================================
\section*{Acknowledgements}
This work was in part supported by NSF CMMI-1925403 and NSF CNS-2101052. The views and conclusions contained in this document are those of the authors and should not be interpreted as representing the official policies, either expressed or implied, of the National Science Foundation or the U.S. Government.

\section*{Impact Statement}
This paper is concerned with advancing the field of differential game theory and artificial intelligence. The developed theories and tools help understand how players play strategies by taking into account their information advantage and disadvantage. Similar to the development of imperfect-information games, protocols to mitigate the potential societal consequences and risks (e.g., deception by robots and machines) shall be comprehensively investigated. 
\bibliography{ref}
\bibliographystyle{icml2024}

\appendix
\onecolumn

\section{Proof of Theorem \ref{theorem:existence}}\label{sec:proofexist}
The following proofs extend results and techniques from Cardaliaguet~\yrcite{cardaliaguet2007differential} to zero-sum differential games with one-sided information and state constraints. 
To overview, we start by showing that the upper and lower values $V^{\pm}$ are Lipschitz continuous within the safe and unsafe state sets (Lemma~\ref{lemma:regularity}) and convex with respect to $p$ (Lemma~\ref{lemma:convex}). We then show that: (1) $V^{-*}$ satisfies a subdynamic principle and is therefore a subsolution of a dual HJ equation and hence $V^-$ is a dual supersolution of the corresponding primal HJ (Lemma~\ref{lemma:dualsubsolution}, Lemma~\ref{lemma:reformulation}, Lemma~\ref{lemma:subdynamic}); and (2) $V^+$ also satisfies a subdynamic principle (Lemma~\ref{lemma:subdynamic2}) and is therefore a dual subsolution of the primal HJ (Lemma~\ref{lemma:primalsubsolution}). We can then use a comparison principle (see \cite{cardaliaguet2007differential}) to show that since $V^-$ is a dual supersolution and $V^+$ is a dual subsolution of the primal HJ while both share the same terminal value, $V^- \geq V^+$. On the other hand, $V^- \leq V^+$ by definition and hence $V^- = V^+$.  

% Define time-dependent unsafe sets as
%     \begin{equation}
%         \bar{\mathcal{X}}^t := \{x \in \mathbb{R}^{d_x} ~|~ \forall \eta \in \mathcal{H}_r(t), \exists \zeta \in \mathcal{Z}_r(t), ~s.t.~ \rho(\mathcal{X}^{t,x,\eta, \zeta}_T) = +\infty\},
%     \end{equation}
% and let the safe sets be $\mathcal{X}^t := \mathbb{R}^{d_x}/\bar{\mathcal{X}}^t$. Our main contribution is the extension of the subdynamic principles to games where $\bar{\mathcal{X}}^t \neq \emptyset$.

We start with the following regularity result 
(see proof of Lemma 3.1 in \cite{cardaliaguet2007differential}):
\begin{lemma}
\label{lemma:regularity}
    (regularity of $V^{\pm}$). 
    $V^{\pm}(t_0, x_0, p)$ are Lipschitz continuous for all $x_0 \in \mathcal{Q}(t_0)$. $V^{\pm}(t_0, x_0, p) = +\infty$ for all $x_0 \in \bar{\mathcal{Q}}(t_0)$.
\end{lemma}

The following convexity result was originally developed for repeated games with incomplete information~\cite{de1996repeated} and was later extended to differential games~\cite{cardaliaguet2007differential}. The same convexity result holds for imperfect-information dynamic games~\cite{brown2020combining}. 
% To reduce notational burden, we introduce $G_i(t,x,\eta,\zeta) := g_i(X^{t,x,\eta, \zeta}_T)\delta(t,x,\eta,\zeta)$.

\begin{lemma}
\label{lemma:convex}
    (convexity property of $V^{\pm}$). For any $(t, x) \in [0, T] \times \mathbb{R}^{d_x}$, $V^{\pm}$ are convex in $p$ on $\Delta(I)$.
\end{lemma}
\begin{proof}
    Let $p^{\lambda} = (1-\lambda)p^0 + \lambda p^1$ for some $p^0$, $p^1 \in \Delta(I)$. Let $((\eta_i^0), \zeta^0)$ and $((\eta_i^1), \zeta^1)$ be the equilibrial strategies for $V(t,x,p^0)$ and $V(t,x,p^1)$, respectively. Introduce a set of ``splitting'' behavioral strategies $(\eta_i^{\lambda})$ for $(t, x, p^{\lambda})$ such that for any type $i$, $\eta_i^{\lambda} = \eta_i^0$ with probability $(1-\lambda)p_i^0/p_i^{\lambda}$ and $\eta_i^{\lambda} = \eta_i^1$ with probability $\lambda p_i^1/p_i^{\lambda}$. Then we have
    \begin{equation}
    \begin{aligned}
        & \sup_{\zeta} \sum_i p_i^{\lambda} 
        % \mathbb{E}_{\eta_i^{\lambda}, \zeta}\left[
        G_i(t,x,\eta_i^{\lambda},\zeta)
        % \right]
        \\
        =& \sup_{\zeta} \sum_i \left(p_i^{\lambda}\frac{(1-\lambda)p_i^0}{p_i^{\lambda}}
        % \mathbb{E}_{\eta_i^{0}, \zeta}\left[
        G_i(t,x,\eta_i^0,\zeta)
        % \right]
        + p_i^{\lambda}\frac{\lambda p_i^1}{p_i^{\lambda}}
        % \mathbb{E}_{\eta_i^{1}, \zeta}\left[
        G_i(t,x,\eta_i^1,\zeta)
        % \right]
        \right)\\
        \leq & (1-\lambda) \sup_{\zeta} \sum_i p_i^0 
        % \mathbb{E}_{\eta_i^{0}, \zeta}\left[
        G_i(t,x,\eta_i^{0},\zeta)
        % \right]
        + \lambda \sup_{\zeta} \sum_i p_i^1
        % \mathbb{E}_{\eta_i^{1}, \zeta}\left[
        G_i(t,x,\eta_i^{1},\zeta)
        % \right]
        .
    \end{aligned}
    \label{eq:proof_convexity}
    \end{equation}
    Since the inequality in Eq.~\eqref{eq:proof_convexity} holds for any ``splitting'' $(\eta_i^{\lambda})$, we have
    \begin{equation}
        V^{\pm}(t,x,p^\lambda) \leq (1-\lambda)V^{\pm}(t,x,p^0) + \lambda V^{\pm}(t,x,p^1)
    \end{equation}
    for any $t \in [0,T]$ and $x \in \mathcal{Q}(t)$. For $x \in \bar{\mathcal{Q}}(t)$, the equality holds since $V^{\pm}(t,x,\cdot) = +\infty$.
\end{proof}

\textbf{Remarks.} 
(1) The proof says that by playing a ``splitting'' strategy, the value at $(t,x,p^\lambda)$ should at least be as good as a linear interpolation between those at $(t,x,p^0)$ and at $(t,x,p^1)$. Hence the value is a convex hull in $\Delta(I)$ at any $(t,x) \in [0,T] \times \mathbb{R}^{d_x}$. (2) Assuming that the ``splitting'' strategy of Player 1 is known by Player 2, then the latter can perform Bayesian inference on Player 1's type based on his actions. For any type $i$, let $u_i^0$ and $u_i^1$ be two distinct actions to be taken at $(t,x)$ following $\eta_i^0$ and $\eta_i^1$, respectively, and let $p^\lambda$ be the current common belief. Then under observation of $u_i^0$ (resp. $u_i^1$), the common belief becomes $p^0$ (resp. $p^1$) with probability $(1-\lambda)$ (resp. $\lambda$). Since $p^\lambda = (1-\lambda)p^0 + \lambda p^1$, common belief is a martingale.

% We note that in the case with $I$ Player 1 types, the proof can be extended by introducing $p^j \in \Delta (I)$, $\lambda^j \in \Delta (I)$, and $u^j \in \mathcal{U}$ for $j = 1,...,I$ such that $p^\lambda = \sum_j p^j \lambda^j$, and $u^j$ being optimal in
% \begin{equation}
%     = \argmin_{u \in \mathcal{U}} 
% \end{equation}
% . 

Next, we introduce a reformulation of $V^{-*}$ to facilitate the derivation of its subdynamic principle. The proof of this reformulation is extended from Lemma 4.1 of Cardaliaguet~\yrcite{cardaliaguet2007differential} to incorporate the state constraints.
\begin{lemma}
\label{lemma:reformulation}
    (reformulation of $V^{-*}$). We have
    \begin{equation}
        V^{-*}(t_0, x_0, \hat{p}) = 
        \inf_{\zeta \in \mathcal{Z}_r(t_0)} \sup_{\eta \in \mathcal{H}_r(t_0)} 
        % \max_{i\in \{1,...,I\}} \left\{ \hat{p}_i - \min_{\tau \in [t, T]} \max\left\{C(t,x,\tau, \zeta, (\eta_i)),~ \textbf{E}_{\eta_i, \zeta}\left(g_i(X_{\tau}^{t,x,\eta_i,\zeta})\right) - z\right\} \right\}
        \max_i \left\{ \hat{p}_i - 
        % \mathbb{E}_{\eta,\zeta}\left [
        G_i(t_0,x_0,\eta,\zeta)
        % \right]
        \right\}
        \label{eq:reformulation}
    \end{equation}
\end{lemma}

\begin{proof}
    For later use, we first note that
    \begin{equation}
    \begin{aligned}
        V^-(t_0, x_0, p) & = \sup_{\zeta} \inf_{(\eta_i)} \sum_i p_i 
        % \mathbb{E}_{\eta,\zeta}\left [
        G_i(t_0,x_0,\eta,\zeta)
        % \right] 
        \\
        & = \sup_{\zeta} \sum_i p_i \inf_{\eta} 
        % \mathbb{E}_{\eta,\zeta}\left [
        G_i(t_0,x_0,\eta,\zeta)
        % \right]
        .
    \end{aligned}
    \end{equation}
    Let the right-hand side of Eq.~\eqref{eq:reformulation} be $z=z(t_0,x_0,\hat{p})$. We can show that $z$ is convex with respect to $\hat{p}$ using a technique similar to that of Lemma~\ref{lemma:convex}. 
    % Also see Lemma 4.1 of \cite{cardaliaguet2007differential}.
    
    Then by the definition of $z$:
    \begin{equation}
    \begin{aligned}
        z^*(t_0,x_0,p) & = \sup_{\hat{p}} p^T \hat{p} - \inf_{\zeta} \sup_{\eta} \max_i \left \{  \hat{p}_i - 
        % \mathbb{E}_{\eta,\zeta}\left [
        G_i(t_0,x_0,\eta,\zeta) 
        % \right] 
        \right\} \\
        & = \sup_{\hat{p}} p^T \hat{p} - \inf_{\zeta} \max_i \left \{  \hat{p}_i - \inf_{\eta}  
        % \mathbb{E}_{\eta,\zeta}\left [
        G_i(t_0,x_0,\eta,\zeta) 
        % \right]
        \right\} \\
        & = \sup_{\zeta}\sup_{\hat{p}} \min_i \left \{ p^T \hat{p} -     \hat{p}_i + \inf_{\eta} 
        % \mathbb{E}_{\eta,\zeta}\left [
        G_i(t_0,x_0,\eta,\zeta) 
        % \right]
        \right\}. \\
    \end{aligned}
    \end{equation}
    In this last expression, $\sup_{\hat{p}}$ is attained by setting $\hat{p}_i = \inf_{\eta}  
    % \mathbb{E}_{\eta,\zeta}\left [
    G_i(t_0,x_0,\eta,\zeta) 
    % \right]
    $, in which case we have
    \begin{equation}
        z^*(t_0,x_0,p) = \sup_{\zeta} \sum_i p_i \inf_{\eta} 
        % \mathbb{E}_{\eta,\zeta}\left [
        G_i(t_0,x_0,\eta,\zeta) 
        % \right] 
        = V^-(t_0,x_0,p,z).
    \end{equation}
    Since $z$ is convex with respect to $\hat{p}$, we have $V^{-*} = z^{**} = z$.
\end{proof}

Next, to introduce the subdynamic principle of $V^{-*}$, we first introduce
\begin{equation}
    U^{-*}(t_0,x_0,\hat{p}) := \inf_{\zeta \in \mathcal{Z}_r(t_0)} \sup_{\eta \in \mathcal{H}_r(t_0)} \max_i \left\{ \hat{p}_i - \mathbb{E}_{\eta,\zeta}\left [g_i(X^{t_0,x_0,\eta,\zeta}_{T})\right] \right\}
\end{equation}
as the conjugate lower value of the \textit{unconstrained} version of the game, and $U^{\pm}$ as the corresponding upper and lower values. From Lemma~\ref{lemma:regularity} and Lemma~\ref{lemma:reformulation}, $U^{-*}$ is Lipschitz continuous and convex in $\hat{p}$.  

\begin{lemma}
\label{lemma:subdynamic}
    (subdynamic principle for $V^{-*}$). For any $(t_0, x_0, \hat{p}) \in [0,T) \times \mathbb{R}^{d_x} \times \mathbb{R}^I$ and any $t_1 \in (t_0, T]$, denote $x_1 = X^{t_0,x_0,\eta,\zeta}_{t_1}$ and $\mathcal{X}_1 = \mathcal{X}^{t_0,x_0,\eta,\zeta}_{t_1}$. We have
    \begin{equation}
    \begin{aligned}
        V^{-*}(t_0,x_0,\hat{p}) \leq \inf_{\zeta \in \mathcal{Z}(t_0)} \sup_{\eta \in \mathcal{H}(t_0)}  \min & \left\{ \rho(\mathcal{X}_1) U^{-*}\left(t_1,x_1,\frac{\hat{p}}{\rho(\mathcal{X}_1)}\right), V^{-*}\left(t_1,x_1,\hat{p}\right)\right\}
    \end{aligned}
        \label{eq:subdynamic}
    \end{equation}
\end{lemma}

\begin{proof}
    Denote $V_1^{-*}(t_0, t_1, x_0, \hat{p}):= \inf_{\zeta \in \mathcal{Z}(t_0)}\sup_{\eta \in \mathcal{H}(t_0)} V^{-*}\left(t_1,X^{t_0,x_0,\eta,\zeta}_{t_1},\hat{p}\right)$. $U_1^{-*}$ is similarly defined. We need the following preparations for the proof.
    % Following the arguments in Theorem~\ref{lemma:regularity}, $V_1^{-*}$ is Lipschitz continuous. 
    
    \textbf{Player 1 plays a pure strategy in $V^{-*}$.} We show below that best responses are always pure. In particular, Player 1 can play in pure strategies in $V^{-*}$, namely,
    \begin{equation}
        V^{-*}(t_0,x_0,\hat{p}) = \inf_{\zeta \in \mathcal{Z}_r(t_0)} \sup_{\eta \in \mathcal{H}(t_0)} \max_i \left\{ \hat{p}_i - 
        % \mathbb{E}_{\zeta}\left [
        G_i(t_0,x_0,\eta,\zeta)
        % \right] 
        \right\}
    \end{equation}
    for any $(t_0,x_0,\hat{p})$. First from Theorem \ref{lemma:reformulation} and using $\mathcal{H}(t) \subset \mathcal{H}_r(t)$, we have
    \begin{equation}
    \begin{aligned}
        V^{-*}(t_0, x_0, \hat{p}) & = 
        \inf_{\zeta \in \mathcal{Z}_r(t_0)} \sup_{\eta \in \mathcal{H}_r(t_0)} 
        \max_i \left\{ \hat{p}_i - 
        % \mathbb{E}_{\eta,\zeta}\left [
        G_i(t_0,x_0,\eta,\zeta)
        % \right]
        \right\} \\
        & \geq \inf_{\zeta \in \mathcal{Z}_r(t_0)} \sup_{\eta \in \mathcal{H}(t_0)} \max_i \left\{ \hat{p}_i - \mathbb{E}_{\zeta}\left [g_i(X^{t_0,x_0,\eta,\zeta}_T)\rho(\mathcal{X}^{t_0,x_0,\eta,\zeta}_T)\right]\right\}.
        \label{eq:subdynamic1}
    \end{aligned}
    \end{equation}
    For the reverse inequality, we first note that for any $\eta \in \mathcal{H}_r(t_0)$ and $\omega_1 \in \Omega_{\eta}$, $\eta(\omega_1, \cdot) \in \mathcal{H}(t_0)$. With a fixed $\zeta \in \mathcal{Z}_r(t_0)$, and by using the convexity of $\max_i$ (i.e., $\max_i \mathbb{E}_{\omega}[f_i(\omega)] \leq \mathbb{E}_{\omega}[\max_i f_i(\omega)]$), we have
    % \begin{equation}
    %     \begin{aligned}
    %         &\sup_{(\eta_i) \in (\mathcal{H}_r(t))^I} \max_i \left\{ \hat{p}_i - \max_i C(t,x,\eta_i,\zeta) \right\} \\
    %         = &\sup_{(\eta_i) \in (\mathcal{H}_r(t))^I} \max_i \left\{ \hat{p}_i - \max_i \max_{\omega \in \Omega_{\eta}} \max_{\tau \in [t,T]} c(X^{t,x,\eta_i(\omega,\cdot),\zeta}_{\tau}) \right\} \\
    %         \leq &\sup_{\eta \in \mathcal{H}(t)} \max_i \left\{ \hat{p}_i - C(t,x,\eta,\zeta) \right\}
    %     \end{aligned}
    %     \label{eq:3.5a}
    % \end{equation}
    % We also have
    \begin{equation}
        \begin{aligned}
            &\sup_{\eta \in \mathcal{H}_r(t_0)} \max_i \left\{ \hat{p}_i - 
            % \mathbb{E}_{\eta,\zeta}\left [
            G_i(t_0,x_0,\eta,\zeta)
            % \right] 
            \right\} \\
            \leq & \sup_{\eta \in \mathcal{H}_r(t_0)} \int_{\Omega_{\eta}} \max_i \left\{ \hat{p}_i - \mathbb{E}_{\zeta}\left [
            g_i(X^{t_0,x_0,\eta(\omega_1,\cdot),\zeta}_T)\rho(\mathcal{X}^{t_0,x_0,\eta(\omega_1,\cdot),\zeta}_T)\right]
            % G_i(t_0,x_0,\eta(\omega_1,\cdot),\zeta)\right]
            \right\}d\textbf{P}_{\eta}(\omega_1) \\
            \leq &\sup_{\eta \in \mathcal{H}_r(t_0)} \sup_{\omega_1 \in \Omega_{\eta}} 
            \max_i \left\{ \hat{p}_i - \mathbb{E}_{\zeta}\left [g_i(X^{t_0,x_0,\eta(\omega_1,\cdot),\zeta}_T)\rho(\mathcal{X}^{t_0,x_0,\eta(\omega_1,\cdot),\zeta}_T)\right]\right\}\\
            % G_i(t_0,x_0,\eta(\omega_1,\cdot),\zeta)\right]
            \leq & \sup_{\eta \in \mathcal{H}(t_0)} \max_i \left\{ \hat{p}_i - \mathbb{E}_{\zeta}\left [
            g_i(X^{t_0,x_0,\eta(\omega_1,\cdot),\zeta}_T)\rho(\mathcal{X}^{t_0,x_0,\eta(\omega_1,\cdot),\zeta}_T)\right]
            % G_i(t_0,x_0,\eta,\zeta)\right]
            \right\}.
        \end{aligned}
        \label{eq:3.5b}
    \end{equation}
    % Similarly
    % \begin{equation}
    %     \sup_{\eta \in \mathcal{H}_r(t)} \max_i \left\{ \hat{p}_i - G_i(t,x,\eta,\zeta) + z \right\} \leq \sup_{\eta \in \mathcal{H}(t)} \max_i \left\{ \hat{p}_i - G_i(t,x,\eta,\zeta) + z \right\}
    %     \label{eq:3.5b}
    % \end{equation}
    Since Eq.~\eqref{eq:3.5b} holds for any $\zeta$, together with Eq.~\eqref{eq:subdynamic1}, we have
    % \begin{equation}
    %     V^{-*}(t_0,x_0,\hat{p},z) \leq \inf_{\zeta \in \mathcal{Z}_r(t_0)} \sup_{\eta \in \mathcal{H}(t_0)} \max_i \left\{ \hat{p}_i - \max\left\{C, ~\textbf{E}_{\zeta} \left[ g_i\left(X^{t_0,x_0,\eta,\zeta}_{T}\right)\right] - z\right\}\right\}.
    % \end{equation}
    % Hence
    \begin{equation}
    \begin{aligned}
         V^{-*}(t,x,\hat{p}) & = \inf_{\zeta \in \mathcal{Z}_r(t_0)} \sup_{\eta \in \mathcal{H}(t_0)} \max_i \left\{ \hat{p}_i -  \mathbb{E}_{\zeta}\left [
         g_i(X^{t_0,x_0,\eta,\zeta}_T)\rho(\mathcal{X}^{t_0,x_0,\eta,\zeta}_T)
         % G_i(t_0,x_0,\eta,\zeta)
         \right]\right\}.
    \end{aligned}
    \end{equation}
    % where 
    % \begin{equation}
    %     V^{-*}_c(t,x,\hat{p},z):=\inf_{\zeta \in \mathcal{Z}_r(t)} \sup_{\eta \in \mathcal{H}(t)} \max_i \hat{p}_i - \max_i C_i(t,x,\alpha,\zeta)
    % \end{equation}
    % and 
    % \begin{equation}
    %     V^{-*}_g(t,x,\hat{p},z):= \inf_{\zeta \in \mathcal{Z}_r(t)} \sup_{\eta \in \mathcal{H}(t)} \max_i \hat{p}_i - G_i(t,x,\alpha,\zeta) + z
    % \end{equation}
    % are introduced for later use.
    Note that one can reach the same conclusion for $U^{-*}$. 
    
    \textbf{$\epsilon$-optimal strategy of Player 2.} Let $\epsilon > 0$ and $\zeta^0 \in \mathcal{Z}(t_0)$ be some pure $\epsilon$-optimal strategy for $V_1^{-*}(t_0,t_1,x,\hat{p})$. For any $x_1 \in \mathbb{R}^{d_x}$, we can find some $\epsilon$-optimal strategy $\zeta^{x_1} \in \mathcal{Z}_r(t_1)$ for Player 2 in the game $V^{-*}(t_1, x_1, \hat{p})$. 
    Let $B_{\rho}(x)$ be a ball around $x$ with radius $\rho$, and let $\partial \mathcal{Q}(t)$ be the boundary of $\mathcal{Q}(t)$, i.e., for any $x \in \partial \mathcal{Q}(t)$ and $\rho>0$, there exist $y \in B_{\rho}(x) \bigcap \mathcal{Q}(t)$ and $y' \in B_{\rho}(x) \bigcap \bar{\mathcal{Q}}(t)$.
    
    For $x_1, y \in \mathcal{Q}(t_1) 
\setminus \partial \mathcal{Q}(t_1)$, from Lipschitz continuity of the map $y \rightarrow V^{-*}(t_1, y, \hat{p})$,
    % \begin{equation}
    %     y \rightarrow \sup_{\alpha \in \mathcal{A}(t)}\max_i \left\{ \right\}
    % \end{equation}
    $\zeta^{x_1}$ is also $(2\epsilon)$-optimal for $V^{-*}(t_1, y, \hat{p})$ if $y \in B_{\rho}(x_1)$ for some radius $\rho >0$.
    The same applies to $x_1, y \in \bar{\mathcal{Q}}(t_1)$ and $y \in B_{\rho}(x_1)$ since $V^{-*}(t_1, x, \hat{p}) = +\infty$ is constant for $x \in \bar{\mathcal{Q}}(t_1)$.
    
    Since the dynamics $f$ is bounded, we also know that the reachable states from $(t_0,x_0)$ is bounded in some ball $B_{R}(0)$. Let us set $M = \|f\|_{\infty}$ and some small $\sigma > 0$ such that $M\sigma \leq \rho/2$. Then we choose $(x_l)_{l=1,...,l_0}$ such that $\bigcup_{l=1}^{l_0} B_{\rho/2}(x_l)$ contains $B_R(0)$. Let $(E_l)_{l=1,...,l_0}$ be a Borel partition of $B_R(0)$ such that, for any $l$, $E_l \subset B_{\rho/2}(x_l)$. We also require $(x_l)$ to be chosen properly so that $E_l \subset \mathcal{Q}(t_1)$ or $E_l \subset \bar{\mathcal{Q}}(t_1)$.
    
    We set
    \begin{equation}
        \zeta^l = \zeta^{x_l},~ \Omega^l = \Omega_{\zeta^{x_l}},~ \mathcal{F}^l = \mathcal{F}_{\zeta^{x_l}},\text{ and } \textbf{P}^l = \textbf{P}_{\zeta^{x_l}}
    \end{equation}
    for $l = 1,...,l_0$. We choose some delay $\tau \in (0, \sigma]$ for all the strategies $\zeta^l$. Note that if for some open-loop control $(\alpha,\delta) \in \mathcal{A}(t_0) \times \mathcal{D}(t_0)$ and for some $l$, we have $X_{t_1-\tau}^{t_0,x_0,\alpha,\delta} \in E_l$, then
    \begin{equation}
        |X_{t_1-\tau}^{t_0,x_0,\alpha,\delta} - X_{t_1}^{t_0,x_0,\alpha,\delta}| \leq \|f\|_{\infty}\tau \leq M\sigma \leq \rho/2,
    \end{equation}
    so that $X_{t_1}^{t_0,x_0,\alpha,\delta}$ belongs to $B_{\rho}(x_l)$. Hence $\zeta^l$ is $(2\epsilon)$-optimal for $V^{-*}$ at $(t_1,X_{t_1}^{t_0,x_0,\alpha,\delta},\hat{p})$.

    Let us now define a new strategy $\zeta \in \mathcal{Z}_r(t_0)$ by setting
    \begin{equation}
        \Omega_{\zeta} = \prod_{l=1}^{l_0}\Omega^l,~\mathcal{F}_{\zeta} = \mathcal{F}^1 \otimes ...\otimes \mathcal{F}^{l_0},\text{ and } \textbf{P}_{\zeta} = \textbf{P}^1 \otimes ...\otimes \textbf{P}^{l_0}.
    \end{equation}
    For any $\omega = (\omega^1,...,\omega^{l_0}) \in \Omega_{\zeta}$ and $\alpha \in \mathcal{A}(t_0)$, set
    \begin{equation}
        \zeta(\omega, \alpha) = \left\{ \begin{array}{ll}
            \zeta^0(\alpha)(\tau) & \text{if } \tau \in [t_0,t_1) \\
            \zeta^l(\omega^l,\alpha)(\tau) & \text{if } \tau \in [t_1, T] \text{ and }  X_{t_1-\tau}^{t_0,x_0,\alpha,\zeta^0} \in E_l.
        \end{array}\right.
    \end{equation}
    For any pure strategy $\eta \in \mathcal{H}(t_0)$, we have
    \begin{equation}
        g_i(X^{t_0,x_0,\eta,\zeta}_T) = \sum_{l=1}^{l_0} g_i\left( X^{t_1,X^{t_0,x_0,\eta,\zeta^0}_{t_1},\eta,\zeta^l}_T \right)\textbf{1}_{O^l},
    \end{equation}
    where $O^l := \left\{X^{t_0,x_0,\eta,\zeta^0}_{t_1-\tau} \in E_l\right\}$.
    % , $C(t_0,t_1,x_0,\eta,\zeta^0)$ by $C_0$, and $C\left(t_1,T,X^{t_0,x_0,\eta,\zeta^0}_{t_1},\eta,\zeta^l \right)$ by $C_l$, 

    \textbf{Property of $\rho(\cdot)$.} Let
    % \begin{equation}
    %     \delta_0 := \delta\left(\mathcal{X}^{t_0,x_0,\eta,\zeta^0}_{t_1}\right),
    % \end{equation}
    $        \rho_0 := \rho\left(\mathcal{X}^{t_0,x_0,\eta,\zeta^0}_{t_1}\right)$ and 
    % \begin{equation}
    %     \delta_1^l :=  \delta\left(\mathcal{X}^{t_1,x_1^l,\eta,\zeta^l}_{T}\right),
    % \end{equation}
    $\rho_1^l :=  \rho\left(\mathcal{X}^{t_1,x_1^l,\eta,\zeta^l}_{T}\right)$,
    where $x_1^l := X^{t_0,x_0,\eta,\zeta^0}_{t_1}$ if the state falls in $E_l$ following pure strategies $(\eta, \zeta^0)$.
    Then we have
    \begin{equation}
    \begin{aligned}
        g_i\left(X^{t_0,x_0,\eta,\zeta}_T\right) \rho\left(\mathcal{X}^{t_0,x_0,\eta,\zeta}_{T}\right) & = g_i(X^{t_0,x_0,\eta,\zeta}_T) \sum_l \max\{\rho_0, \rho_1^l\} \textbf{1}_{O^l} \\
        & = \sum_l g_i(X^{t_1,x_1^l,\eta,\zeta}_T) \textbf{1}_{O^l} \sum_l \max\{\rho_0, \rho_1^l\} \textbf{1}_{O^l} \\
        & = \sum_l g_i(X^{t_1,x_1^l,\eta,\zeta}_T) \max\{\rho_0, \rho_1^l\} \textbf{1}_{O^l}.
    \end{aligned}
    \end{equation}
    For any set $(\omega^l) \in (\Omega^l)$, let $a^l := g_i(X^{t_1,x_1^l,\eta,\zeta(\omega^l,\cdot)}_T) \geq 0$. Also note that $\rho_0$ and $\rho^l_1$ only take values from $\{1, +\infty\}$. 
    % Without loss of generality, let $\delta_1^1,...,\delta_1^{l_1} = 1$ and $\delta_1^{l_1+1},...,\delta_1^{l_0} = K$. If $\delta_0 = 1$, we have 
    One can show that the following always holds:
    \begin{equation}
         \sum_l a^l \max\{\rho_0, \rho_1^l\} \textbf{1}_{O^l} = 
         \max\left\{ \sum_l a^l \rho_0 \textbf{1}_{O^l}, \sum_l a^l \rho_1^l \textbf{1}_{O^l}\right\}.
    \end{equation}
    % \begin{equation}
    %      \sum_l a^l \max\{\delta_0, \delta_1^l\} \textbf{1}_{O^l} = 
    %      \sum_l a^l \delta_1^l \textbf{1}_{O^l} = \max\left\{ \sum_l a^l \delta_0 \textbf{1}_{O^l}, \sum_l a^l \delta_1^l \textbf{1}_{O^l}\right\}.
    % \end{equation}
    % If $\delta_0 = K$, then
    % \begin{equation}
    %      \sum_l a^l \max\{\delta_0, \delta_1^l\} \textbf{1}_{O^l} = 
    %      \sum_l a^l \delta_0 \textbf{1}_{O^l} = \max\left\{ \sum_l a^l \delta_0 \textbf{1}_{O^l}, \sum_l a^l \delta_1^l \textbf{1}_{O^l}\right\}.
    % \end{equation}
    % Therefore we always have 
    Similarly we have
    \begin{equation}
    \begin{aligned}
         & \int_{\Omega^l} \sum_l \left [ g_i\left(X^{t_1,X^{t_0,x_0,\eta,\zeta^0}_{t_1},\eta,\zeta^l(\omega^l, \cdot)}_{T}\right) \max\{\rho_0, \rho_1^l \}\textbf{1}_{O^l}\right] d\textbf{P}^l(\omega^l)\\
        = & \max\left\{ 
        \begin{array}{l}
             \int_{\Omega^l} \sum_l \left [ g_i\left(X^{t_1,X^{t_0,x_0,\eta,\zeta^0}_{t_1},\eta,\zeta^l(\omega^l, \cdot)}_{T}\right) \rho_0 \textbf{1}_{O^l}\right] d\textbf{P}^l(\omega^l),  \\\\
             \int_{\Omega^l} \sum_l \left [ g_i\left(X^{t_1,X^{t_0,x_0,\eta,\zeta^0}_{t_1},\eta,\zeta^l(\omega^l, \cdot)}_{T}\right) \rho_1^l \textbf{1}_{O^l}\right] d\textbf{P}^l(\omega^l) 
        \end{array}
          \right\}    
    \end{aligned}
    \end{equation}
    
    Now we derive an upper bound of $\max_i \left\{ \hat{p}_i - \mathbb{E}_{\zeta}\left [g_i(X^{t_0,x_0,\eta,\zeta}_{T}) \rho(\mathcal{X}^{t_0,x_0,\eta,\zeta}_{T})\right]\right\}$:
    \begin{equation}
    \begin{aligned}
        % & \max_i \left\{ \hat{p}_i - \mathbb{E}_{\zeta}\left [G_i(t_0,x_0,\eta,\zeta)\right]\right\} \\
        % = 
        & \max_i \left\{ \hat{p}_i -  \mathbb{E}_{\zeta}\left [g_i(X^{t_0,x_0,\eta,\zeta}_{T}) \rho(\mathcal{X}^{t_0,x_0,\eta,\zeta}_{T}) \right] \right\},\\
        % X^{t_0,x_0,\eta,\zeta^0}_{t_1}
        = & \max_i \left\{\hat{p}_i - \int_{\Omega^l} \sum_l \left [ g_i\left(X^{t_1,X^{t_0,x_0,\eta,\zeta^0}_{t_1},\eta,\zeta^l(\omega^l, \cdot)}_{T}\right) \max\{\rho_0, \rho_1^l \}\textbf{1}_{O^l}\right] d\textbf{P}^l(\omega^l)\right\},\\  
        \leq & \sup_{\eta' \in \mathcal{H}(t_1)} \max_i \left\{\hat{p}_i - \sum_l \left [ \int_{\Omega^l} g_i\left(X^{t_1,X^{t_0,x_0,\eta,\zeta^0}_{t_1},\eta',\zeta^l(\omega^l, \cdot)}_{T}\right) \max\{\rho_0, \rho_1^l \}d\textbf{P}^l(\omega^l) \textbf{1}_{O^l}\right] \right\},\\ 
        = & \sup_{\eta' \in \mathcal{H}(t_1)} \max_i \left\{\min \left\{ 
        \begin{array}{l}
             \hat{p}_i - \sum_l \left [ \int_{\Omega^l} g_i\left(X^{t_1,X^{t_0,x_0,\eta,\zeta^0}_{t_1},\eta',\zeta^l(\omega^l, \cdot)}_{T}\right) \rho_0 d\textbf{P}^l(\omega^l) \textbf{1}_{O^l}\right],  \\\\
             \hat{p}_i - \sum_l \left [ \int_{\Omega^l} g_i\left(X^{t_1,X^{t_0,x_0,\eta,\zeta^0}_{t_1},\eta',\zeta^l(\omega^l, \cdot)}_{T}\right) \rho_1^l d\textbf{P}^l(\omega^l) \textbf{1}_{O^l}\right]
        \end{array} 
        \right\} \right\},\\ 
        % & \text{(because of the convexity of $\max_i$)} \\
        \leq & \min \left\{ \begin{array}{l}
             \sup_{\eta' \in \mathcal{H}(t_1)} \max_i \left\{ \hat{p}_i - \sum_l \left [ \int_{\Omega^l} g_i\left(X^{t_1,X^{t_0,x_0,\eta,\zeta^0}_{t_1},\eta',\zeta^l(\omega^l, \cdot)}_{T}\right) \rho_0 d\textbf{P}^l(\omega^l) \textbf{1}_{O^l}\right] \right\},   \\\\
             \sup_{\eta' \in \mathcal{H}(t_1)} \max_i \left\{\hat{p}_i - \sum_l \left [ \int_{\Omega^l} g_i\left(X^{t_1,X^{t_0,x_0,\eta,\zeta^0}_{t_1},\eta',\zeta^l(\omega^l, \cdot)}_{T}\right) \rho_1^l d\textbf{P}^l(\omega^l) \textbf{1}_{O^l}\right]  \right\}
        \end{array} 
         \right\},\\ 
        % & \text{(because $\zeta^l$ is $2\epsilon$-optimal for $V^{-*}$ at $(t_1,X_{t_1}^{t_0,x_0,\eta,\zeta^0},\hat{p},z)$)}\\
        \leq & \min \left\{ \begin{array}{l}
             \sum_l \left [\rho_0 \sup_{\eta' \in \mathcal{H}(t_1)}  \max_i \left\{ \hat{p}_i/\rho_0 -  \int_{\Omega^l} g_i\left(X^{t_1,X^{t_0,x_0,\eta,\zeta^0}_{t_1},\eta',\zeta^l(\omega^l, \cdot)}_{T}\right)  d\textbf{P}^l(\omega^l) \right\}\textbf{1}_{O^l}\right],   \\\\
             \sum_l \left [ \sup_{\eta' \in \mathcal{H}(t_1)} \max_i \left\{\hat{p}_i - \int_{\Omega^l} g_i\left(X^{t_1,X^{t_0,x_0,\eta,\zeta^0}_{t_1},\eta',\zeta^l(\omega^l, \cdot)}_{T}\right) \rho_1^l d\textbf{P}^l(\omega^l) \right\} \textbf{1}_{O^l}\right]
        \end{array} 
         \right\}.\\ 
        % \leq & \min\left\{\delta_0 U_1^{-*}(t_0,t_1,x_0,\hat{p}/\delta_0),  V^{-*}_1\left(t_0, t_1, x_0, \hat{p}\right) + 3 \epsilon \right\}.\\
        % & \text{(because $\zeta^0$ is $\epsilon$-optimal for $V^{-*}_1\left(t_0, t_1, x_0, \hat{p}, z\right)$)}
    \end{aligned}
    \end{equation}
    % We also have
    % \begin{equation}
    %     \begin{aligned}
    %         \max_i \{\hat{p}_i - \max_i C_i(t,x,\eta,\zeta)\} & = \max_i\left\{ \hat{p}_i - \max_i \max_{y \in \mathcal{X}^{t,x,\eta,\zeta}} c(y) \right\} \\
    %         & = \max_i \left\{ \hat{p}_i - \max\left\{ \sum_{l=1}^{l_0} \max_{\omega^l} \max_{\tau \in [t_1, T]} c(X^{t_1, X^{t,x,\eta,\zeta^0}_{t_1}, \eta, \zeta^l})\textbf{1}_{O^l}, \max_{s \in [t,t_1]} c(X^{t,x,\eta,\zeta^0})_s \right\}\right\} \\
    %         & \leq \max_i \left\{ \hat{p}_i - \sum_{l=1}^{l_0} \max_{\omega^l} \max_{\tau \in [t_1, T]} c(X^{t_1, X^{t,x,\eta,\zeta^0}_{t_1}, \eta, \zeta^l})\textbf{1}_{O^l}\right\} \\
    %         & \leq \sum_{l=1}^{l_0} \sup_{\eta' \in \mathcal{H}(t_1)} \max_i \left\{ \hat{p}_i -  \max_{\tau \in [t_1, T]} c(X^{t_1, X^{t,x,\eta,\zeta^0}_{t_1}, \eta', \zeta^l})\right\}\textbf{1}_{O^l}
    %     \end{aligned}
    % \end{equation}

    Shorten the first and second terms in the above upper bound $\min\{\cdot, \cdot\}$ as $A$ and $B$, respectively. For $B$: 
    \begin{equation}
    \begin{aligned}
        & \sum_l \left [ \sup_{\eta' \in \mathcal{H}(t_1)} \max_i \left\{\hat{p}_i - \int_{\Omega^l} g_i\left(X^{t_1,X^{t_0,x_0,\eta,\zeta^0}_{t_1},\eta',\zeta^l(\omega^l, \cdot)}_{T}\right) \rho_1^l d\textbf{P}^l(\omega^l) \right\} \textbf{1}_{O^l}\right] \\
        \leq & \sum_l \left ( V^{-*}(t_1,X^{t_0,x_0,\eta,\zeta^0}_{t_1},\hat{p}) + 2\epsilon \right) \textbf{1}_{O^l} \\
        & \text{(because $\zeta^l$ is ($2\epsilon$)-optimal for $V^{-*}$ at $(t_1, x_1, \hat{p})$ for any $x_1 \in E_l$.)}\\
        = & V^{-*}(t_1,X^{t_0,x_0,\eta,\zeta^0}_{t_1},\hat{p}) + 2\epsilon \\
        \leq & V^{-*}_1(t_0,t_1,x_0,\hat{p}) + 3\epsilon\\
        & \text{(because $\zeta^0$ is $\epsilon$-optimal for $V^{-*}_1(t_0,t_1,x_0,\hat{p})$.)}\\
    \end{aligned}
    \end{equation}

    For $A$, we consider the following scenarios: (1) If $\inf_{\eta'\in \mathcal{H}(t_1)}\sum_l \rho_1^l \textbf{1}_{O^l} = +\infty$, i.e., there is always a chance for Player 2 to achieve constraint violation when the game starts at $(t_1, X^{t_0,x_0,\eta,\zeta^0}_{t_1})$, and $\rho_0 = 1$, i.e., $(\eta,\zeta^0)$ does not induce constraint violation, then
    \begin{equation}
        \begin{aligned}
            A = & \sum_l \left [\sup_{\eta'\in \mathcal{H}(t_1)} \max_i \left\{ \hat{p}_i -  \inf_{\eta' \in \mathcal{H}(t_1)} \int_{\Omega^l} g_i\left(X^{t_1,X^{t_0,x_0,\eta,\zeta^0}_{t_1},\eta',\zeta^l(\omega^l, \cdot)}_{T}\right)  d\textbf{P}^l(\omega^l) \right\}\textbf{1}_{O^l}\right] \\
            > & \sum_l \left [\sup_{\eta'\in \mathcal{H}(t_1)} \max_i \left\{ \hat{p}_i -  \inf_{\eta' \in \mathcal{H}(t_1)} \int_{\Omega^l} g_i\left(X^{t_1,X^{t_0,x_0,\eta,\zeta^0}_{t_1},\eta',\zeta^l(\omega^l, \cdot)}_{T}\right)\rho_1^l  d\textbf{P}^l(\omega^l) \right\}\textbf{1}_{O^l}\right] \\
            = & B = -\infty.
        \end{aligned}
    \end{equation}
    
    If $\inf_{\eta'\in \mathcal{H}(t_1)}\sum_l \rho_1^l \textbf{1}_{O^l} = \rho_0 = 1$, then $A = B$. This is the scenario where the game reduces to its unconstrained version. Applying the same analysis from $B$ to have $A \leq \rho_0 (U_1^{-*}(t_0,t_1,x_0,\hat{p}/\rho_0) + 3\epsilon)$. If $\inf_{\eta'\in \mathcal{H}(t_1)}\sum_l \rho_1^l \textbf{1}_{O^l} = \rho_0 = +\infty$, then $A = B = -\infty$.

    If $\inf_{\eta'\in \mathcal{H}(t_1)}\sum_l \rho_1^l \textbf{1}_{O^l} = 1$ and $\rho_0 = +\infty$, the game starting from $(t_1, X^{t_0,x_0,\eta,\zeta^0}_{t_1})$ will be played as an unconstrained one. $A < B$. Hence $A \leq \rho_0 (U_1^{-*}(t_0,t_1,x_0,\hat{p}/\rho_0) + 3\epsilon) = -\infty$.

    Combining these scenarios we have
    \begin{equation}
        \min\{A,B\} \leq \min\{\rho_0 (U_1^{-*}(t_0,t_1,x_0,\hat{p}/\rho_0) + 3\epsilon), V^{-*}_1(t_0,t_1,x_0,\hat{p}) + 3\epsilon \}.
    \end{equation}
    
    Since $\epsilon$ can be arbitrarily small, we have
    \begin{equation}
    \begin{aligned}
        V^{-*}\left(t_0, x_0, \hat{p}\right) & \leq 
        \min\left\{ \rho_0 U_1^{-*}(t_0,t_1,x_0,\hat{p}/\rho_0), V^{-*}_1(t_0,t_1,x_0,\hat{p}) \right\} \\
        & = \inf_{\zeta \in \mathcal{Z}(t_0)} \min\left\{ \sup_{\eta \in \mathcal{H}(t_0)} \rho_0 U^{-*}(t_1,X^{t_0,x_0,\eta,\zeta}_{t_1},\hat{p}/\rho_0), \sup_{\eta \in \mathcal{H}(t_0)} V^{-*}(t_1,X^{t_0,x_0,\eta,\zeta}_{t_1},\hat{p}) \right\} \\
    \end{aligned}
    \label{eq:subdynamic_final}
    \end{equation}
Lastly, if $x_0 \in \bar{\mathcal{Q}}(t_0)$, $\inf_{\zeta}\sup_{\eta}\rho_0 = -\infty$ by definition; otherwise, $\rho_0 = 1$ and $U^{-*} = V^{-*}$ at $t_1$. In both cases, the RHS of Eq.~\eqref{eq:subdynamic_final} becomes
\begin{equation}
    \inf_{\zeta \in \mathcal{Z}(t_0)}\sup_{\eta \in \mathcal{H}(t_0)} \min\left\{  \rho_0 U^{-*}(t_0,X^{t_0,x_0,\eta,\zeta}_{t_1},\hat{p}/\rho_0), V^{-*}(t_0,X^{t_0,x_0,\eta,\zeta}_{t_1},\hat{p}) \right\}.
\end{equation}

\end{proof}

\begin{theorem}
\label{lemma:dualsubsolution}
    ($V^{-*}$ is a subsolution of HJ). For any $\hat{p} \in \mathbb{R}^{I}$, the map $(t, x) \rightarrow V^{-*}(t, x, \hat{p})$ is a viscosity subsolution of the dual Hamilton-Jacobi equation:
    \begin{equation}
        \min\left\{\rho(x)U^{-*}(t,x,\hat{p}/\rho(x)) - w,  ~w_t + H^*(x, Dw) \right\} = 0 \text{ in } [0,T] \times \mathbb{R}^{d_x},
    \end{equation}
    where $H$ is defined by Eq.~\eqref{eq:hamiltonian} and $H^*(x, \xi) = - H (x, -\xi)$. $\bar{\rho}(t_0,x_0) = 1$ if $x_0 \in \mathcal{Q}(t_0)$. 
    % and $\bar{\rho}(t_0,x_0) = +\infty$ otherwise.
\end{theorem}

\begin{proof}
    Let $\hat{p} \in \mathbb{R}^I$ be fixed, and let $\phi$ be a smooth test function such that
    \begin{equation}
        \phi(t, x) \geq V^{-*}(t, x, \hat{p}) \quad \forall (t, x) \in [0, T] \times \mathbb{R}^{d_x},
    \end{equation}
    with an equality at $(t_0, x_0)$, where $t_0 \in [0, T)$. For any $v \in \mathcal{V}$, define a pure strategy $\zeta \in \mathcal{Z}(t_0)$ by setting 
    \begin{equation}
        \zeta(\alpha)(t) = v \quad \forall \alpha \in \mathcal{A}(t_0),~t \in [t_0, T].
    \end{equation}
    Since $V^{-*}$ satisfies the subdynamic programming principle of Lemma~\ref{lemma:subdynamic}, these exist $\epsilon>0$, $h>0$, and a pure strategy $\eta_h \in \mathcal{H}(t_0)$ such that
    \begin{equation}
    \begin{aligned}
        V^{-*}(t_0, x_0, \hat{p}) \leq \min & \left\{\rho\left(\mathcal{X}^{t_0,x_0,\eta_h,\zeta}_{t_0+h}\right) U^{-*}\left(t_0+h,X^{t_0, x_0, \eta_h, \zeta}_{t_0+h},\hat{p}/\rho\left(\mathcal{X}^{t_0,x_0,\eta_h,\zeta}_{t_0+h}\right)\right), \right.\\
        & ~~ \left. V^{-*}\left(t_0+h, X^{t_0, x_0, \eta_h, \zeta}_{t_0+h}, \hat{p}\right) \right\} + \epsilon h,
    \end{aligned}
    \end{equation}
    or equivalently
    \begin{equation}
    \rho\left(\mathcal{X}^{t_0,x_0,\eta_h,\zeta}_{t_0+h}\right) U^{-*}\left(t_0+h,X^{t_0, x_0, \eta_h, \zeta}_{t_0+h},\hat{p}/\rho\left(\mathcal{X}^{t_0,x_0,\eta_h,\zeta}_{t_0+h}\right)\right) - V^{-*}\left(t_0, x_0, \hat{p}\right) + \epsilon h \geq 0,
    \label{eq:dualsub1}
    \end{equation}
    and
    \begin{equation}
    V^{-*}\left(t_0+h, X^{t_0, x_0, \eta_h, \zeta}_{t_0+h}, \hat{p}\right) - V^{-*}(t_0, x_0, \hat{p}) + \epsilon h \geq 0.
    \label{eq:dualsub2}
    \end{equation}

    Set the open-loop control $\alpha_h(s) := \eta_h(v)(s)$ and the trajectory $x_h(s) = X^{t_0, x_0, \eta_h, \beta}_s = X^{t_0, x_0, \alpha_h, v}_s$. Then
    \begin{equation}
        x_h(t_0 + h) = x_0 + \int_{t_0}^{t_0 + h} f(x_h(s), \alpha_h(s), v) ds = x_0 + \int_{t_0}^{t_0 + h} f(x_0, \alpha_h(s), v) ds + h\epsilon(h),
    \end{equation}
    where $\epsilon(h) \rightarrow 0$ as $h \rightarrow 0^+$. 
    For Eq.~\eqref{eq:dualsub1}, let $\epsilon \rightarrow 0^+$ and $h \rightarrow 0^+$, we have
    \begin{equation}
        \rho(x_0)U^{-*}(t_0,x_0,\hat{p}/ \rho(x_0)) - \phi(t_0, x_0) \geq 0.
    \end{equation}
    
    For Eq.~\eqref{eq:dualsub2}, we have
    \begin{equation}
    \begin{aligned}
        0 & \leq V^{-*}\left(t_0+h, X^{t_0, x_0, \eta_h, \zeta}_{t_0+h}, \hat{p}\right) - V^{-*}(t_0, x_0, \hat{p}) + \epsilon h \\  
        & \leq \phi\left(t_0+h, x_0 + \int_{t_0}^{t_0 + h} f(x_0, \alpha_h(s), v) ds + h\epsilon(h), z\right) - \phi(t_0, x_0) + \epsilon h \\
        & \leq h \phi_t(t_0, x_0) + \int_{t_0}^{t_0 + h} D\phi(t_0, x_0)^T f(x_0, \alpha_h(s), v)ds + h\epsilon_1(h) + \epsilon h \\
        & \leq h \phi_t(t_0, x_0) + h \sup_{u \in \mathcal{U}} D\phi(t_0, x_0)^T f(x_0, u, v) + h\epsilon_1(h) + \epsilon h,
    \end{aligned}
    \end{equation}
    where $\epsilon_1(h) \rightarrow 0$ as $h \rightarrow 0^+$. Dividing the last inequality by $h$, letting $h \rightarrow 0^+$, $\epsilon \rightarrow 0^+$, and taking the infimum over $v \in \mathcal{V}$ to have
    \begin{equation}
        \phi_t(t_0, x_0) + \inf_{v \in \mathcal{V}} \sup_{u \in \mathcal{U}} D\phi(t_0, x_0)^T f(x_0, u, v) \geq 0.
    \end{equation}
    Now notice that by definition
    \begin{equation}
        H^*(x, D\phi) = -H(x, -D\phi) =  \inf_{v \in \mathcal{V}} \sup_{u \in \mathcal{U}} f(x, u, v)^T D\phi.
    \end{equation}
    Hence
    \begin{equation}
        \phi_t(t_0, x_0) + H^*(x_0, D\phi(t_0, x_0)) \geq 0.
    \end{equation}
    
    Hence
    \begin{equation}
        \min\left\{\rho(x)U^{-*}(t,x,\hat{p}/\rho(x)) - \phi,  ~\phi_t + H^*(x, D\phi) \right\} \geq 0 \text{ in } [0,T] \times \mathbb{R}^{d_x}.
    \end{equation}
\end{proof}

Using the same proof techniques we can derive the subdynamic principle for $V^+$ (Lemma~\ref{lemma:subdynamic2}) and prove that $V^+$ is a viscosity subsolution to the primal HJ (Lemma~\ref{lemma:primalsubsolution}):

\begin{lemma}
    \label{lemma:subdynamic2}
    (subdynamic principle for $V^{+}$). We have for any $(t_0, x_0, p) \in [0,T) \times \mathbb{R}^{d_x} \times \Delta(I)$ and any $t_1 \in (t_0, T]$
    \begin{equation}
    \begin{aligned}
        V^{+}(t_0,x_0,p) \leq \inf_{\eta \in \mathcal{H}(t_0)} \sup_{\zeta \in \mathcal{Z}(t_0)} \max & \left\{ \rho(\mathcal{X}^{t_0,x_0,\eta,\zeta}_{t_1}) U^{+}\left(t_1,X^{t_0,x_0,\eta,\zeta}_{t_1},p\right), \right. \\
        & ~~\left.V^{+}\left(t_1,X^{t_0,x_0,\eta,\zeta}_{t_1},p\right)\right\}
    \end{aligned}
        \label{eq:subdynamic2}
    \end{equation}
\end{lemma}

\begin{lemma}
\label{lemma:primalsubsolution}
    ($V^{+}$ is a subsolution of HJ). For any $p \in \Delta(I)$, the map $(t, x) \rightarrow V^{+}(t, x, p)$ is a viscosity subsolution of the primal Hamilton-Jacobi equation:
    \begin{equation}
        \min\left\{\rho(x)U^{+}(t,x,p) - w,  ~w_t + H(x, Dw) \right\} = 0 \text{ in } [0,T] \times \mathbb{R}^{d_x},
    \end{equation}
    where $H$ is defined by Eq.~\eqref{eq:hamiltonian}.
\end{lemma}

%%%%%%%%%%%%%%%%%%%%%%%%%% END PROOF THM 1 %%%%%%%%%%%%%%%%%%%%%%%%%%%%%%%%%%%%%%%

%%%%%%%%%%%%%%%%%%%%%%%%%%%% BEGIN PROOF THM 2 %%%%%%%%%%%%%%%%%%%%%%%%%%%%%%%%%%%%%
\section{Proof of Theorem \ref{thm:value_convergence}}
\label{sec:value_convergence}
\begin{proof}
The following proof follows that of \cite{cardaliaguet2009numerical}, with the additional treatment of the state constraint. Note that in numerical approximation, we use $K>0$ to replace infinite values so that we can introduce bounded test functions. The HJ equations thus become
\begin{equation*}
    \Biggl\{\begin{array}{ll}
        w_t + H(x, Dw) = 0 & (t,x) \in \Omega  \\
        \min\{K - w, w_t + H(x, Dw)\} = 0 & (t,x) \in \bar{\Omega}, 
    \end{array}
\end{equation*}
where $\Omega$ (resp. $\bar{\Omega}$) contains all $(t,x)$ such that $V(t,x) < K$ (resp. $V(t,x) = K$).

Consider $w$ be any cluster point in the topology of uniform convergence on compact subsets of $[0, T] \times \mathbb{R}^{d_x} \times \Delta (I)$ of $V_\tau$ as $\tau \rightarrow 0^+$. $w$ is convex with respect to $p$ and satisfies:
\begin{equation}\label{eq:bc}
    w(T, x, p) = \sum_{i=1}^I p_ig_i(x),
\end{equation}
for any $(T,x, p) \in \Omega \times \Delta(I)$ and $w(T,x,p) = K$ for any $(T,x, p) \in \bar{\Omega}\times \Delta(I)$.
% From Remark \eqref{rem:backup}, 
Let $\phi$ be a test function such that $w(\cdot, \cdot, p) - \phi$ has a strict local maximum at $(t_0, x_0)$, and $w(t_0, x_0, p) = \phi(t_0, x_0)$. Then there are $(t_k, x_k)$ converging to $(t_0, x_0)$ such that $V_\tau (\cdot, \cdot, p) - \phi$ has a local maximum at $(t_k, x_k)$.

% it is sufficient to show that $V_\tau$ converges to $V$ uniformly $\forall x \in \mathcal{C}$. For any $x \notin \mathcal{C}$, $V_\tau = V = \infty$. We can rewrite the above equation as:
% \begin{equation}
% w(T, x, p) = \sum_{i=1}^I p_i g_i(x)\quad \forall (x, p) \in \mathcal{C} \times \Delta(I).
% \end{equation}
% Using similar arguments in ~\cite{cardaliaguet2009numerical}, 

First consider $(t_k,x_k) \in \Omega$. For any $x \in \mathbb{R}^{d_x}$,
\begin{equation*}
    V_\tau (t_{k+1}, x, p) - \phi(t_{k+1}, x) \le V_\tau (t_k, x_k, p) - \phi(t_k, x_k)
\end{equation*}
Then, rearranging (\ref{eq:backup_cons}) to have
\begin{align*}
    0 &= \text{Vex}_p\left(\min_u \max_v V_\tau(t_{k+1}, x_k + \tau f(x_k, u, v), p)\right) - V_\tau (t_k, x_k, p)\\
    &\le \min_u \max_v V_\tau(t_{k+1}, x_k + \tau f(x_k, u, v), p) - V_\tau (t_k, x_k, p)\\
    &\le \min_u \max_v \phi(t_{k+1}, x_k + \tau f(x_k, u, v)) - \phi (t_k, x_k)
\end{align*}
Then, from standard arguments (see \cite{cardaliaguet2009numerical} and references therein)
\begin{equation}
    \frac{\partial \phi}{\partial t}(t_0, x_0) + \min_{u \in U} \max_{v \in V} f(x_0, u, v) \frac{\partial \phi}{\partial x}(t_0, x_0) \ge 0.
\end{equation}

Now consider $(t_k, x_k) \in \bar{\Omega}$, in which case $\phi(t_k,x_k) = V_{\tau}(t_k, x_k, p) = K$. When $\tau \rightarrow 0^+$,
\begin{equation*}
    \min_u \max_v \phi(t_{k+1}, x_k + \tau f(x_k, u, v)) = K,
\end{equation*}
hence
\begin{equation*}
     \frac{\partial \phi}{\partial t}(t_0, x_0) + \min_{u \in U} \max_{v \in V} f(x_0, u, v) \frac{\partial \phi}{\partial x}(t_0, x_0) = 0.
\end{equation*}

% From the definition of $\Omega$, $\bar{U}(t_k,x_k,p) = V(t_k,x_k,p)$. Since $V_{\tau}$ uniformly converges to $w$, $V_{\tau}(t_k,x_k,p) = V(t_k,x_k,p)$. Therefore $\phi(t_k,x_k) = V_{\tau}(t_k,x_k,p) = \bar{U}(t_k,x_k,p)$. 
Then we have
\begin{equation*}
    \min\left\{K - \phi,  ~\phi_t + H(x, D\phi) \right\} = 0 \text{ in } \bar{\Omega}.
\end{equation*}

Hence, $w$ is a dual subsolution of the HJI. 
We can follow the same technique to show that $w$ is a supersolution in the dual sense, and therefore $w = V$.
% Following similar arguments above and in \cite{cardaliaguet2009numerical}, it can be easily proved. 
\end{proof}
% \section{Additional Results}

%%%%%%%%%%%%%%%%%%%%%%%%%% END PROOF THM 2 %%%%%%%%%%%%%%%%%%%%%%%%%%%%%%%%%%%%%%%

%%%%%%%%%%%%%%%% BEGIN BELLMAN BACKUP %%%%%%%%%%%%%%%%%%%%%%%%%%
\section{Bellman Backup of the Conjugate Value with the Presence of Instantaneous Loss}\label{sec:modify}

Here we extend the subdynamic principle for $V^{-*}$ (Lemma~\ref{lemma:subdynamic}) when instantaneous loss is present. Since we will use letter $l$ to index possible states reached at $t_1$ from $t_0$, we denote the instantaneous loss for pure strategies $(\eta,\zeta)$ at time $s$ as $L(\eta,\zeta,s)$ instead. For conciseness, let us consider $t_0 \in [0,T]$ and $(x_0,\hat{p}) \in \mathcal{Q}(t_0) \times \mathbb{R}^{I}$, i.e., states for which Player 1 can play to avoid state constraint violation. 
To recap, let $\eta$ be a pure strategy of Player 1; let $\zeta$ be such that $\zeta = \zeta^0$ in $[t_0,t_1]$ where $\zeta^0$ is pure and $\epsilon$-optimal for $V_1^{-*}(t_0,t_1,x_0,\hat{p}) := := \inf_{\zeta \in \mathcal{Z}(t_0)}\sup_{\eta \in \mathcal{H}(t_0)} V^{-*}\left(t_1,X^{t_0,x_0,\eta,\zeta}_{t_1},\hat{p} - \int_{t_0}^{t_1} L(\eta, \zeta,s)ds \right)$, and $\zeta = \zeta^l$ in $[t_1,T]$ where $\zeta^l$ is mixed and $(2\epsilon)$-optimal for $V^{-*}$ at $(t_1,x_1,\hat{p})$ for any $x_1 \in E_l$.

By definition and using the fact that Player 1 plays a pure strategy in the dual game, we have
\begin{equation}
    V^{-*}(t_0,x_0,\hat{p}) = \inf_{\zeta} \sup_{\eta} \max_i \left\{
    \hat{p}_i - \int_{\omega} 
    \left( g_i(X^{t_0,x_0,\eta,\zeta(\omega,\cdot)}_T) + \int_{t_0}^{T} L(\eta, \zeta(\omega,\cdot), s)ds
    \right) d\textbf{P}(\omega)
    \right\}.
\end{equation}
Here
\begin{equation}
\small
\begin{aligned}
    & \max_i \left\{
    \hat{p}_i - \int_{\omega} 
    \left( g_i(X^{t_0,x_0,\eta,\zeta(\omega,\cdot)}_T) + \int_{t_0}^{T} L(\eta, \zeta(\omega,\cdot),s)ds
    \right) d\textbf{P}(\omega)
    \right\} \\
    = & \max_i \left\{
    \hat{p}_i - \sum_{l} \left(\int_{\omega^l} 
    \left( g_i(X^{t_1,X^{t_0,x_0,\eta,\zeta^0}_{t_1},\eta,\zeta^l(\omega,\cdot)}_T) + \int_{t_0}^{t_1} l(\eta(s), \zeta^0(\omega,\cdot)(s))ds + \int_{t_1}^{T} L(\eta, \zeta^l(\omega,\cdot),s)ds
    \right) d\textbf{P}^l(\omega^l)
    \right) \textbf{1}_{O^l}
    \right\} \\
    \leq & \sum_{l} \sup_{\eta'} \max_i \left\{
    \hat{p}_i -  \left(\int_{\omega^l} 
     g_i(X^{t_1,X^{t_0,x_0,\eta,\zeta^0}_{t_1},\eta',\zeta^l(\omega^l,\cdot)}_T) + \int_{t_0}^{t_1} L(\eta, \zeta^0,s)ds + \int_{t_1}^{T} L(\eta', \zeta^l(\omega^l,\cdot), s)ds
     d\textbf{P}^l(\omega^l)\right)
    \right\}\textbf{1}_{O^l} \\
    \leq & \sum_l \left ( V^{-*}\left(t_1,X^{t_0,x_0,\eta,\zeta^0}_{t_1},\hat{p} - \int_{t_0}^{t_1} L(\eta, \zeta^0,s)ds\right) + 2\epsilon \right) \textbf{1}_{O^l} \\
        & \text{(because $\zeta^l$ is ($2\epsilon$)-optimal for $V^{-*}$ at $(t_1, x_1, \hat{p})$ for any $x_1 \in E_l$.)}\\
        = & V^{-*}\left(t_1,X^{t_0,x_0,\eta,\zeta^0}_{t_1},\hat{p} - \int_{t_0}^{t_1} L(\eta, \zeta^0,s)ds\right) + 2\epsilon \\
        \leq & V^{-*}_1(t_0,t_1,x_0,\hat{p}) + 3\epsilon\\
        & \text{(because $\zeta^0$ is $\epsilon$-optimal for $V^{-*}_1(t_0,t_1,x_0,\hat{p})$.)}\\
\end{aligned}
\end{equation}
Since $\epsilon$ can be arbitrarily small, we have
\begin{equation}
    V^{-*}(t_0,x_0,\hat{p}) \leq \inf_{\zeta \in \mathcal{Z}(t_0)}\sup_{\eta \in \mathcal{H}(t_0)} V^{-*}(t_1,X^{t_0,x_0,\eta,\zeta}_{T},\hat{p} - \int_{t_0}^{t_1} L(\eta, \zeta,s)ds)
\end{equation}
%%%%%%%%%%%%%%%%%%%% END BELLMAN BACKUP %%%%%%%%%%%%%%%

%%%%%%%%%%%%%%% BEGIN EXAMPLES %%%%%%%%%%%%%%%%%
\section{Examples of Zero-Sum Games with One-Sided Information}
\label{sec:example}
Here we discuss two games in detail, namely, the zero-sum beer-quiche game which is extensive-form with sequential actions, and Hexner's game~\cite{hexner1979differential} which is differential and with simultaneous actions. Both games have one-sided information and all analytically solved. We show that the characterization of value proposed by Cardaliaguet~\cite{cardaliaguet2007differential} leads to the true equilibrium behavioral strategies for both games.  

\subsection{Zero-sum beer-quiche game}
\label{sec:beerquiche}
We present a zero-sum variant of the classic beer-quiche game~\footnote{For more information about the original beer-quiche game, please see \hyperlink{https://gametheory101.com/courses/game-theory-101/the-beer-quiche-game/}{https://gametheory101.com/courses/game-theory-101/the-beer-quiche-game/}}, which is an incomplete-information game with a perfect Bayesian equilibrium. 
\paragraph{Game settings.} In this sequential game, Player 1 first chooses to take either quiche (Q) or beer (B), and based on his choice, Player 2 chooses to either defer (d) or bully (b). Player 1 has a probability of $p_T$ to be tough (T) and $p_W = 1-p_T$ to be weak (W). The exact type is unknown to Player 2 but $p = [p_T, p_W]^T$ is common knowledge. The payoffs to be maximized by Player 1 follow Table~\ref{tab:payoff}. For example, if Player 1 is tough and chooses to eat quiche (Q) while Player 2 chooses to bully (b), then Player 1 receives a payoff of 1.

\begin{table}[!ht]
    \caption{Payoff table for a zero-sum beer-quiche game}
    \label{tab:payoff}
    \begin{minipage}{.55\linewidth}
    \caption*{Tough}
    \centering
\begin{tabular}{ccc}
                       & b                      & d                      \\ \cline{2-3} 
\multicolumn{1}{l|}{B} & \multicolumn{1}{l|}{2} & \multicolumn{1}{l|}{1} \\ \cline{2-3} 
\multicolumn{1}{l|}{Q} & \multicolumn{1}{l|}{1} & \multicolumn{1}{l|}{0} \\ \cline{2-3} 
\end{tabular}
    \end{minipage}%
    \begin{minipage}{.17\linewidth}
    \caption*{Weak}
\begin{tabular}{ccc}
                       & b                       & d                      \\ \cline{2-3} 
\multicolumn{1}{l|}{B} & \multicolumn{1}{l|}{-2} & \multicolumn{1}{l|}{0} \\ \cline{2-3} 
\multicolumn{1}{l|}{Q} & \multicolumn{1}{l|}{-1} & \multicolumn{1}{l|}{2} \\ \cline{2-3} 
\end{tabular}
    \end{minipage}
\end{table}

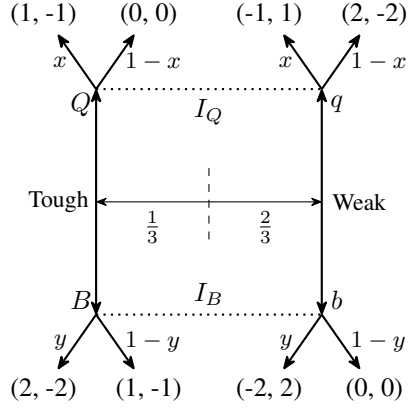
\begin{figure}
    \centering
    % \begin{tikzpicture}
% \draw [<->, thick, -{Stealth[round]}] (-1, -1) -- (-1, 1); 
% \draw [<->, thick, -{Stealth[round]}] (1, -1) -- (1, 1);
% \draw [dotted, thick] (-1, 1) -- (1, 1);
% \draw [dotted, thick] (-1, -1) -- (1, -1);
% % top left
% \draw [->, thick, -{Stealth[round]}] (-1, 1) -- (-1.3, 1.4);
% \draw [->, thick, -{Stealth[round]}] (-1, 1) -- (-0.7, 1.4);
% % top right 
% \draw [->, thick, -{Stealth[round]}] (1, 1) -- (1.3, 1.4);
% \draw [->, thick, -{Stealth[round]}] (1, 1) -- (0.7, 1.4);
% % bottom left
% \draw [->, thick, -{Stealth[round]}] (-1, -1) -- (-1.3, -1.4);
% \draw [->, thick, -{Stealth[round]}] (-1, -1) -- (-0.7, -1.4);
% % bottom right 
% \draw [->, thick, -{Stealth[round]}] (1, -1) -- (1.3, -1.4);
% \draw [->, thick, -{Stealth[round]}] (1, -1) -- (0.7, -1.4);
% \end{tikzpicture}

\begin{tikzpicture}[>={Stealth[round]}]
    \node (Q) at (-1.2-0.5, 1-0.2 + 0.5) {$Q$};
    % \node (Quiche) at (-1.7, 1.1) {\small{Quiche}};
    % \node (Beer) at (-1.7, -1.1) {\small{Beer}};
    \node (q) at (1.2+0.5, 1-0.2+0.5) {$q$};
    \node (B) at (-1.2-0.5, -1+0.2-0.5) {$B$};
    \node (b) at (1.2+0.5, -1+0.2-0.5) {$b$};
    \node [coordinate](QB) at (-1-0.5, 0) {};
    \node [coordinate](qb) at (1+0.5, 0) {};
    \path[->] (0, 0) edge node[below] {$\frac{1}{3}$}  (QB);
    \path[->] (0, 0) edge node[below] {$\frac{2}{3}$}  (qb);
    \draw [<->, thick] (-1-0.5, -1-0.5) -- (-1-0.5, 1+0.5); 
    \draw [<->, thick] (1+0.5, -1-0.5) -- (1+0.5, 1+0.5);
    \path [dotted, thick] (-1-0.5, 1+0.5) edge node[below] {$I_Q$} (1.5, 1.5);
    \path [dotted, thick] (-1.5, -1.5) edge node[above] {$I_B$} (1.5, -1.5);
    \node (T) at (-2, 0) {\small{Tough}};
    \node (W) at (2, 0) {\small{Weak}};
    \node (p1) at (-2.2, 2.5) {(1, -1)};
    \node (p2) at (-0.8, 2.5) {(0, 0)};
    \path [->, thick] (-1.5, 1.5) edge node[below, left] {\small{$x$}} (p1);
    \path [->, thick] (-1.5, 1.5) edge node[below, right] {\small{$1-x$}} (p2);

    \node (p3) at (0.8, 2.5) {(-1, 1)};
    \node (p4) at (2.2, 2.5) {(2, -2)};
    \path [->, thick] (1.5, 1.5) edge node[below, left] {\small{$x$}} (p3);
    \path [->, thick] (1.5, 1.5) edge node[below, right] {\small{$1-x$}} (p4);

    \node (p5) at (-2.2, -2.5) {(2, -2)};
    \node (p6) at (-0.8, -2.5) {(1, -1)};
    \path [->, thick] (-1.5, -1.5) edge node[below, left] {\small{$y$}} (p5);
    \path [->, thick] (-1.5, -1.5) edge node[below, right] {\small{$1-y$}} (p6);

    \node (p7) at (0.8, -2.5) {(-2, 2)};
    \node (p8) at (2.2, -2.5) {(0, 0)};
    \path [->, thick] (1.5, -1.5) edge node[below, left] {\small{$y$}} (p7);
    \path [->, thick] (1.5, -1.5) edge node[below, right] {\small{$1-y$}} (p8);
    % \path [->, thick] (-1, 1) edge node[below left] {$x$} (-1.4, 1.5);
    % \draw [->, thick, -{Stealth[round]}] (-1, 1) -- (-0.7, 1.4);
    % % top right 
    % \draw [->, thick, -{Stealth[round]}] (1, 1) -- (1.3, 1.4);
    % \draw [->, thick, -{Stealth[round]}] (1, 1) -- (0.7, 1.4);
    % % bottom left
    % \draw [->, thick, -{Stealth[round]}] (-1, -1) -- (-1.3, -1.4);
    % \draw [->, thick, -{Stealth[round]}] (-1, -1) -- (-0.7, -1.4);
    % % bottom right 
    % \draw [->, thick, -{Stealth[round]}] (1, -1) -- (1.3, -1.4);
    % \draw [->, thick, -{Stealth[round]}] (1, -1) -- (0.7, -1.4);

    \draw [dashed] (0, -0.5) -- (0, 0.5);

\end{tikzpicture}
    \caption{Zero-Sum Variant of the Beer-Quiche Game}
    \label{fig:beerquiche}
\end{figure}

 % ***Mukesh***
\paragraph{Perfect Bayesian equilibrium.} The standard approach finds the behavioral strategies of both players for a particular $p$. Consider the extensive form of the game as shown in Fig.~\ref{fig:beerquiche}. Dotted lines represent info sets that Player 2 cannot distinguish. Here, Player 1 has $p_T = \frac{1}{3}$ to be Tough. The behavioral strategies for each player are derived as follows: 

Let $Q, B, q, b$ represent probabilities of Player 1 choosing quiche given he is tough, beer given he is tough, quiche given he is weak, and beer given he is weak, respectively. Assume $x$ and $y$ be the probability of Player 2 bullying Player 1 who chooses quiche and beer, respectively. First, we find the beliefs of Player 2 when Player 1 chooses quiche or beer (info-set $\mathcal{I_Q}$ and $\mathcal{I_B}$, respectively):
\begin{align*}
    \text{if } (Q, q) \neq (0, 0), \quad \mu_2(T|\mathcal{I_Q}) &= \frac{\frac{1}{3} Q}{\frac{1}{3} (Q) + \frac{2}{3} (q)} = \frac{Q}{Q+2q} \quad \text{and}\\
    \mu_2(W|\mathcal{I_Q}) &= \frac{\frac{2}{3}q}{\frac{1}{3}Q + \frac{2}{3}q} = \frac{2q}{Q+2q}\\
    \text{if } (B, b) \neq (0, 0), \quad \mu_2(T|\mathcal{I_B}) &= \frac{B}{B+2b}, \quad \text{and} \quad \mu_2(W|\mathcal{I_B}) = \frac{2b}{B+2b}
\end{align*}
Then, the expected payoffs for bully and defer at $\mathcal{I_Q}$ are:
\begin{align*}
    E_2(\text{bully}|\mathcal{I_Q}) &= \frac{Q}{Q+2q}(-1) + \frac{2q}{Q+2q}(1) = -\frac{Q-2q}{Q+2q}\\
    E_2(\text{defer}|\mathcal{I_Q}) &= \frac{Q}{Q+2q}(0) + \frac{2q}{Q+2q}(-2) = -\frac{4q}{Q+2q}\\
\end{align*}
Given Player 2's strategy at $\mathcal{I_Q}$, his expected payoff can be expressed as:
\begin{align*}
    E_2(\mathcal{I_Q}) &= -\frac{Q-2q}{Q+2q}x - (1-x) \frac{4q}{Q+2q}\\
                       &= \frac{-(Q-6q)x - 4q}{Q+2q} %\zhe{+4q?}
\end{align*}
The value of $x$ that maximizes $E_2(\mathcal{I_Q})$ is:
\begin{align*}
    x &= \begin{cases} \text{any} & \text{if} \quad (Q, q) = (0, 0) \\ %\zhe{(0, 0)?}\\
                       1 & \text{if} \quad Q < 6q\\
                       \text{any} & \text{if} \quad Q = 6q\\
                       0 & \text{if} \quad Q > 6q\end{cases}
\end{align*}

Applying the same reasoning to info-set $\mathcal{I_B}$, we find the value of $y$ that maximizes $E_2(\mathcal{I_B})$ as: 
\begin{align*}
    y &= \begin{cases} \text{any} & \text{if} \quad (B, b) = (0, 0) \\ %\zhe{(0, 0)?}\\
                       0 & \text{if} \quad 4b < B\\ 
                       \text{any} & \text{if} \quad 4b = B\\
                       1 & \text{if} \quad 4b > B\end{cases}
\end{align*}
Given Player 2's strategy, the expected payoffs to Player 1 for his strategies are:
$$E_1(Q) = x, \quad E_1(B) = y+1, \quad E_1(q) = 2 - 3x, \quad E_1(b) = -2y$$ 

As a result the expected value for each of P1's type are:
$$E_1(T) = B(1+y) + (1-B)x$$
$$E_1(W) = b(-2y) + (1-b)(2-3x)$$
Assume $B \ge 4b$. Then, 
\begin{align*}
    1-Q = B &\ge 4b = 4(1-q) \\
    \implies 1-Q &\ge 4-4q \\
    \implies 4q &\ge Q + 3 \\
    \implies 6q &> Q
\end{align*}
Hence, $x = 1$. Thus, $E_1(Q) = 1 < E_1(B) = y+1$. As a result, $B=1$, and $Q=0$. 

Assuming $B > 4b$, following the process as above, we reach to a contradiction. Therefore, 
\begin{align*}
    B &= 4b \implies b = \frac{B}{4} = \frac{1}{4}\\
\end{align*}
Then, 
% $$E_1(q) = -1$$
$$\frac{\partial E_1(W)}{\partial b} = 1 - 2y$$
Hence, for $b= 1/4$ to be feasible, we need: $$y = \frac{1}{2}$$
% $$E_1(b) = -2y  = -1 \implies y = \frac{1}{2}$$

% Then, $1-Q = B \ge 4b = 4(1-q)$. So, $1-Q \ge 4-4q \implies 4q \ge Q + 3 \implies 6q > Q$. Hence, $x = 1$. Thus, $E_1(Q) = 1 < E_1(B) = y+1$. As a result, $B=1$, and $Q=0$. We can easily check that $B > 4b$ will lead to a contradiction. Hence, $B = 4b$ is feasible. This implies $b = \frac{B}{4} = \frac{1}{4}$. And, $E_1(q) = -1 \implies E_1(b) = -2y = -1 \implies y = \frac{1}{2}$.

Therefore, we find an equilibrium with $$x = 1, \;y = \frac{1}{2},\; q = \frac{3}{4},\; b = \frac{1}{4},  \text{ and } B = 1, Q=0$$

To summarize, Player 2 always bullies the person who eats quiche and bullies the person drinking beer half the time. The tough guy always drinks beer while the weak guy drinks beer a quarter of the time and eats quiche three-quarters of the time. One can easily check that the $B < 4b$ case also leads to a contradiction, resulting in a unique equilibrium for the game. 

\paragraph{Solution using primal and dual backward induction.} Now we solve the game through backward induction of its primal and dual values (denoted by $V$ and $C$ respectively). Here we introduce discrete-time $t=0,1,2$: Players 1 and 2 make their respective decisions at $t=0$ and $t=1$, and the game ends at $t=2$. We describe the states of the game as the decisions being made up to the corresponding time, e.g., $x=(B,b)$ at $t=2$ means that Player 1 has chosen beer and Player 2 to defer. \textbf{Primal game:} At the terminal time step ($t=2$), based on the payoff table, we have
\begin{equation}
    V(2, x, p) = \left \{ \begin{array}{ll}
        4p_T-2 & \text{if } x = (B, b) \\
        p_T & \text{if } x = (B, d) \\
        2p_T-1 & \text{if } x = (Q, b) \\
        2-2p_T & \text{if } x = (Q, d)
    \end{array}
    \right..
\end{equation}
At the intermediate time step ($t=1$), we have
\begin{equation}
    V(1, x, p) = \min_{v \in \{b, d\}} V(2, (x, v), p).
\end{equation}
We can find the best responses of Player 2 for both actions of Player 1. This leads to
\begin{equation}
    V(1, x, p) =  \left \{ \begin{array}{lll} 
        
        p_T & \text{if } x = B, ~3p_T-2 \geq 0 & (v^* = d) \\
        4p_T-2 & \text{if } x = B, ~3p_T-2 < 0 & (v^* = b) \\
        2-2p_T & \text{if } x = Q, ~4p_T-3 \geq 0 & (v^* = d) \\
        2p_T-1 & \text{if } x = Q, ~4p_T-3 < 0 & (v^* = b)
    \end{array}
    \right..
\end{equation}
Note that since Player 1 does not take an action in this time step, we do not need to take a concave hull of $V(1, x, \cdot)$. 
At the beginning of the game ($t=0$), we have
\begin{equation}
    V(0, \emptyset, p) = \text{Cav}\left( \max_{u \in \{B, Q\}} V(1, u, p)\right).
\end{equation}

\begin{figure}
    \centering
    \includegraphics[width=0.7\linewidth]{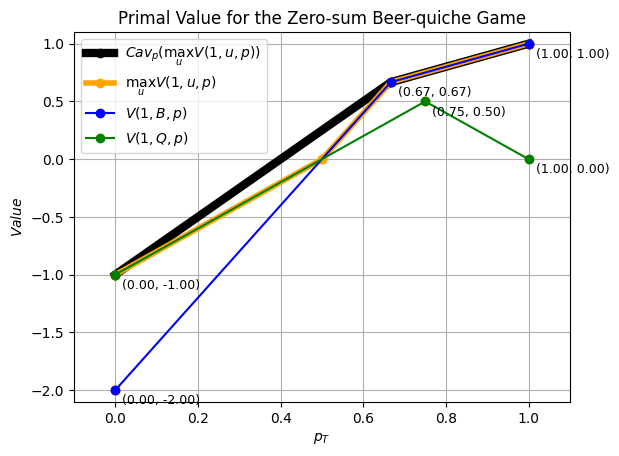}
    \caption{Primal value $V(0,\emptyset,p_T)$ at $t=0$.}
    \label{fig:beer1}
\end{figure}

By taking the concave hull with respect to $p_T$ (see Fig.~\ref{fig:beer1}), we get
\begin{equation}
    V(0, \emptyset, p) =  \left \{ \begin{array}{lll} 
        5p_T/2  - 1 & \text{if } p_T < 2/3 \\
        p_T & \text{if } p_T \geq 2/3 \\
    \end{array}
    \right..
\end{equation}
Note that from Fig.~\ref{fig:beer1}, when $p_T \in [0, 2/3)$, $V(0, \emptyset, p) = \lambda \max_u V(1, u, p^1) + (1-\lambda) \max_u V(1, u, p^2)$, where $p^1 = [0, 1]^T$, $p^2 = [2/3, 1/3]^T$, and $\lambda p^1 + (1-\lambda) p^2 = p$. When $p_T = 1/3$, $\lambda = 1/2$, Player 1's strategy is thus
\begin{equation}
\begin{aligned}
    & \Pr(u=Q|T) = \frac{\lambda p^1[1]}{p[1]} = 0,     & \Pr(u=Q|W) = \frac{\lambda p^1[2]}{p[2]} = 3/4, \\
    & \Pr(u=B|T) = \frac{(1-\lambda) p^2[1]}{p[1]} = 1, 
    & \Pr(u=B|W) = \frac{(1-\lambda) p^2[2]}{p[2]} = 1/4.
\end{aligned}
\end{equation}
This result is consistent with the true perfect Bayesian equilibrium we previously derived. 

\textbf{Dual game:} To solve for Player 2's equilibrium, we first derive the dual variable $\hat{p} \in \partial_{p} V(0, \emptyset, p)$ for $p = [1/3, 2/3]^T$. By definition, $\hat{p}^T p$ defines the concave hull of $V(0, \emptyset, p)$, and therefore we have 
\begin{equation}
\begin{aligned}
    & [1/3, 2/3] \hat{p} = V(0, \emptyset, p) = -1/6 \\
    & [0, 1] \hat{p} = V(0, \emptyset, [0, 1]) = -1.
\end{aligned}
\end{equation}
This leads to $\hat{p} = [3/2, -1]^T$.

At the terminal time, we have
\begin{equation}
    \begin{aligned}
        C(2, x, \hat{p}) & = \min\{\hat{p}[1] - g_T(x), \hat{p}[2] - g_W(x)\} \\
        & = \left \{ \begin{array}{ll}
            \min\{\hat{p}[1] - 2, \hat{p}[2] + 2\} & \text{if } x = (B, b) \\
            \min\{\hat{p}[1] - 1, \hat{p}[2]\} & \text{if } x = (B, d) \\
            \min\{\hat{p}[1] - 1, \hat{p}[2] + 1\} & \text{if } x = (Q, b) \\
            \min\{\hat{p}[1], \hat{p}[2] - 2\} & \text{if } x = (Q, d) \\
        \end{array} \right.
    \end{aligned}
\end{equation}

At $t=1$, we have
\begin{equation}
        C(1, u, \hat{p}) = \text{Cav}_{\hat{p}}\left ( \max_{v} C(2, (u, v), \hat{p})\right).
\end{equation}
When $u = B$, the conjugate value is a concave hull of a piece-wise linear function:
\begin{equation}
\begin{aligned}
        C(1, B, \hat{p}) & = \text{Cav}_{\hat{p}}\left (
        \left \{ \begin{array}{lll}
            \hat{p}[1] - 1 & \text{if } \hat{p}[2] \geq \hat{p}[1] - 1 & (v^* = d) \\
            \hat{p}[2] & \text{if } \hat{p}[2] \in [\hat{p}[1] - 2, \hat{p}[1] - 1) & (v^* = b) \\
            \hat{p}[1] - 2 & \text{if } \hat{p}[2] \in [\hat{p}[1] - 4, \hat{p}[1] - 2) & (v^* = d)  \\
            \hat{p}[2] + 2 & \text{if } \hat{p}[2] < \hat{p}[1] - 4 & (v^* = b)  \\
        \end{array} \right.
        \right) \\
        & = \left \{ \begin{array}{lll}
            \hat{p}[1] - 1 & \text{if } \hat{p}[2] \geq \hat{p}[1] - 1 & (v^* = d) \\
            2/3 \hat{p}[1] + 1/3 \hat{p}[2] - 2/3 & \text{if } \hat{p}[2] \in [\hat{p}[1] - 4, \hat{p}[1] - 1) & (\text{mixed strategy})\\
            \hat{p}[2] + 2 & \text{if } \hat{p}[2] < \hat{p}[1] - 4 & (v* = b)  \\
        \end{array} \right.
\end{aligned}
\end{equation}

\begin{figure}
    \centering
    \includegraphics[width=0.7\linewidth]{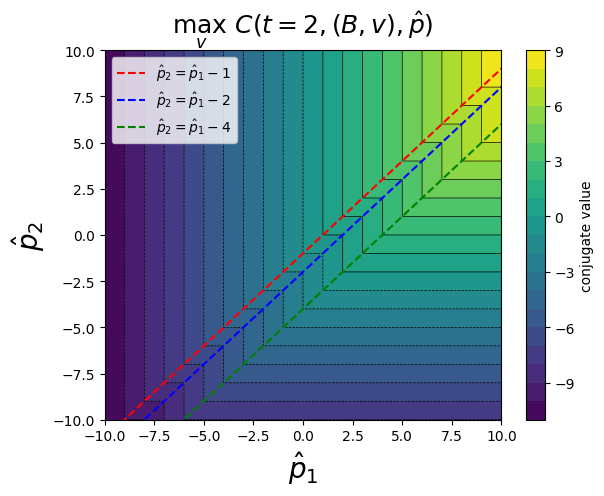}
    \caption{Conjugate value $\max_v ~C(2,B,\hat{p})$ at $t=2$.}
    \label{fig:beer2}
\end{figure}

Fig.~\ref{fig:beer2} visualizes $C(1, B, \hat{p})$. For $\hat{p} = [3/2, -1]^T$ which satisfies $\hat{p}[2] \in [\hat{p}[1] - 4, \hat{p}[1] - 1)$, Player 2 follows a mixed strategy determined based on $\{\lambda_1, \lambda_2, \lambda_3\} \in \Delta(3)$ and $\hat{p}^j \in \mathbb{R}^2$ for $j = 1, 2, 3$ such that
\begin{enumerate}[label=(\roman*)]
    \item At least one of $\hat{p}^j$ for $j = 1, 2, 3$ should satisfy $\hat{p}[2] = \hat{p}[1] - 1$ (denoted as line 1) and another $\hat{p}[2] = \hat{p}[1] - 4$ (denoted as line 2). The last could be on either line 1 or 2. These conditions are necessary for $C(1, B, \hat{p})$ to be a concave hull:
    \begin{equation}
            C(1, B, \hat{p}) = \sum_{j=1}^3 \lambda_j \max_v C(2, (B,v), \hat{p}^j).
    \end{equation}    
    Without loss of generality, we will set $\hat{p}^1$ on line 1 and both $\hat{p}^2$ and $\hat{p}^3$ on line 2;
    \item $\sum_{j=1}^3 \lambda_j \hat{p}^j = \hat{p}$.
\end{enumerate}
These conditions leads to $\lambda_1 = 1/2$ and $\lambda_2 + \lambda_3 = 1/2$. Therefore Player 2 chooses to defer and bully with equal chance when Player 1 takes beer. 

When $u = Q$, we similarly have
\begin{equation}
    C(1, Q, \hat{p}) = \left \{ \begin{array}{lll}
            \hat{p}[1] & \text{if } \hat{p}[2] \geq \hat{p}[1] + 2 & (v^* = d) \\
            ... & \text{if } \hat{p}[2] \in [\hat{p}[1] - 2, \hat{p}[1] + 2) & (\text{mixed strategy})\\
            \hat{p}[2] + 1 & \text{if } \hat{p}[2] < \hat{p}[1] - 2 & (v* = b)  \\
        \end{array} \right.
\end{equation}
We omitted the derivation of the concave hull when $\hat{p}[2] \in [\hat{p}[1] - 2, \hat{p}[1] + 2)$ because for $\hat{p} = [3/2, -1]^T$, $C(1, Q, \hat{p}) = \hat{p}[2] + 1 = 0$ while $v^* = b$, i.e. if Player 1 takes quiche, Player 2 chooses to bully with certainty.

The value and its conjugate provide behavioral strategies for Player 1 (informed) and Player 2 (non-informed), respectively, for arbitrary initial belief $p$. Moreover, the convexity of the value reveals subsets of $p$ where Player 1 should use a mixed strategy that manipulates the belief in order to improve its value. Similarly, the convexity of the conjugate value reveals subsets of dual variables $\hat{p}$ where Player 2 should use a mixed strategy to mitigate risks due to its uncertainty about Player 1.

\subsection{Hexner's game}
\label{sec:hexner}

Here we discuss the solution to Hexner's game using Cardaliaguet's method based on the reformulation proposed by Hexner. To recap, the payoff to be minimized by Player 1 is
\begin{equation}
    J(t,\tilde{\theta}_1, \tilde{\theta}_2) = \mathbb{E}_{\theta}\left[\int_{\tau = t}^T (\tilde{\theta}_1(\tau) - \theta)^2 d_1(\tau) - (\tilde{\theta}_2(\tau) - \theta)^2 d_2(\tau) d\tau\right],
\end{equation}
where $d_1$, $d_2$, $p_{\theta}$ are common knowledge; $\theta$ is only known to Player 1; the scalar $\tilde{\theta}_1$ (resp. $\tilde{\theta}_2$) is Player 1's (resp. Player 2's) strategy. We consider two player types $\theta \in \{-1, 1\}$ and therefore $p_{\theta} \in \Delta (2)$. Since the reformulation contains no system state, the strategies are functions of only time. Hexner's solution is as follows: 
\begin{align}
    \tilde{\theta}_1(s) = \tilde{\theta}_2(s) = 0 \quad \forall s\in [0, t_r] \\  
    \tilde{\theta}_1(s) = \tilde{\theta}_2(s) = \theta \quad \forall s\in (t_r, T],
\end{align}
where 
\begin{equation}\label{eq:tr}
    t_r = \argmin_t \int_0^t (d_1(s)-d_2(s))ds,
\end{equation}
and $(d_1,d_2)$ are defined in Eq.~\eqref{eq:d}. 

We will need the following preparation before introducing Cardaliaguet's solution. First, introduce time stamps $[T_k]_{k=1}^{2r}$ as roots of the time-dependent function $d_1-d_2$, with $T_0 = 0$, $T_{2q+1} = t_r$, and $T_{2r+1} = T$. Without loss of generality, we assume that:
\begin{align}
    d_1 - d_2 < 0 \quad \forall t\in (T_{2k}, T_{2k+1})~\forall k=0,...,r, \\
    d_1 - d_2 \geq 0 \quad \forall t\in [T_{2k-1}, T_{2k}]~\forall k=1,...,r.
\end{align}
We also introduce $D_k := \int_{T_k}^{T_{k+1}} (d_1-d_2)ds$ and  
\begin{equation}
    \tilde{D}_k = \left \{ \begin{array}{ll}
        \tilde{D}_{k+1} + D_k & \text{if } \tilde{D}_{k+1} + D_k < 0 \\
        0 & \text{otherwise}
    \end{array}\right.,
\end{equation}
with $\tilde{D}_{2r+1} = 0$. 

\begin{lemma}
    (Properties of $D_k$ and $\tilde{D}_k$) The following properties will be useful: 
\begin{enumerate}
    \item $\int_{k}^{2q+1} (d_1-d_2)ds = \sum_{k}^{2q} D_k < 0$, $\forall k=0,...,2q$;
    \item $\int_{2q+1}^k (d_1-d_2)ds = \sum_{2q+1}^{k-1} D_k > 0$, $\forall k=2q+2,...,2r+1$;
    \item $\tilde{D}_{2q+2} + D_{2q+1} > 0$;
    \item $\tilde{D}_k < 0, ~ \forall k<2q+1$.
\end{enumerate}
\end{lemma}

\begin{proof}
Properties 1 and 2 are results directly from the definition of $D_k$. 

For property 3, if $\tilde{D}_{2q+2} + D_{2q+1} \leq 0$, then $\tilde{D}_{2q+2} = \tilde{D}_{2q+3} + D_{2q+2} \leq -D_{2q+1}$, then $ \tilde{D}_{2q+3} \leq -(D_{2q+2}+D_{2q+1}) < 0$ (property 2). This leads to $\tilde{D}_{2q+k} \leq -\sum_{i=1}^{k-1}D_{2q+i} <0$ for $k=1,...,2r-2q$. Thus $\tilde{D}_{2r} < 0$. Contradiction.

For property 4, first we have $\tilde{D}_{2q+1}=0$ (property 3). Since $D_{2q}<0$ (property 1), $\tilde{D}_{2q} = D_{2q} < 0$.
\end{proof}

\paragraph{Primal game.} We start with $V(T,p) = 0$ where we use $p := p_{\theta}[1]$ as the probability of $\theta = -1$. The Hamiltonian can be derived as
\begin{align*}
    H(p) & = \min_{\tilde{\theta}_1} \max_{\tilde{\theta}_2} \mathbb{E}_{\theta}\left[ (\tilde{\theta}_1 - \theta)^2 d_1 - (\tilde{\theta}_2 - \theta)^2 d_2\right] \\
    & = 4p(1-p)(d_1-d_2).
\end{align*}
The optimal actions for the Hamiltonian are $\tilde{\theta}_1 = \tilde{\theta}_2 = 1-2p$.
From Bellman backup, we can get
\begin{equation*}
    V(T_k, p) = 4p(1-p)\tilde{D}_k.
\end{equation*}
Therefore, at $T_{2q+1}$, we have
\begin{align*}
    V(T_{2q+1}, p) & = Vex_p\left(V(T_{2q+2},p) + 4p(1-p)D_{2q+1}\right) \\
    & = Vex_p\left(4p(1-p)(\tilde{D}_{2q+2} + D_{2q+1})\right).
\end{align*}
Notice that $\tilde{D}_{2q+2} + D_{2q+1} >0$ (property 3) and $\tilde{D}_k <0$ for all $k < 2q+1$ (property 4), $T_{2q+1}$ is the first time such that the right-hand side term inside the convexification operator, i.e., $4p(1-p)(\tilde{D}_{2q+2} + D_{2q+1})$, becomes concave. Therefore, splitting of belief happens at $T_{2q+1}$ with $p^1 = 0$ and $p^2 = 1$. Player 1 plays $\tilde{\theta}_1 = -1$ (resp. $\tilde{\theta}_1 = 1$) with probability 1 if $\theta = -1$ (resp. $\theta = 1$), i.e., Player 1 reveals its type. This result is consistent with Hexner's. 

\paragraph{Dual game.} To find Player 2's strategy, we need to derive the conjugate value which follows
\begin{align*}
    C(t, \hat{p}) & = \left \{ \begin{array}{ll}
       \max_{i\in\{1,2\}} \hat{p}[i] & \forall t \geq T_{2q+1} \\
        \hat{p}[2] - \tilde{D}_t\left(1 - \frac{\hat{p}[1]-\hat{p}[2]}{4\tilde{D}_t}\right)^2 & \forall t < T_{2q+1}, ~4\tilde{D}_t \leq \hat{p}[1] - \hat{p}[2] \leq -4\tilde{D}_t \\
        \hat{p}[1] & \forall t < T_{2q+1}, ~\hat{p}[1] - \hat{p}[2] \geq 4\tilde{D}_t \\
        \hat{p}[2] & \forall t < T_{2q+1}, ~\hat{p}[1] - \hat{p}[2] < 4\tilde{D}_t
    \end{array}\right.
\end{align*}

Here $\hat{p} \in \nabla_{p_\theta} V(0,p_{\theta})$ and $V(0,p_{\theta}) = 4p[1]p[2]\tilde{D}_0$. For any particular $p_* \in \Delta(2)$, from the definition of subgradient, we have $\hat{p}[1] p_*[1] + \hat{p}[2] p_*[2] = 4p_*[1]p_*[2]\tilde{D}_0$ and $\hat{p}[1] - \hat{p}[2] =4(p_*[2]-p_*[1])\tilde{D}_0$. Solving these to get $\hat{p} = [4p_*[2]^2\tilde{D}_0, 4p_*[1]^2\tilde{D}_0]^T$. Therefore $\hat{p}[1] - \hat{p}[2] = 4\tilde{D}_0 (1-2p_*[1]) \in [4\tilde{D}_0, -4\tilde{D}_0]$, and
\begin{equation*}
    C(0,\hat{p}) = \hat{p}[2] - \tilde{D}_0\left(1 - \frac{\hat{p}[1]-\hat{p}[2]}{4\tilde{D}_0}\right)^2.
\end{equation*}
% , $C(0,\hat{p})$ is convex to $\hat{p}$ and therefore no splitting of $\hat{p}$.

Notice that $C(t, \hat{p})$ is convex to $\hat{p}$ since $\tilde{D}_0 <0$ (property 4) for all $t \in [0, T]$. Therefore, there is no splitting of $\hat{p}$ during the dual game, i.e., $\tilde{\theta}_2 = 1-2p$. This result is also consistent with Hexner's.

%%%%%%%%%%%%%%%% END EXAMPLES %%%%%%%%%%%%%%%%

%%%%%%%%%%%%%%%%% BEGIN BRT %%%%%%%%%%%%%
\section{Backward Reachable Tube}\label{sec:BRT}
The computation of the Backward Reachable Tube (BRT) allows us to classify the state space into feasible and infeasible regions at different times from Player 1's perspective.

\paragraph{Computation of BRT.} For the simplified football game, the state constraint is defined as $c(x) := \|(d_{x_1}, d_{y_1}) - (d_{x_2}, d_{y_2})\|_2 - r$, and $\mathcal{C} = \{x: c(x) \leq 0\}$.
The Hamilton-Jacobi-Isaacs Variational Inequality (HJI VI) is denoted by $L$ and satisfies the boundary condition $D$~\cite{bansal2021deepreach}:
\begin{equation}
    \begin{aligned}
        & L(\tilde{V}, t, x)=\min\{\nabla_t \tilde{V}(t,x) + H(t,x), c(x) - \tilde{V}(t,x)\} = 0, \\
        & D(\tilde{V}, x) = \tilde{V}(T,x) - c(x) = 0,
    \end{aligned}
\end{equation}
where $H$ is the Hamiltonian: 
\begin{equation}
    H(t,x) = \max_u \min_v \langle \nabla_x \tilde{V}(t,x), f(x, u, v)\rangle.
\end{equation}

We use Physics-Informed Neural Network (PINN) to learn the value function $\tilde{V}(t,x)$, the sub-zero level set of which represents the BRT: 
\begin{equation}
    \bar{\mathcal{Q}}(t) = \{x \in \mathbb{R}^{d_x}: \tilde{V}(t,x) \le 0\}.
\end{equation}

We denote PINN dataset $\mathcal{D} = \left \{\left(t^{(k)}, x^{(k)}\right) \right\}_{k=1}^K$ containing uniformly sampled data points in $[0, T] \times \mathbb{R}^{d_x}$ and define the loss function as:
\begin{equation}
    \min_{\tilde{V}} \quad \mathcal{L}\left(\tilde{V}\right) = \sum_{k=1}^K \left\|L(\tilde{V}^{(k)}, t^{(k)}, x^{(k)})\right\|_1 + C_1 \left\|D(\tilde{V}^{(k)}, x^{(k)})\right\|_1,
    \label{eq:deepreach}
\end{equation}
where $\tilde{V}^{(k)}$ is an abbreviation for $\tilde{V} \left(t^{(k)}, x^{(k)}\right)$ and $C_1$ is the hyperparameter that balances the loss term $\|L\|_1$ and $\|D\|_1$. 

\begin{figure}[!ht]
    \centering
    \includegraphics[width=0.45\linewidth]{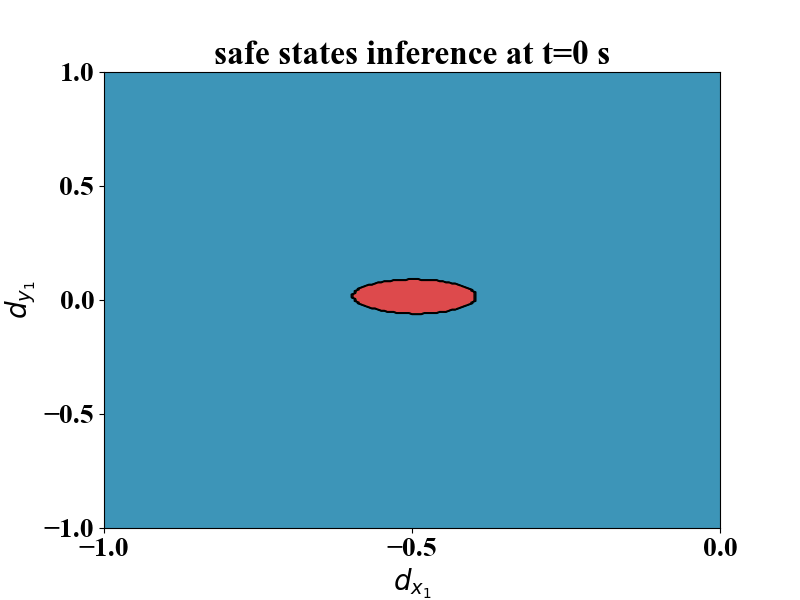}
    \caption{A visualization of safe/unsafe initial position for the attacker when the defender is fixed at $(-0.05, 0)$. The initial velocities for both players are zero in the beginning. The red (blue) region represents unsafe (safe) states.}
    \label{fig:BRT_initial}
\end{figure}

\paragraph{Training.} We uniformly sample 60k input states $x \in [-1, 1]$ (specifically, 10k states $x \in \mathcal{C}$) and use curriculum learning proposed in~\cite{bansal2021deepreach} to improve the training convergence. The rest of the dynamics parameters are chosen as: $T=1$, $r=0.05$, $u_x \in [-6, 6]$, $u_y \in [-12, 12]$, $v_x \in [-6, 6]$, $v_y \in [-4, 4]$, velocities are sampled as mentioned in Sec.~\ref{sec: data sampling} and are normalized between $[-1, 1]$. The PINN utilizes a fully-connected network with 3 hidden layers, each comprising 512 neurons with \texttt{sin} activation function. The network adopts the Adam optimizer with a fixed learning rate of $2 \times 10^{-5}$. We first pretrain the network over 10k iterations to satisfy the boundary condition $D$ and then refine the network through 100k gradient descent steps, with states sampled from an expanding time window starting from the terminal. Fig.~\ref{fig:BRT_initial} shows the visualization of BRT in a 2D frame given $t=0$ and fixed states except $(d_{x_1}, d_{y_1})$. 

%%%%%%%%%%%%%%%% END BRT %%%%%%%%%%%%%%%

%%%% BEGIN ERROR PROP %%%%%%%%%%%%%%%%%%%
\section{Proof of Proposition~\ref{prop:convex_error}}
\label{sec:convex_error}
% \begin{proposition}\label{prop:convex_error}
%     For given $(t,x)$, let the Lipschitz constant of $\vartheta(t,x,\cdot)$ be $L$, and let $d_{\mathcal{P}}$ be the minimum distance between two neighboring nodes of the lattice $\mathcal{P}$. $\varepsilon_{vex}(t,x) \leq d_{\mathcal{P}}L$.
% \end{proposition}

\begin{proof}
    Let $f^{0}: [0,1]^{I-1} \rightarrow \mathbb{R}$ be a bounded and Lipschitz continuous function, $\mathcal{P} \subset [0,1]^{I-1}$ be a lattice, and $f$ be a convex hull computed from the data $\{f(p), p\}_{p \in \mathcal{S}}$. Let the true convex hull of $f^0$ be $Vex(f^0)$: $Vex(f^0)(p) \leq f(p)$ for all $p \in [0,1]^{I-1}$, with equality reached at least for $p \in \mathcal{S}$.
    
    Introduce a set $P^0 = \{p^{(i)} \in \mathcal{S}\}_{i=1}^{I}$ and a space $\mathcal{P}^0 = \{p \in [0,1]^{I-1}~|~ \exists \lambda \in \Delta(I) \text{ s.t. } p = \sum_{i=1}^{I} \lambda[i] p^{(i)},~p^{(i)}\in P^0 \}$ so that $f(p) = \sum_{i=1}^{I} \lambda[i]f(p^{(i)})$ for all $p\in \mathcal{P}^0$, i.e., $P^0$ are vertices of a segment $\mathcal{P}^0$ of $[0,1]^{I-1}$ within which $f$ is affine.  

    Let $U := \{u^{(i)}\}_{i=1}^N = \mathcal{P} \bigcap \mathcal{P}^0$ be the set of lattice nodes contained in $\mathcal{P}^0$. Since $f$ is a convex hull of $f^0$, we have $f(u^{(i)}) \leq f^0(u^{(i)})$ for all $i=1,...,N$. $U$ defines a segmentation $\mathcal{E}$ of $\mathcal{P}^0$: Each $e \in \mathcal{E}$ is associated with $U_e := \{u^{(e_i)}\}_{i=1}^{I} \subset U$ such that $e = \{p \in [0,1]^{I-1} ~|~ \exists \lambda \in \Delta(I) \text{ s.t. } p = \sum_{i=1}^{I} \lambda[i] u^{(e_i)},~u^{(e_i)}\in U_e\}$ and $u \notin e$ for any $u \in U \setminus U_e$.

    For any $p \in e$, we have the following loose lower bound on $f^0(p)$:
    \begin{equation}
        f^0(p) \geq \min_i f^0(u^{e_i}) - \Delta L \geq \min_i f(u^{e_i}) - \Delta_e L,
    \end{equation}
    where $\Delta := \max_{i,j} \|u^{e_i}-u^{e_j}\|_2$, and $L$ is the Lipschitz constant of $f^0$. $\Delta$ is a constant for a given lattice $\mathcal{P}$.

    Therefore within $e$, the convexification error is lower bounded by
    \begin{equation}
        \max_{\lambda \in \Delta(I)} \left\{f(\sum_i \lambda[i] u^{(e_i)}) - \min_i f(u^{e_i}) + \Delta L\right\} = \max_i f(u^{e_i}) - \min_i f(u^{e_i}) + \Delta L \leq 2\Delta L.
    \end{equation}

    Since this error is constant, and $f - 2\Delta L$ is a convex lower bound of $Vex(f^0)$, we have $\varepsilon_{vex} \leq 2\Delta L$.
\end{proof}

%%%%%%%%%%%%% END ERROR PROP %%%%%%%%%%%%%%%%
%% add d1-d2 plot for hexner unconstrained game
\section{Details on Case Studies}\label{sec:case_details}
The code for the implementation is available at \href{https://github.com/ghimiremukesh/OSIIG}{\texttt{https://github.com/ghimiremukesh/OSIIG}}.
\subsection{Hexner's Strategy}
For the unconstrained simplified football game discussed in Sec.~\ref{sec:cases}, the strategies depend on the trajectory of the $d_1 - d_2$. In Fig.\ref{fig:d1d2}, we plot the trajectory and determine the critical time from Eq.\eqref{eq:tr}. For the choices of parameters, we determine $t_r = 0.4$s. We set $R_A = \texttt{diag}(0.05, 0.025)$, and $R_D = \texttt{diag}(0.05, 0.1)$.

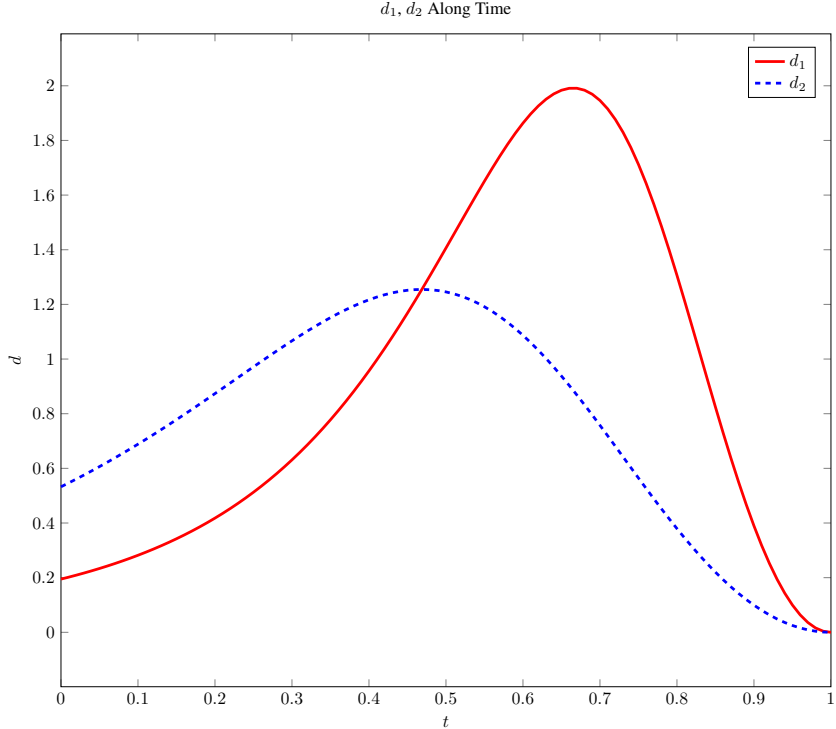
\begin{figure}[!h]
    \centering
    \begingroup%
    \StrCount{figures/d1_d2}{/}[\matches]%
    \StrBefore[\matches]{figures/d1_d2}{/}[\datapath]%
    \pgfplotstableread[col sep=comma,]{\datapath/d1-d2_highY.csv}\datatablea
\begin{tikzpicture}[]
\begin{axis}[title={$d_1$, $d_2$ Along Time},
    legend entries={$d_1$, $d_2$},
    xlabel={$t$},
    ylabel={$d$},
    xmin=0,
    xmax=1,
]
\addplot [ultra thick, red] table [x expr=0.01 * \coordindex, y=d1]{\datatablea};
\addplot [ultra thick, blue, dashed] table [x expr=0.01 * \coordindex, y=d2]{\datatablea};
\end{axis}
\end{tikzpicture}%
    \endgroup%

    \caption{Plot of $d_1$ and $d_2$ along time. The critical time $t_r$ occurs at $t\approx 0.4$. The attacker will conceal its type until $t_r$ and reveal it after $t_r$. }
    \label{fig:d1d2}
\end{figure}

\subsection{Data Sampling}
\label{sec: data sampling}
\paragraph{Unconstrained Game.}
For the unconstrained game, we sample positions $(d_x, d_y)$ and velocities $(\dot{d}_x, \dot{d}_y)$ for both players. As the arena is bounded between $[-1, 1]$ in both $x$ and $y$ directions, we sample the positions of the two players in $[-1, 1]$. However, when it comes to velocities, we experimentally determine the range from the LQR problem as the following: 
$\dot{d}_{x_1} \in [-6, 6]$, $\dot{d}_{y_1} \in [-4, 4]$, $\dot{d}_{x_2} \in [-6, 6]$, and $\dot{d}_{y_2} \in [-4, 4]$. 
We then normalize the velocities between $[-1, 1]$ and compute the values as described in algorithm \ref{alg:value}. The resulting normalized joint states $(\mathcal{X})$ and values $(V)$ are stored for training the value network. At each time step we sample 10000 states and set $|\mathcal{P}| = 100$. This brings the total training data at each time step to 1M for the unconstrained case.
\paragraph{Constrained Primal Game.}
For the constrained game, we sample the positions between $[-1, 1]$ and all velocities between the ranges discussed above. As in the unconstrained case, these are normalized to $[-1, 1]$ before computing the values and storing the training data. With the same $\mathcal{P}$, we sample 5000 states from the feasible set $\mathcal{Q}(t)$, resulting in 500,000 training data at each time step. Solving constrained game requires evaluating $\min_u \max_v V(t, x+\tau f(x, u, v), p)$ over all possible pairs of $x'(\text{i.e. } x+\tau f(x, u, v))$, which is memory intensive. Based on the available resources we set the total number of initial states to be sampled to 5000. To speed up the calculation, and capture a wide range of data, we divide the state space into 50 uniform intervals, and distribute the computation to 56 CPU cores, with 515,271 MB of total memory. Each minimax computation is independent and hence can be evaluted in parallel.   

\paragraph{Constrained Dual Game.} In the dual game, the uninformed player keeps track of the process $\hat{p} \in \mathbb{R}^I$. As a result, the dual value is a 10-D function, which increases the complexity of the computation due to the need for convexification of the value along $I$ dimensions (here, $I=2$). We follow the same procedure as in the primal game and collect 250,000 samples for training. The range of $\hat{p}$ was determined to be $[-14, 14]$ based on the primal value at the initial time as discussed in Sec.~\ref{sec:value}. Furthermore, due to the additional input dimension in the dual value network, the dual value approximation suffers from relatively higher error compared to the primal value. Ultimately this affects the strategy of the uninformed player (Player 2). We compare the resulting strategy of Player 2 from the dual value with that of the ground truth strategy in the unconstrained game.

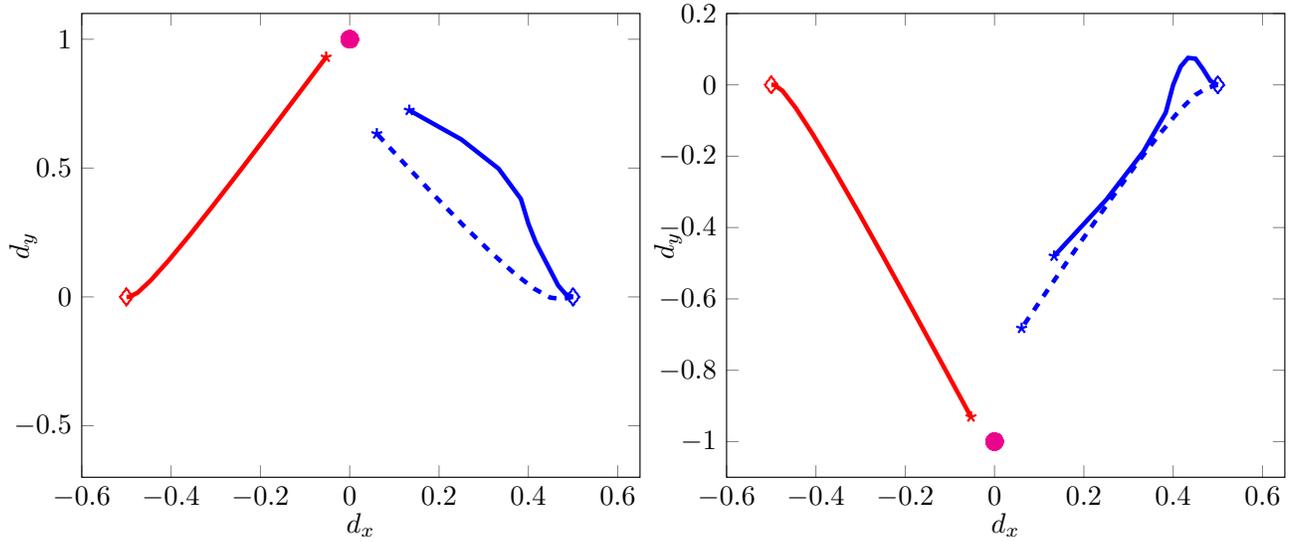
\begin{figure}[!ht]
\begin{subfigure}[b]{0.5\linewidth}
    \centering
    \begingroup%
    \StrCount{figures/compare_dual_new}{/}[\matches]%
    \StrBefore[\matches]{figures/compare_dual_new}{/}[\datapath]%
    \pgfplotstableread[col sep=comma,]{\datapath/new_results_2/primal_dual_1_0.csv}\datatablea

\pgfplotstableread[col sep=comma,]{\datapath/uncons_inc_1.csv}\datatableb
\begin{tikzpicture}[]
\begin{axis}[title={},
    xmin=-0.6,
    xmax=0.65,
    ymin=-0.7,  % -1.1
    ymax=1.1, 
    point meta=explicit,
    ylabel shift=-15pt,
    xlabel shift=-5pt,
    % legend entries={Attacker, Defender},
    % legend image post style={thick},
    xlabel={$d_x$},
    ylabel={$d_y$},
            colormap name = p1,
    % colorbar,
    point meta min=0, 
    point meta max=1,
    % colorbar horizontal,
    %     colorbar style={
    %   at={(0.75,-0.1)},
    %   anchor=center,
    %   width=4.5cm
    % },
  %   colorbar/draw/.append code={
  % \begin{axis}[
  %   % colormap/p2/.style={
  %   % colormap name=p2,
  %   % }, 
  %   colormap name=p2,
  %   colorbar horizontal,
  %   point meta min=0,
  %   point meta max=1,
  %   every colorbar,
  %     at={(0.25,-0.1)},
  %     anchor=center,
  %     width=4.5cm,
  %   colorbar shift,
  %   colorbar=false,
  %   ]
  %       \pgfkeysvalueof{/pgfplots/colorbar addplot}
  %     \end{axis}
  %   }]
]
\addplot[color=red, mark=diamond, mark size=3pt] (-0.5, 0);
\addplot[color=blue, mark=diamond, mark size=3pt] (0.5, 0);
\addplot [mark=o, mark size=0 pt, red, ultra thick] table [x={x1}, y={y1}]
\datatablea;
\addplot [mark=o, mark size=0 pt, blue, ultra thick] table [x={x2}, y={y2}]
\datatablea;
\addplot[color=magenta, mark=*, mark options={fill=magenta}, only marks, mark size=3pt](0, 1);
\addplot[color=magenta, mark=*, mark options={fill=white}, only marks, mark size=3pt](0, -1);

\pgfplotstablegetelem{10}{x1}\of{\datatablea}
\pgfmathsetmacro{\a}{\pgfplotsretval}
\pgfplotstablegetelem{10}{y1}\of{\datatablea}
\pgfmathsetmacro{\b}{\pgfplotsretval}
\addplot [red, mark=star, only marks, mark size=2pt, thick] coordinates {(\a, \b)};

\pgfplotstablegetelem{10}{x2}\of{\datatablea}
\pgfmathsetmacro{\a}{\pgfplotsretval}
\pgfplotstablegetelem{10}{y2}\of{\datatablea}
\pgfmathsetmacro{\b}{\pgfplotsretval}
\addplot [blue, mark=star, only marks, mark size=2pt, thick] coordinates {(\a, \b)};

% \addplot [mark=o, mark size=0 pt, red, ultra thick, dashed] table [x={x1}, y={y1}]
% \datatableb;
\addplot [mark=o, mark size=0 pt, blue, ultra thick, dashed] table [x={x2}, y={y2}]
\datatableb;
\pgfplotstablegetelem{10}{x2}\of{\datatableb}
\pgfmathsetmacro{\a}{\pgfplotsretval}
\pgfplotstablegetelem{10}{y2}\of{\datatableb}
\pgfmathsetmacro{\b}{\pgfplotsretval}
\addplot [blue, mark=star, only marks, mark size=2pt, thick] coordinates {(\a, \b)};
\end{axis}
\end{tikzpicture}%
    \endgroup%

    \end{subfigure}%
    \begin{subfigure}[b]{.5\linewidth}
    \centering
    \begingroup%
    \StrCount{figures/compare_dual_2_new}{/}[\matches]%
    \StrBefore[\matches]{figures/compare_dual_2_new}{/}[\datapath]%
    \pgfplotstableread[col sep=comma,]{\datapath/new_results_2/primal_dual_-1_0.csv}\datatablea

\pgfplotstableread[col sep=comma,]{\datapath/uncons_inc_2.csv}\datatableb
\begin{tikzpicture}[]
\begin{axis}[title={},
    xmin=-0.6,
    xmax=0.65,
    ymin=-1.1,  % -1.1
    ymax=0.2, 
    point meta=explicit,
    ylabel shift=-15pt,
    xlabel shift=-5pt,
    % legend entries={Attacker, Defender},
    % legend image post style={thick},
    xlabel={$d_x$},
    ylabel={$d_y$},
            colormap name = p1,
    % colorbar,
    point meta min=0, 
    point meta max=1,
    % colorbar horizontal,
    %     colorbar style={
    %   at={(0.75,-0.1)},
    %   anchor=center,
    %   width=4.5cm
    % },
  %   colorbar/draw/.append code={
  % \begin{axis}[
  %   % colormap/p2/.style={
  %   % colormap name=p2,
  %   % }, 
  %   colormap name=p2,
  %   colorbar horizontal,
  %   point meta min=0,
  %   point meta max=1,
  %   every colorbar,
  %     at={(0.25,-0.1)},
  %     anchor=center,
  %     width=4.5cm,
  %   colorbar shift,
  %   colorbar=false,
  %   ]
  %       \pgfkeysvalueof{/pgfplots/colorbar addplot}
  %     \end{axis}
  %   }]
]
\addplot[color=red, mark=diamond, mark size=3pt] (-0.5, 0);
\addplot[color=blue, mark=diamond, mark size=3pt] (0.5, 0);
\addplot [mark=o, mark size=0 pt, red, ultra thick] table [x={x1}, y={y1}]
\datatablea;
\addplot [mark=o, mark size=0 pt, blue, ultra thick] table [x={x2}, y={y2}]
\datatablea;
\addplot[color=magenta, mark=*, mark options={fill=magenta}, only marks, mark size=3pt](0, -1);
\addplot[color=magenta, mark=*, mark options={fill=white}, only marks, mark size=3pt](0, 1);

\pgfplotstablegetelem{10}{x1}\of{\datatablea}
\pgfmathsetmacro{\a}{\pgfplotsretval}
\pgfplotstablegetelem{10}{y1}\of{\datatablea}
\pgfmathsetmacro{\b}{\pgfplotsretval}
\addplot [red, mark=star, only marks, mark size=2pt, thick] coordinates {(\a, \b)};

\pgfplotstablegetelem{10}{x2}\of{\datatablea}
\pgfmathsetmacro{\a}{\pgfplotsretval}
\pgfplotstablegetelem{10}{y2}\of{\datatablea}
\pgfmathsetmacro{\b}{\pgfplotsretval}
\addplot [blue, mark=star, only marks, mark size=2pt, thick] coordinates {(\a, \b)};

% \addplot [mark=o, mark size=0 pt, red, ultra thick, dashed] table [x={x1}, y={y1}]
% \datatableb;
\addplot [mark=o, mark size=0 pt, blue, ultra thick, dashed] table [x={x2}, y={y2}]
\datatableb;
\pgfplotstablegetelem{10}{x2}\of{\datatableb}
\pgfmathsetmacro{\a}{\pgfplotsretval}
\pgfplotstablegetelem{10}{y2}\of{\datatableb}
\pgfmathsetmacro{\b}{\pgfplotsretval}
\addplot [blue, mark=star, only marks, mark size=2pt, thick] coordinates {(\a, \b)};
\end{axis}
\end{tikzpicture}%
    \endgroup%

    \end{subfigure}
    \caption{Comparison between the P2's ground truth strategy and the strategy synthesized from the dual value. P1's trajectory is shown red and P2's in blue. Solid trajectories correspond to that obtained when P2 plays its equilibrium strategy. Dotted trajectories represent the ground truth solution.}
    \label{fig:compare_dual}
\end{figure}

% Next, we visualize the slices of the BRT at different times across the relative $x$ and $y$ positions of the two players. Velocities are fixed at zero. At the final time, the unsafe states are in a ball of radius $r=0.05$, as shown in Fig.(\ref{fig:BRS_final}), which defines the state constraint in the game. 
% \begin{figure}[!ht]
%     \centering
%     \includegraphics[width=0.45\linewidth]{icml2024/figures/BRS_1.00.png}
%     \caption{A visualization of safe/unsafe states at the final time. $\Delta x$ is the relative position of the attacker to defender in the $x$-direction. All the velocities are set to zero. The unsafe states lie in the ball of radius $r=0.05$.}
%     \label{fig:BRS_final}
% \end{figure}
% \begin{figure}[!h]
%     \centering
%     \includegraphics[width=\linewidth]{icml2024/figures/BRT_combined.png}
%     \caption{Slices of BRTs at different time steps. The red region corresponds to unsafe states and the blue region corresponds to the safe states. If the state of the system lies in the unsafe region then the attacker, under no circumstance, will be able to avoid getting tackled by the defender.}
%     \label{fig:enter-label}
% \end{figure}

% Increasing the collision radius also increases the size of the unsafe sets. We visualize the effect of increasing the collision radius to $r =0.15$ in the following figures. 
\end{document}